\documentclass[aps,reprint]{revtex4-1}
\usepackage[breaklinks=true]{hyperref}
\usepackage{hyperref}
\hypersetup{
    colorlinks,
    linkcolor={red!50!black},
    citecolor={blue!50!black},
    urlcolor={blue!80!black}
}
\usepackage{blindtext}
\usepackage{amsmath,amssymb,amsthm,mathtools,mathrsfs}
\usepackage{tikz}
\usetikzlibrary{arrows.meta}
\usetikzlibrary{decorations.pathmorphing,shapes}
\usetikzlibrary{calc}
\usepackage{pifont}
\usepackage{comment}
\usepackage{adjustbox}
\usepackage{dsfont}
\usepackage[normalem]{ulem}
\usepackage{enumerate}
\usepackage[all]{xy} 

\theoremstyle{plain}
\newtheorem{theorem}{Theorem}[section]
\newtheorem{lemma}[theorem]{Lemma}
\newtheorem{proposition}[theorem]{Proposition}

\theoremstyle{definition}
\newtheorem{definition}[theorem]{Definition}
\newtheorem{remark}[theorem]{Remark}

\newcommand{\bd}{\begin{definition}}
\newcommand{\ed}{\end{definition}}

\newcommand{\cmark}{\ding{51}}
\newcommand{\xmark}{\ding{55}}

\newcommand\define[1]{\emph{\textbf{#1}}}

\newcommand{\bt}{\begin{theorem}}
\newcommand{\et}{\end{theorem}}
\newcommand{\be}{\begin{equation}}
\newcommand{\ee}{\end{equation}}

\let\C=\Chi

\newcommand{\matr}{\mathbb{M}}

\def\mA{{{\mathcal{A}}}}
\def\mB{{{\mathcal{B}}}}
\def\mC{{{\mathcal{C}}}}
\def\mE{{{\mathcal{E}}}}

\def\M{{{\mathbb{M}}}}

\def\R{{{\mathbb R}}}
\def\C{{{\mathbb C}}}

\newcommand{\id}{\mathrm{id}}

\newcommand{\tr}{{\rm tr} }

\newcommand{\St}{\mathcal{S}}
\newcommand{\CPTP}{\mathrm{CPTP}}
\newcommand{\sa}{\mathrm{sa}}

\newcommand{\Map}{\mathrm{Map}}
\newcommand{\TP}{\mathrm{TP}}

\newcommand{\HPTP}{\mathrm{HPTP}}
\newcommand{\SB}{\mathrm{J}}
\newcommand{\LS}{\mathrm{LS}}

\newcommand{\<}{\langle}
\renewcommand{\>}{\rangle}

\newcommand{\Ad}{\mathrm{Ad}}

\begin{document}
\title{From time-reversal symmetry to quantum Bayes' rules}
\author{Arthur J.~Parzygnat}
\email{parzygnat@nagoya-u.jp}
\affiliation{Graduate School of Informatics, Nagoya University, Chikusa-ku, 464-8601 Nagoya, Japan}
\author{James Fullwood}
\email{fullwood@sjtu.edu.cn}
\affiliation{School of Mathematical Sciences, Shanghai Jiao Tong University, 800 Dongchuan Road, Shanghai, China}

\begin{abstract}
Bayes' rule $\mathbb{P}(B|A)\mathbb{P}(A)=\mathbb{P}(A|B)\mathbb{P}(B)$ is one of the simplest yet most profound, ubiquitous, and far-reaching results of classical probability theory, with applications
in any field utilizing statistical inference.
Many attempts have been made to extend this rule to quantum systems, the significance of which we are only beginning to understand. 
In this work, we develop a systematic framework for defining Bayes' rule in the quantum setting, and we show that a vast majority of the proposed quantum Bayes' rules appearing in the literature are all instances of our definition. 
Moreover, our Bayes' rule is based upon a simple relationship between the notions of \emph{state over time} and a time-reversal symmetry map, both of which are introduced here.
\end{abstract}

\maketitle

\section{Introduction}

Bayes' rule is a cornerstone of inference, prediction, retrodiction, and decision-making which is used throughout the natural sciences
~\cite{Pearl88,Bala97,Be06,Fr10,Ba14,Gh15,Ja19,BuSc21,AwBuSc21,Sm21}. 
Due to its ubiquitous stature, many have proposed extensions beyond classical probability theory into the setting of quantum mechanics~\cite{Bu77,Lu06,Oz97,Ko99,Ko01,SBC01,Fu01,Wa05,WaKu10,Le06,Le07,CoSp12,Te12,FaKo12,LeSp13,CHPSSW19,PaRu19,PaBayes,PaRuBayes,GPRR21,ChMa20,Ts22,Liu22,MaCh22}, with applications to cosmology~\cite{Te12},  entanglement wedge reconstruction in AdS/CFT~\cite{CHPSSW19}, and quantum foundations~\cite{Fu01,LeSp13,BeBe22}. 

A common approach to formulating quantum generalizations of the classical Bayes' rule, namely,
\be \label{CXBR17}
\mathbb{P}(y|x)\mathbb{P}(x)=\mathbb{P}(x|y)\mathbb{P}(y),
\ee
is by ``quantizing'' Equation~\eqref{CXBR17}, i.e., by defining operator analogues of $\mathbb{P}(y|x),\ldots ,\mathbb{P}(y)$ in such a way that the substitution of such analogues into equation \eqref{CXBR17} also yields a valid equation. But since there are various approaches to formulating operator analogues of $\mathbb{P}(y|x),\ldots ,\mathbb{P}(y)$, and moreover,  
since multiplication of such operators is not necessarily commutative, different formulations of quantum Bayes' rules have appeared in the literature, each having its own advantages over the others~\cite{Bu77,Lu06,Oz97,Ko99,Ko01,SBC01,Fu01,Wa05,WaKu10,Le06,Le07,CoSp12,Te12,FaKo12,LeSp13,CHPSSW19,PaRu19,PaBayes,PaRuBayes,GPRR21,ChMa20,Ts22,Liu22,MaCh22}. 

The approach taken here however, is to view equation \eqref{CXBR17} as a reflection of a certain time-reversal symmetry for classical systems, and that the associated symmetry transformation relating both sides of equation \eqref{CXBR17} is more fundamental than the particular form of the equation. Moreover, our emphasis on transformations, as opposed to equations, reveals that various formulations of quantum Bayes' rules are all manifestations of what we refer to as a Bayes' rule \emph{with respect to a state over time}, which is a quantum analogue of a joint distribution $\mathbb{P}(x,y)$ associated with a physical system on two timelike separated regions.

To explain how such a perspective may be taken in the classical case, let $(\Omega,\mathbb{P})$ be a finite probability space, where $\Omega$ is a finite set corresponding to all possible outcomes of an experiment, or data-generating process. Given random variables
$\mathcal{X}:\Omega\to X$ and $\mathcal{Y}:\Omega\to Y$ on $\Omega$,  let $p=\mathcal{X}_*\mathbb{P}$ and $q=\mathcal{Y}_*\mathbb{P}$ be the associated probability mass functions, so that for all $x\in X$ and $y\in Y$,
\be
p(x)=\mathbb{P}\left(\mathcal{X}^{-1}(x)\right) \quad \text{and} \quad q(y)=\mathbb{P}\left(\mathcal{Y}^{-1}(y)\right).
\ee
Since we are considering not one, but two random variables $\mathcal{X}$ and $\mathcal{Y}$, there are two associated joint distributions, or rather, \emph{states over time} $\vartheta:X\times Y\to [0,1]$ and $\vartheta^*:Y\times X\to [0,1]$, depending on whether the obeservation of $\mathcal{X}$ precedes $\mathcal{Y}$ or vice-versa. In particular, $\vartheta$ is the state over time corresponding to first observing $\mathcal{X}$ and then observing $\mathcal{Y}$, which is given by
\be
\vartheta(x,y)=p(y|x)p(x),
\ee
where $p(y|x)$ represents the conditional probability of observing $y$ \emph{given} that $x$ was observed first. Similarly, $\vartheta^*$ corresponds to the time-reversal of this procedure, since it describes the probability of first observing $\mathcal{Y}$ and then observing $\mathcal{X}$, which is given by
\be
\vartheta^*(y,x)=q(x|y)q(y).
\ee

Bayes' rule~\eqref{CXBR17} in this context may then be reformulated as
\be \label{BRFX197}
\vartheta=\gamma(\vartheta^*),
\ee 
where $\gamma$ is the canonical map sending distributions on $Y\times X$ to distributions on $X\times Y$, which we view as a reversal of time for joint distributions representing timelike separated variables. Equation \eqref{BRFX197} then says the states over time $\vartheta$ and $\vartheta^*$ are related by the time-reversal transformation $\gamma$, and it is this formulation of Bayes' rule that we will take as a guide into the quantum realm.

But while states over time in the classical case are essentially unique, there are various approaches to defining states over time in the quantum setting. This stems from the fact that while the postulates of quantum mechanics make clear that a joint state supported on \emph{spacelike} separated regions of a system is represented by a density matrix on the tensor product of the Hilbert spaces for each region, the postulates are silent regarding what mathematical entity should faithfully represent \emph{timelike} separated states of a system. In particular, if two timelike separated states are causally related,  
then a suitable notion of state over time for the system should be an operator on the tensor product of the Hilbert spaces at the two different times, encoding not only the two timelike separated states, but also causal correlations between the states as well. 
Moreover, it has been emphasized in 
Refs.~\cite{FJV15,FuPa22} that such a state over time should admit negative eigenvalues if it is to encode temporal correlations, and as such, states over time should not be positive in general. 
It is because of these reasons that defining states over time is not so straightforward.

In Ref.~\cite{HHPBS17}, a minimal list of axioms was proposed that any general state over time construction should satisfy, and a no-go theorem was proved stating that there is no such construction satisfying their list of axioms. However, while the aforementioned no-go theorem is mathematically correct, in Ref.~\cite{FuPa22} the present authors found a loop-hole to this theorem by slightly weakening the hypotheses in a way which did not alter their physical significance and interpretation, thus resulting in an explicit state over time construction satisfying the axioms put forth in Ref.~\cite{HHPBS17}. While it is still not known whether or not our state over time construction is \emph{characterized} by such axioms, we are nevertheless starting to better understand states over time from both a mathematical and physical perspective, as will be further supported in this work.

In particular, 
we use Equation~\eqref{BRFX197} to formulate a general quantum Bayes' rule which takes into account the choice of a state over time construction. By doing so, we show that the various formulations of quantum Bayes' rules appearing in Refs.~\cite{Bu77,SBC01,Lu06,Oz97,Fu01,LeSp13,Te12,PaBayes,PaRuBayes,Ts22} may all be obtained from our Bayes' rule once an appropriate notion of state over time is specified in each case. We also establish a list of axioms for state over time constructions similar to those in Ref.~\cite{HHPBS17}, and we prove general results for arbitrary state over time constructions satisfying certain subsets of these axioms. For example, we show that our Bayes' rule with respect to any state over time construction satisfying what we refer to as the \emph{classical limit axiom} yields the state-update rule associated with quantum measurment~\cite{vN18,Lu06,Kr83}. Moreover, with such axioms we are able to identify the key differences between these various approaches towards quantum Bayes' rules, while at the same time incorporating them all into a single framework. 

Another application of our quantum Bayes' rule is in regard to the notion of time-reversal in quantum theory. In particular, our Bayes' rule yields a novel notion of a \emph{Bayesian inverse} of a quantum channel with respect to a prior state which is to dynamically evolve according to the channel. Moreover, we show that the Bayesian inverse of a CPTP map generally differs from that of its Hilbert--Schimdt adjoint. Thus, our Bayesian inverses provide a more robust notion of time-reversal in quantum theory in parallel with other results on retrodictability~\cite{Wat55,ReAh95,BuSc21,AwBuSc21,DDR21,PaBu22}. And while we find that the Hilbert--Schmidt adjoint does not in general provide an appropriate notion of time-reversal, in the case of bistochastic channels, we show that the Bayesian inverse with respect to the uniform prior is indeed the Hilbert--Schmidt adjoint of the channel. This result is in fact independent of a state over time construction provided it satisfies the classical limit axiom. This clarifies why bistochastic channels are often viewed as the only channels exhibiting a canonical time-reversal map~\cite{CGS17,CAZ21,SSGC22}.  

Throughout our work, we illustrate our definitions through a multitude of examples, including the two-state vector formalism~\cite{Wat55,ABL64,ReAh95}, 
the time-dependent two-point correlator~\cite{BDOV13,BHR17}, 
the symmetric bloom of the present authors~\cite{FuPa22}, 
the non-commutative Bayes' theorem of the first author~\cite{PaRuBayes}, 
the quantum Bayes' rule of C.\ Fuchs~\cite{Fu01}, 
the causal state formalism of M.\ Leifer and R.\ Spekkens with the Petz recovery map~\cite{LeSp13,Pe84}, 
the compound states of M.\ Ohya~\cite{Oh83a,Oh83b},
generalized conditional expectations~\cite{Ts22}, 
and the state-update rule associated with quantum measurements, the latter of which has often been called the quantum analogue of Bayes' rule by J.\ Bub, M.\ Ozawa, M.\ Tegmark, and others~\cite{Bu77,Oz97,Oz98,Te12}.

\section{States over time, Bayes' rules, and Bayesian inverses}

\subsection{Notation and conventions}

In this work, we represent finite-dimensional hybrid classical/quantum systems by \define{multi-matrix algebras} (as in Ref.~\cite{Fa01} and~\cite[Chapter 2]{GdHJ89}), which are direct sums of matrix algebras, and which will be denoted by $\mA,\mB,\mC,\dots$. Any such multi-matrix algebra $\mA$ is therefore of the form $\mA=\bigoplus_{x\in X}\matr_{m_{x}}$, with $X$ some finite index set, $\{m_{x}\}$ positive integers, and $\matr_{d}$ denoting $d\times d$ complex matrices. By using multi-matrix algebras instead of just operators on some Hilbert space, we will have a single framework where all concepts such as conditional probabilities, density matrices, positive operator-valued measures (POVMs), ensemble preparations, instruments, and quantum channels are all instances of completely positive trace-preserving (CPTP) maps between such multi-matrix algebras.
By working with CPTP maps, we will be working in the Schr\"odinger picture of quantum theory, in contrast to our previous work~\cite{FuPa22}, where our results were formulated in the Heisenberg picture. 
As such, some notation and terminology will differ.  
We now provide the basic definitions and notation which will be used throughout~\cite{Fa01,OhPe93}.

Every multi-matrix algebra $\mA=\bigoplus_{x\in X}\matr_{m_{x}}$ 
has a trace $\tr$ whose value on $A=\bigoplus_{x\in X}A_{x}$ is given by $\tr(A)=\sum_{x\in X}\tr(A_{x}),$ where the latter trace is the standard (un-normalized) trace on the matrix algebra $\matr_{m_{x}}$. A \define{density matrix} (or \define{state}) in $\mA$ is an element $\rho$ of $\mA$ such that $\tr(\rho)=1$ and $\rho$ is positive, which means there exists some element $A\in\mA$ such that $A^{\dag}A=\rho$. Here, the $\dag$ denotes the component-wise conjugate transpose, namely $A^{\dag}=\bigoplus_{x\in X}A_{x}^{\dag}$, and the multiplication is also component-wise $AA'=\bigoplus_{x\in X}A_{x}A'_{x}$. The set of all states in $\mA$ is denoted by $\St(\mA)$.

We immediately give some examples. First, if $m_{x}=1$ for all $x\in X$, then each matrix algebra is one-dimensional so that $\rho$ corresponds to a collection of non-negative numbers $\{\rho_{x}\}$ whose sum satisfies $\sum_{x\in X}\rho_{x}=1$. In other words, a density matrix can be viewed as a probability distribution on the index set $X$. On the other extreme, if $X$ is an index set with only a single element, then $\rho$ is an ordinary density matrix in a matrix algebra. Note that in the general case of a multi-matrix algebra, the $x$ component of $\rho$, namely $\rho_{x}$, need not be a density matrix since its trace can be less than $1$. However, if the trace is non-zero, then $\frac{\rho_{x}}{\tr(\rho_{x})}$ 
is a density matrix on $\matr_{m_{x}}$ in the usual sense. Hence, a density matrix $\rho$ on a multi-matrix algebra, also called a \emph{classical-quantum state}, can be viewed as a collection of density matrices on possibly different matrix algebras weighted by some probability distribution, namely $\rho=\bigoplus_{x\in X}\tr(\rho_{x})\frac{\rho_{x}}{\tr(\rho_{x})}$. 

For more examples, let $\mB=\bigoplus_{y\in Y}\matr_{n_{y}}$ be another multi-matrix algebra and let  $\mathcal{E}:\mA\to\mB$ be a CPTP map. 
This specializes to many cases of interest in quantum information theory (see Table~\ref{table:multimatrixex} for a brief summary and the introduction of Ref.~\cite{SSGC22} for more details). 
\begin{table}
\centering
\begin{tabular}{c|c}
$\bigoplus_{x\in X}\matr_{m_{x}}\xrightarrow{\mathcal{E}}\bigoplus_{y\in Y}\matr_{n_{y}}$&concept in quantum information\\
\hline
$|X|=1$, $m=1=n_{y}$ $\forall\; y$ & probability distribution on $Y$\\
$|X|=1=|Y|$, $m=1$ & density matrix on $\matr_{n}$\\
$m_{x}=1=n_{y}$ $\forall\;x,y$ & classical channel $X\to Y$\\
$|X|=1=|Y|$ & quantum channel $\matr_{m}\to\matr_{n}$\\
$|X|=1$, $n_{y}=1$ $\forall\;y$ & POVM on $\matr_{m}$ \\
$|Y|=1$, $m_{x}=1$ $\forall\;x$ & ensemble preparation on $\matr_{n}$ \\
$|X|=1$, $n_{y}=n$ $\forall\;y$ & quantum instrument
\end{tabular}
\caption{This table summarizes some of the classical and quantum information-theoretic realizations of a CPTP map between multi-matrix algebras. The notation $m$ ($n$), as opposed to $m_{x}$ ($n_{y}$), is used for the size of the matrices in the algebra when $|X|=1$ ($|Y|=1$) or when all values coincide.}
\label{table:multimatrixex}
\end{table}
\begin{enumerate}
\item
When $m_{x}=1$ and $n_{y}=1$ for all $x\in X, y\in Y$, the map $\mathcal{E}$ corresponds to a classical channel, i.e., a collection of conditional probabilities $\mathbb{P}(y|x)$. More explicitly, if $\delta_{x}$ denotes the unit vector with $1$ in the $x$ component and $0$ otherwise,  then $\mathbb{P}(y|x)$ is the $y$ component of $\mathcal{E}(\delta_{x})$. 
\item
When $X$ and $Y$ have only one element each, $\mathcal{E}$ corresponds to a quantum channel between matrix algebras. 
\item
When $n_{y}=1$ for all $y\in Y$ and $X$ has only a single element, $\mathcal{E}$ corresponds to a POVM. Indeed, each $y$ component of $\mathcal{E}$ defines a positive functional $\mathcal{E}_{y}:\matr_{m}\to\C$, which equals $\tr(M_{y}\;\cdot\;)$ for some unique positive matrix $M_{y}\in\matr_{m}$. The trace-preserving condition guarantees that $\sum_{y\in Y}M_{y}=\mathds{1}_{m}$. 
\item
Dually, when $m_{x}=1$ for all $x\in X$ and $Y$ has only a single element, $\mathcal{E}$ corresponds to an ensemble preparation. Indeed, $\mathcal{E}$ sends each unit vector $\delta_{x}$ to some positive matrix $\rho_{x}$ in $\matr_{n}$. The trace-preserving condition guarantees that $\tr(\rho_{x})=1$ so that $\rho_{x}$ is a density matrix for each $x\in X$. 
\item
When $n_{y}=n$ for all $y\in Y$ for some positive integer $n$ and $X$ has only a single element, $\mathcal{E}:\matr_{m}\to\bigoplus_{y\in Y}\matr_{n}$ corresponds to a quantum instrument. Indeed, the projection of $\mathcal{E}$ onto the $y$ component defines a CP map $\mathcal{E}_{y}:\matr_{m}\to\matr_{n}$. By the trace-preserving condition on $\mathcal{E}$, the sum $\sum_{y\in Y}\mathcal{E}_{y}:\matr_{m}\to\matr_{n}$ is CPTP. This is the usual definition for an instrument with a finite outcome space $Y$~\cite{DaLe70}.
\end{enumerate}

More generally, arbitrary direct sums of matrix algebras can be used to describe certain superselection sectors, as discussed in Ref.~\cite{WWW52} and~\cite[Section 2.3]{St75}.
Thus, we find that a number of fundamental concepts in quantum information theory may be formulated in terms of completely positive maps between multi-matrix algebras. 
Moreover, as every finite-dimensional $C^*$-algebra is isomorphic to a multi-matrix algebra~\cite{Fa01}, our formulation provides a stepping stone towards generalizations to infinite-dimensional quantum systems, further justifying our use of the multi-matrix algebra formalism.

Henceforth, if $\mA$ and $\mB$ are multi-matrix algebras, the collection of all CPTP maps from $\mA$ to $\mB$ will be denoted by $\CPTP(\mA,\mB)$. 
If $\mathcal{E}:\mA\to\mB$ is a linear map, its Hilbert--Schmidt adjoint $\mathcal{E}^*:\mB\to\mA$ is the unique linear map satisfying
\be
\tr\big(\mathcal{E}(A)^{\dag}B\big)=\tr\big(A^{\dag}\mathcal{E}^*(B)\big)
\ee
for all $A\in\mA$ and $B\in\mB$. 
Let $\mu_{\mA}:\mA\otimes\mA\to\mA$ denote the multiplication map, which is uniquely determined by its assignment on tensors via $A_{1}\otimes A_{2}\mapsto A_{1}A_{2}$. 
Given $\mathcal{E}$ as above, its
associated \emph{channel state} is the element of $\mA\otimes\mB$ given by 
\begin{equation}
\label{eq:Jcs}
\mathscr{D}_{\mA,\mB}[\mathcal{E}]:=(\id_{\mA}\otimes\mathcal{E})\big(\mu_{\mA}^{*}(1_{\mA})\big),
\end{equation}
where $1_{\mA}$ is the unit element in $\mA$ and $\mu_{\mA}^{*}$ is the Hilbert--Schmidt adjoint of the multiplication map. 
When the algebras are clear from context, the shorthand $\mathscr{D}[\mathcal{E}]$ will be used.
Although this definition may look unfamiliar at first, it reduces to two familiar cases for certain multi-matrix algebras. 
First, when $\mA=\matr_{m}$ and $\mB=\matr_{n}$, the channel state reduces to the associated \emph{Jamio{\l}kowski state}~\footnote{
The Jamio{\l}kowski state is commonly defined as the partial transpose of the \emph{Choi matrix}~\cite{Ch75}. Equation~\eqref{eq:Jcs} provides a manifestly basis-independent formula for the Jamio{\l}kowski state that is also applicable to arbitrary multi-matrix algebras.
}
\be
\mathscr{D}[\mathcal{E}]=\sum_{i,j}E_{ij}^{(m)}\otimes\mathcal{E}(E_{ji}^{(m)}), 
\ee
where $\{E_{ij}^{(m)}\}$ are the standard matrix units in $\matr_{m}$. In the special case $\mathcal{E}=\id_{\mA}$, the Jamio{\l}kowski state becomes the swap operator, which satisfies $\mathscr{D}[\id_{\mA}]\big(|i\>\otimes|j\>\big)=|j\>\otimes|i\>$ for all $i,j\in\{1,\dots,m\}$ (Dirac notation has been implemented).
Second, when $\mA=\bigoplus_{x\in X}\C\equiv\C^{X}$ and $\mB=\bigoplus_{y\in Y}\C\equiv\C^{Y}$, so that a positive trace-preserving map $\mathcal{E}$ corresponds to conditional probabilities $\mathbb{P}(y|x)$, the channel state is an element of $\bigoplus_{(x,y)\in X\times Y}\C\equiv\C^{X\times Y}$ 
whose $(x,y)$ component is precisely given by $\mathbb{P}(y|x)$. 

\subsection{Main definitions}

\begin{definition}
\label{DXS81}
A \define{state over time function} associates every pair $(\mA,\mB)$ of multi-matrix algebras with a map $\star_{\mA\mB}:\CPTP(\mA,\mB)\times\St(\mA)\to\mA\otimes\mB$, whose value on $(\mathcal{E},\rho)$ is denoted by $\mathcal{E}\star_{\mA\mB}\rho$ (or just $\mathcal{E}\star\rho$ when $\mA$ and $\mB$ are clear) 
\footnote{The notation $\mathcal{E}\star\rho$ is based on Refs.~\cite{LeSp13,HHPBS17}.},
such that 
$\star_{\mA\mB}$ preserves marginal states in the sense that 
\be
\tr_{\mB}\left(\mathcal{E}\star_{\mA\mB}\rho\right)=\rho
\quad\text{and}\quad
\tr_{\mA}\left(\mathcal{E}\star_{\mA\mB}\rho\right)=\mathcal{E}(\rho),
\ee
where $\tr_{\mB}:\mA\otimes\mB\to\mA$ and $\tr_{\mA}:\mA\otimes\mB\to\mB$ denote the partial traces. 
In such a case, the element $\mathcal{E}\star_{\mA\mB}\rho\in \mA\otimes\mB$ 
will be referred to as the \define{state over time} associated with $\star_{\mA\mB}$ and the input $(\mathcal{E},\rho)$. 
\end{definition}

To incorporate more examples and also various approaches to states over time, the definition given here is less restrictive than the definition of a state over time function given in Refs.~\cite{HHPBS17,FuPa22}. In particular, this definition includes only the bare minimum of what one would expect from a state over time function, namely, that the output element $\mathcal{E}\star\rho$ has the expected marginals. 
Note that unitality, as defined in Refs.~\cite{FuPa22,HHPBS17}, holds automatically because if $\rho$ is a density matrix and $\mathcal{E}$ is trace-preserving, then $\tr(\mathcal{E}\star\rho)=1$ because of the marginal preservation property.
Of course however, one would expect a physically meaningful state over time function to satisfy additional properties.

\begin{definition}
\label{defn:saxs}
A state over time function $\star$ 
\begin{enumerate}[(P1)]
\item
\label{item:hermitian}
is \define{hermitian} iff $\mathcal{E}\star\rho$ is self-adjoint for all $\rho\in\St(\mA)$ and $\mathcal{E}\in\CPTP(\mA,\mB)$.
\item
\label{item:locpos}
is \define{locally positive} (or \define{block positive}) iff 
\be
\tr\big((\mathcal{E}\star\rho)^{\dag}(A\otimes B)\big)\ge0
\ee
for all $\rho\in\St(\mA)$, $\mathcal{E}\in\CPTP(\mA,\mB)$, and for all positive $A\in\mA$ and $B\in\mB$.
\item
\label{item:pos}
is \define{positive} iff $\mathcal{E}\star\rho$ is positive for all $\rho\in\St(\mA)$ and $\mathcal{E}\in\CPTP(\mA,\mB)$.
\item
\label{item:slin}
is \define{state-linear} iff 
\be
\mathcal{E}\star\big(\lambda\rho+(1-\lambda)\sigma\big)
=\lambda\mathcal{E}\star\rho+(1-\lambda)\mathcal{E}\star\sigma
\ee
for all $\lambda\in[0,1]$, $\rho,\sigma\in\St(\mA)$, and $\mathcal{E}\in\CPTP(\mA,\mB)$.
\item
\label{item:plin}
is \define{process-linear} iff 
\be
\big(\lambda\mathcal{E}+(1-\lambda)\mathcal{F}\big)\star\rho
=\lambda\mathcal{E}\star\rho+(1-\lambda)\mathcal{F}\star\rho
\ee
for all $\lambda\in[0,1]$, $\rho\in\St(\mA)$, and $\mathcal{E},\mathcal{F}\in\CPTP(\mA,\mB)$.
\item
\label{item:bilin}
is \define{bilinear} iff $\star$ is state-linear and process-linear.
\item
\label{item:classicallimit}
satisfies the \define{classical limit axiom} iff given any $\rho\in\St(\mA)$ and $\mathcal{E}\in\CPTP(\mA,\mB)$ 
satisfying
$[\mathscr{D}[\mathcal{E}],\rho\otimes1_{\mB}]=0$
implies~%
\footnote{The condition $[\mathscr{D}[\mathcal{E}],\rho\otimes1_{\mB}]=0$ is equivalent to $[\rho,\mathcal{E}^*(B)]=0$ for all $B\in\mB$.}
\be
\mathcal{E}\star\rho
=\mathscr{D}[\mathcal{E}](\rho\otimes1_{\mB}).
\ee
\end{enumerate}
\end{definition}

There is also an associativity axiom guaranteeing that a state over time function yields a consistent notion of a \emph{tripartite} state over time on $\mA\otimes\mB\otimes\mC$ associated with every composable pair $\mA\xrightarrow{\mathcal{E}}\mB\xrightarrow{\mathcal{F}}\mC$ of CPTP maps.
We will mention in our examples when this associativity axiom holds, but we will otherwise not make use of it in this work~
\footnote{
For reference, and to provide a formula in the Schr{\"o}dinger picture, a state over time function $\star$ is \define{associative} iff 
\[
\begin{split}
&\mathscr{D}_{\mA,\mB\otimes\mC}^{-1}\left[\tr(1_{\mA})\left((\mathcal{F}\circ\tr_{\mA})\star\left(\frac{\mathscr{D}_{\mA,\mB}[\mathcal{E}]}{\tr(1_{\mA})}\right)\right)\right]\star\rho\\
&=(\mathcal{F}\circ\tr_{\mA})\star(\mathcal{E}\star\rho)
\end{split}
\]
for all $\rho\in\St(\mA)$ and composable pairs $\mA\xrightarrow{\mathcal{E}}\mB\xrightarrow{\mathcal{F}}\mC$ of CPTP maps. Note that one must be careful about domains to even make sense of this axiom. 
In particular, we point out that the formula presented here is a more appropriate axiom for associativity than the one considered in Ref.~\cite{FuPa22} due to the fact that $\star$ need not be state-linear (this is why we have included the factors of $\tr(1_{\mA})$ in the present formulation). Indeed, if $\star$ is state-linear, then these factors cancel, which need not happen if $\star$ is not state-linear. Note that this change does not alter the main theorem of Ref.~\cite{FuPa22}, since the state over time function constructed there is state-linear.}.

Note that if a state over time function $\star$ satisfies the classical limit axiom, then $\mathcal{E}\star\frac{1_{\mA}}{\tr(1_{\mA})}=\frac{1}{\tr(1_{\mA})}\mathscr{D}[\mathcal{E}]$. This shows that such a $\star$ can be viewed as an  extension of the Jamio{\l}kowski isomorphism to include states besides the maximally mixed state (though it need not be bijective). 
We also note that a state over time that is positive appears in the marginal state problem, and it is sometimes called a \emph{compound state}~\cite{AKMS06,MaCh22,Oh83a,Oh83b,GPS21}. Since we will not require positivity, our considerations will be more general.

We will provide many examples of state over time functions in the remaining sections. But first, we introduce the definition of a Bayesian inverse with respect to a state over time function in terms of a quantum Bayes' rule. 

\begin{definition}
\label{defn:Bayesianinverse}
Let $\star$ be a state over time function. Given a density matrix (a \emph{prior}) $\rho\in\St(\mA)$ and a CPTP map $\mathcal{E}:\mA\to\mB$ (a process), 
a \define{Bayesian inverse} associated with $(\mathcal{E},\rho)$ is a CPTP map 
$\mathcal{E}^{\star}_{\rho}:\mB\to\mA$
such that 
\be
\mathcal{E}\star\rho
=
\tau\left(\mathcal{E}^{\star}_{\rho}\star\mathcal{E}(\rho)\right),
\ee
where $\tau:\mB\otimes\mA\to\mA\otimes\mB$ is the \define{quantum time-reversal map} for states over time, defined as the unique conjugate-linear extension of the assignment
\be
\tau(B\otimes A)=A^{\dag}\otimes B^{\dag}.
\ee
The equation $\mathcal{E}\star\rho=\tau\left(\mathcal{E}^{\star}_{\rho}\star\mathcal{E}(\rho)\right)$  is then referred to as \define{Bayes' rule} associated with $\star$ and the input $(\mathcal{E},\rho)$. 
\end{definition}

By applying the partial trace $\tr_{\mB}$ to both sides of Bayes' rule, it follows that 
\be
\rho=\mathcal{E}^{\star}_{\rho}\big(\mathcal{E}(\rho)\big).
\ee
In other words, if $\rho$ is thought of as a prior, $\mathcal{E}(\rho)$ the associated prediction via $\mathcal{E}$, and $\mathcal{E}^{\star}_{\rho}$ the associated retrodiction map, then the retrodiction applied to the prediction gives back the prior. If $\mathcal{E}^{\star}_{\rho}$ exists and is unique for all $\mathcal{E}$ and $\rho$ (such that $\rho$ and $\mathcal{E}(\rho)$ are faithful), this assignment defines a (universal) \emph{recovery map} in the language of~\cite{FaRe15,Wilde15,JRSWW16,CaVe20}, a \emph{state-retrieval map} in the language of~\cite{SuSc22}, and a \emph{retrodiction family} in the language of~\cite{PaBu22}. 

While the linear swap map
$\gamma:\mB\otimes\mA\to\mA\otimes\mB$ provides a suitable notion of time-reversal for states over time in the classical setting (as described in the introduction)~%
\footnote{For general algebras, the swap map $\gamma:\mA\otimes\mB\to\mB\otimes\mA$ is not to be confused with the swap operator $\mathscr{D}[\id_{\mA}]$, the latter of which is an \emph{element} of $\mA\otimes\mA$.},
we find that composing the swap map with the dagger provides a more robust notion of time-reversal for states over time in the quantum setting (see Figure~\ref{fig:Bayesianinverse}). 
Indeed, the swap map $\gamma$ on its own only guarantees that the right-hand-side of Bayes' rule is an element of $\mA\otimes\mB$.
However, if a quantum channel $\mathcal{E}$ is invertible with its inverse also a channel, then the usage of $\tau$ ensures that $\mathcal{E}^{-1}$ is a Bayesian inverse for $\mathcal{E}$, which would not necessarily be the case if we had simply used $\gamma$ (see Remark~\ref{rmk:LBInv} for more details). Remarks~\ref{rmk:marginB}, \ref{rmk:BaTh}, and~\ref{rmk:TsInv} provide further justifications for supplementing the swap map with the dagger in the quantum setting. 
Moreover, $\tau$ reduces to $\gamma$ for classical systems, and the Bayes' rule from Definition~\ref{defn:Bayesianinverse} coincides with the standard Bayes' rule (Equation~\ref{CXBR17})
on commutative algebras when we take the state over time function to be the standard one~\cite[Section~1]{FuPa22}. We will see this in the next section explicitly, along with many other examples from the literature, some of which have also been called quantum Bayes' rules.

\begin{figure}
    \centering    
    \begin{tikzpicture}
    \node (t0) at (-1.5,0) {$t_0$};
    \node (t1) at (-1.5,2) {$t_1$};
    \node (t02) at (6.25,0) {$t_0$};
    \node (t12) at (6.25,2) {$t_1$};
    \draw[->] (t0) --node[left]{\rotatebox{90}{time}} (t1);
    \draw[->] (t12) --node[right]{\rotatebox{270}{time}} (t02);
    \node (t1L) at (-0.75,0) {$\mA$};
    \node (tL) at (-0.75,2) {$\mB$};
    \node (t1L2) at (5.5,0) {$\mA$};
    \node (tL2) at (5.5,2) {$\mB$};
    \node (t1Lr) at (0,0) {$\rho$};
    \node (tLr) at (0,2) {$\mathcal{E}(\rho)$};
    \node (t1Lr2) at (4.75,0) {$\rho$};
    \node (tLr2) at (4.75,2) {$\mathcal{E}(\rho)$};
    \draw[->] (t1L) to node[left]{$\mathcal{E}$} (tL);
    \draw[->] (tL2) to node[right]{$\mathcal{E}^{\star}_{\rho}$} (t1L2);
    \draw[|->] (t1Lr) to (tLr);
    \draw[|->] (tLr2) to (t1Lr2);
    \node (AB) at (1.25,0.5) {$\mA\otimes\mB$};
    \node (Esr) at (3.25,1.25) {$\mathcal{E}^{\star}_{\rho}\star\mathcal{E}(\rho)$};
    \node (BA) at (3.25,0.5) {$\mB\otimes\mA$};
    \node (Esr2) at (1.25,1.25) {$\mathcal{E}\star\rho$};
    \draw[|->] (Esr) -- node[above]{$\tau$} (Esr2);
    \end{tikzpicture}
    \caption{If $\mathcal{E}:\mA\to\mB$ is viewed as a CPTP map describing some dynamics from initial time $t_0$ to final time $t_1$, then this figure  depicts the two states over time associated with $(\mathcal{E},\rho)$ and a Bayesian inverse. If $\mathcal{E}\star\rho$ is interpreted as having a time orientation $t_{0}\rightarrow t_{1}$, then $\mathcal{E}^{\star}_{\rho}\star\mathcal{E}(\rho)$ has time orientation $t_{1}\rightarrow t_{0}$. The quantum time-reversal map $\tau$ simultaneously reverses the orientation of time and switches the two factors so that the resulting elements can be compared on an equal footing. Bayes' rule says that these two elements are the same. In this way, $\tau$ embodies a fundamental time-reversal symmetry for states over time associated with any input process and state $(\mathcal{E},\rho)$ that admits a Bayesian inverse (see Ref.~\cite{PaBu22} for a closely related inferential form of time-reversal symmetry).}
    \label{fig:Bayesianinverse}
\end{figure}

Furthermore, our formulation of Bayes' rule using our quantum time-reversal map $\tau$ solves several open questions in the literature. First, it resolves a puzzle posed by M.\ Leifer and R.\ Spekkens at the end of~\cite[Section~VII.B.1]{LeSp13}, where they observe that 
using $\gamma$ alone is not sufficient to provide enough symmetry to relate states over time in the forward and backwards time directions. By adding a dagger, we have solved this problem and restored the symmetry.
Second, we show how M.\ Tsang's Bayes' rules obtained from certain inner products~\cite{Ts22}, are derived from a certain class of state over time functions and our Bayes' rule. This answers Tsang's open remark/question (posed at the end of \cite[Section~III.B]{Ts22}) in regard to the relationship between generalized conditional expectations and states over time. Moreover, combining these two previous points yields that our Bayes' rule provides a possible answer to the question of J.\ Baez on the relationship between time-reversal and the inner product in quantum theory~\cite{Ba06}. In particular, our present work combined with Ref.~\cite{PaBu22} suggests that retrodiction based upon our quantum Bayes' rule may provide a mathematically precise relationship between time-reversal and the inner product that is more robust than the ordinary adjoint operation on quantum channels~\cite{Ba06,CoKi17,CGS17,SSGC22,CAZ21,ChLi22}.

\section{First examples}

The usefulness of a definition is illustrated through its examples. In what follows, we first justify the significance of the classical limit axiom by considering not only the case of commutative algebras, but also bistochastic channels (between possibly noncommutative algebras). We then provide an example of a state over time function where the classical limit axiom fails by removing all correlations. Afterwards, we consider 
the Leifer--Spekkens state over time~\cite{LeSp13}. From there, we describe two measurement scenarios. First, we show how prepare-evolve-measure scenarios have an inferential time-reversal symmetry involving our Bayes' rule. Second, we show how the state-update rule due to measurement is an instantiation of our Bayes' rule. 

\subsection{The classical case}

The following example provides some motivation and justification for the importance of the classical limit axiom. Consider the special case where $\mA=\C^{X}\equiv\bigoplus_{x\in X}\C$ and $\mB=\C^{Y}\equiv\bigoplus_{y\in Y}\C$ with $X$ and $Y$ finite sets. Then a state $\rho$ on $\mA$ corresponds to a probability distribution $\{p_{x}\}$ on $X$, and a CPTP map $\mathcal{E}:\mA\to\mB$ corresponds to a stochastic map from $X$ to $Y$ with conditional probabilities denoted by $\mathcal{E}_{yx}$. 

In such a case, we have $[\mathscr{D}[\mathcal{E}],\rho\otimes 1_{\mB}]=0$ 
since the algebras $\mathcal{A}$ and $\mathcal{B}$ are both commutative. Because the associated channel state $\mathscr{D}[\mathcal{E}]$ is the element of $\mA\otimes\mB\cong\bigoplus_{x,y}\C$ whose $(x,y)$ component is $\mathcal{E}_{yx}$, it follows that $(\rho\otimes 1_{\mB})\mathscr{D}[\mathcal{E}]$ has $(x,y)$ component given by $p_{x}\mathcal{E}_{yx}$, which we refer to as the \define{classical state over time}. Write $q_{y}:=\sum_{x\in X}\mathcal{E}_{yx}p_{x}$ as the probability distribution on $Y$ corresponding to $\mathcal{E}(\rho)$. A Bayesian inverse of $(\mathcal{E},\rho)$ is therefore a CPTP map $\mathcal{E}^{\star}_{\rho}:\mB\to\mA$, with corresponding conditional probabilities written as $(\mathcal{E}^{\star}_{\rho})_{xy}$, such that 
\be
\mathcal{E}_{yx}p_{x}=(\mathcal{E}^{\star}_{\rho})_{xy}q_{y},
\ee
which is the classical Bayes' rule (cf.\ Equation~\eqref{CXBR17}). 

If $\mathcal{E}$ is viewed as a genuine stochastic process, the physical intuition behind the term $\mathcal{E}_{yx}p_{x}$ is that it describes the predictive probability of first measuring $x$ and then measuring $y$ after the system has undergone the evolution described by $\mathcal{E}$ (note that this probability does \emph{not} in general equal $p_{x}q_{y}$, which would require the random variables associated with $X$ and $Y$ to be independent/uncorrelated). Conversely, $(\mathcal{E}^{\star}_{\rho})_{xy}q_{y}$ describes the retrodictive probability of measuring $y$ and deducing that $x$ preceded it in the course of evolution through the inference map $\mathcal{E}^{\star}_{\rho}$~\cite{Ba14}. 

Of course, the classical limit axiom covers far more cases than this, such as when the algebras $\mA$ and $\mB$ are not necessarily commutative and yet the commutativity condition still holds. 
One example is the case of \emph{unital} CPTP maps, which includes bistochastic matrices, and which is described in the next section.

\subsection{Bistochastic channels and time-reversal symmetry}

Let $\star$ be any state over time function that satisfies the classical limit axiom. Take $\mA=\matr_{m}$, $\mB=\matr_{n}$ and let $\mathcal{E}:\mA\to\mB$ be a \emph{unital} quantum channel (sometimes called a \emph{bistochastic channel}), i.e., $\mathcal{E}(\mathds{1}_{m})=\mathds{1}_{n}$.
If $\rho=\frac{\mathds{1}_{m}}{m}$, then the associated state over time is always given by $\rho\star\mathcal{E}=\frac{1}{m}\mathscr{D}[\mathcal{E}]$. Physically, $\mathcal{E}$ describes a stochastic evolution that leaves the infinite temperature limit Gibbs state invariant (it is a consequence of unitality and the trace-preserving condition of $\mathcal{E}$ that $m=n$). 
If $\mA$ and $\mB$ were commutative algebras instead of matrix algebras, then $\mathcal{E}$ would be a doubly-stochastic matrix (as the matrix is necessarily square since $\mA\cong\mB$).

Hence, when $\mathcal{E}$ is a unital quantum channel and the prior is the uniform density matrix $\rho=\frac{\mathds{1}_{m}}{m}$, then a Bayesian inverse $\mathcal{E}^{\star}_{\rho}$ must satisfy the equation 
\be
\mathscr{D}[\mathcal{E}^*]
=\gamma\big(\mathscr{D}[\mathcal{E}]\big)
=\gamma\big(\mathscr{D}[\mathcal{E}]^{\dag}\big)
=\mathscr{D}[\mathcal{E}^{\star}_{\rho}],
\ee
where the first equality holds by Lemma~\ref{lem:gamDE}, the second equality holds because $\mathscr{D}[\mathcal{E}]$ is self-adjoint, and the third equality holds by our definition of Bayes' rule.
In other words, $\mathcal{E}^{\star}_{\rho}=\mathcal{E}^*$, which provides some justification for the notation we have used for a Bayesian inverse. 

More importantly, this result has significant implications towards our understanding of time-reversal symmetry in quantum theory. Indeed, it has often been argued that unital quantum channels are the only channels for which a canonical notion of time-reversal is possible ~\cite{CGS17,CAZ21}, and in such a case, the canonical time-reversal is provided by the Hilbert--Schmidt adjoint. Such a claim however is at odds with classical probability theory, as not all classical channels are bistochastic, and yet they have a well-defined notion of time-reversal symmetry furnished by the classical Bayes' rule. It then seems plausible that there may exist a more general notion of time-reversal for quantum channels, which not only agrees with the Bayesian inverse for classical channels, but also reduces to the Hilbert--Schmidt adjoint for bistochastic channels as well. And while no-go theorems were recently proved in Refs.~\cite{CGS17,CAZ21} stating that there is no such notion of time-reversal for arbitrary quantum channels, it was recently shown in Ref.~\cite{PaBu22} via an explicit construction that time-reversal \emph{is} possible for \emph{all} quantum channels once a prior state is incorporated into the data of the theory (the explicit construction will be discussed in more detail later in this work). 
The perspective gained from Ref.~\cite{PaBu22} is that when considering time-reversal in quantum theory, one should be working in the category of \emph{states} on multi-matrix algebras and state-preserving channels, as opposed to the category of channels on their own.

We can therefore provide a possible explanation as to why bistochastic channels are often viewed as the only reversible operations in quantum theory. Indeed, our results show that there is a \emph{unique} Bayesian inverse for any bistochastic channel with the uniform prior, and this result is \emph{independent} of the channel and a choice of a state over time function (as long as the state over time function satisfies the classical limit axiom). Moreover, this \emph{canonical} Bayesian inverse is in fact the Hilbert--Schmidt adjoint of the original channel, which agrees with the standard time-reversal map when the channel is unitary. However, the Hilbert--Schmidt adjoint is \emph{rarely} a Bayesian inverse with respect to a state over time function when one has a \emph{non-uniform} prior and an \emph{arbitrary} quantum channel~\cite{BuSc21}, and it is precisely this more general situation where one needs additional input to specify time-reversal symmetry and Bayesian inverses, and that input is the (non-canonical) choice of a state over time function.

As such, one might hope to find a state over time function whose associated Bayesian inverses simultaneously extend both classical Bayesian inversion for arbitrary stochastic channels and the Hilbert--Schmidt adjoint for all unital quantum channels. 
We will consider such examples soon. But first, we will give some examples of states over time that forget correlations and entanglement. 

\subsection{The uncorrelated state over time}
\label{sec:uncsot}

The assignment 
\be
(\mathcal{E},\rho)\xmapsto{\star}\rho\otimes\mathcal{E}(\rho)
\ee
is a state over time function, called the \define{uncorrelated state over time}, that is positive (and hence also hermitian) and process-linear. However, it is not state-linear, it is not associative~%
\footnote{
This is because 
\[
\begin{split}
&\mathscr{D}_{\mA,\mB\otimes\mC}^{-1}\left[\tr(1_{\mA})\left((\mathcal{F}\circ\tr_{\mA})\star\left(\frac{\mathscr{D}_{\mA,\mB}[\mathcal{E}]}{\tr(1_{\mA})}\right)\right)\right]\star\rho\\
&=\rho\otimes\mathcal{E}(\rho)\otimes\mathcal{F}\left(\mathcal{E}\left(\frac{1_{\mA}}{\tr(1_{\mA})}\right)\right),
\end{split}
\]
while 
\[
(\mathcal{F}\circ\tr_{\mA})\star(\mathcal{E}\star\rho)=\rho\otimes\mathcal{E}(\rho)\otimes\mathcal{F}\big(\mathcal{E}(\rho)\big).
\]
Since $\mathcal{F}\left(\mathcal{E}\left(\frac{1_{\mA}}{\tr(1_{\mA})}\right)\right)$ need not equal $\mathcal{F}\big(\mathcal{E}(\rho)\big)$ for arbitrary $\mathcal{F},\mathcal{E},$ and $\rho$, 
this shows that associativity fails in general.}, 
and, most importantly, it does not satisfy the classical limit axiom.

A Bayesian inverse $\mathcal{E}^{\star}_{\rho}$ must satisfy the equation
\be
\rho\otimes\mathcal{E}(\rho)=\mathcal{E}^{\star}_{\rho}\big(\mathcal{E}(\rho)\big)\otimes\mathcal{E}(\rho). 
\ee
Thus, \emph{any} CPTP map $\mathcal{E}^{\star}_{\rho}:\mB\to\mA$ such that $\mathcal{E}^{\star}_{\rho}\big(\mathcal{E}(\rho)\big)=\rho$ is a Bayesian inverse of $(\mathcal{E},\rho$). There are many CPTP maps that satisfy this condition, one of which is simply the map that sends $B$ to $\tr(B)\rho$. This, however, is not a very good candidate for Bayesian inversion since it loses the information of most other states and essentially ignores the evolution $\mathcal{E}$ in its description. 
Such a state over time therefore treats $\rho$ and $\mathcal{E}(\rho)$ as independent, and as such, it does not encode any correlations between $\rho$ and $\mathcal{E}(\rho)$.

Furthermore, although there might be a better choice among the many state-preserving CPTP maps that satisfy Bayes' rule in this example, this state over time does not offer us any guide as to \emph{which} of those CPTP maps to choose unless we impose further constraints. As such, the uncorrelated state over time function does not provide a robust formulation of Bayesian inversion.

\subsection{The separable compound state and the quantum state marginal problem}
\label{sec:Ohyasot}

There is an improvement on the uncorrelated state over time due to M.\ Ohya~\cite{Oh83a}, which will henceforth be called \define{Ohya's compound state over time}, following similar terminology in Refs.~\cite{Oh83a,Oh83b}. For matrix algebras $\mA=\matr_{m}$ and $\mB=\matr_{n}$, it is defined as follows~%
\footnote{Technically, Ohya's original definition of a compound state does not give a well-defined state over time function because it assumes \emph{additional data} when the input state has repeating eigenvalues. Namely, it requires a \emph{choice} of an eigenvector decomposition. Therefore, we have slightly modified Ohya's construction so that no such additional data are needed, and so that we obtain a well-defined state over time function.}.

Given any state $\rho\in\mA$, let $\rho=\sum_{\alpha}\lambda_{\alpha}P_{\alpha}$ be the spectral decomposition, where each $P_{\alpha}$ is the projection onto the $\lambda_{\alpha}$ eigenspace. Then for any $\mathcal{E}\in\CPTP(\mA,\mB)$, set 
\be
\mathcal{E}\star\rho=\sum_{\alpha}\lambda_{\alpha}P_{\alpha}\otimes\mathcal{E}\left(\frac{P_{\alpha}}{\tr(P_{\alpha})}\right).
\ee
Then this defines a state over time that is positive and process-linear, and therefore gives a solution to the marginal state problem. It is neither state-linear nor does it satisfy our classical limit axiom. However, we note that if the non-zero eigenvalues of $\rho$ have multiplicity $1$, then the classical limit axiom \emph{does} hold for $(\mathcal{E},\rho)$ for arbitrary $\mathcal{E}\in\CPTP(\mA,\mB)$ such that $\big[\rho\otimes1_{\mB},\mathscr{D}[\mathcal{E}]\big]=0$.
Although the compound state over time is not completely uncorrelated, it is still \emph{separable}~\cite{We89}, and hence seems unlikely to keep track of quantum \emph{entanglement} over time. 

\subsection{The causal states of Leifer and Spekkens and the Petz recovery map}
\label{sec:LSsot}

The assignment sending $(\mathcal{E},\rho)\in\CPTP(\mA,\mB)\times\mathcal{S}(\mA)$ to 
\be
\big(\sqrt{\rho}\otimes 1_{\mB}\big)\mathscr{D}[\mathcal{E}]\big(\sqrt{\rho}\otimes 1_{\mB}\big)
\ee
is called the \define{Leifer--Spekkens state over time}~\cite{FuPa22,HHPBS17,Le06,Le07,LeSp13}. The Leifer--Spekkens state over time function is process linear, hermitian, locally positive, and satisfies the classical limit axiom, 
but it is not in general positive nor associative~\cite{LeSp13,HHPBS17}. One might argue that the reason for the lack of positivity is because we chose to use the Jamio{\l}kowski channel state $\mathscr{D}[\mathcal{E}]$ as opposed to the \emph{Choi matrix} for $\mathcal{E}$~\cite{Ch75}. However, the Choi matrix suffers from a few problems in our context. One is that it depends on a choice of basis, and hence, it would require additional data to specify a general state over time. In addition, if we had used the Choi matrix the marginal density matrix $\tr_{\mA}(\mathcal{E}\star\rho)$ would not be $\mathcal{E}(\rho)$, as it would be $\mathcal{E}(\rho^{\mathrm{T}})$, where ${}^{\mathrm{T}}$ denotes the transpose with respect to that chosen basis~\cite{ChMa20} (repeated eigenvalues prevent one from choosing a canonical basis of eigenvectors for $\rho$ as is done in Ref.~\cite{ChMa20} for their construction of a positive compound state). 

A Bayesian inverse $\mathcal{E}^{\star}_{\rho}$ to the Leifer--Spekkens state over time function must satisfy the equation 
\begin{align}
\gamma\Big(\big(\sqrt{\rho}&\otimes 1_{\mB}\big)\mathscr{D}[\mathcal{E}]\big(\sqrt{\rho}\otimes 1_{\mB}\big)\Big)\nonumber\\
&=\left(1_{\mA}\otimes\sqrt{\mathcal{E}(\rho)}\right)\mathscr{D}\big[\mathcal{E}^{\star}_{\rho}\big]\left(1_{\mA}\otimes\sqrt{\mathcal{E}(\rho)}\right),
\end{align}
which has the unique solution
\be
\mathscr{D}\big[\mathcal{E}^{\star}_{\rho}\big]=\left(\sqrt{\mathcal{E}(\rho)^{-1}}\otimes\sqrt{\rho}\right)\mathscr{D}[\mathcal{E}^*]\left(\sqrt{\mathcal{E}(\rho)^{-1}}\otimes\sqrt{\rho}\right)
\ee
whenever $\mathcal{E}(\rho)$ is non-singular. Solving this using Lemma~\ref{lem:gamDE} yields
\be
\mathcal{E}^{\star}_{\rho}=\Ad_{\rho^{1/2}}\circ\mathcal{E}^{*}\circ\Ad_{\mathcal{E}(\rho)^{-1/2}},
\ee
which is the \define{Petz recovery map}
(sometimes called \emph{transpose channel})~\cite{Pe84,Pe88,OhPe93}. 

This map has appeared in the context of 
sufficient quantum statistics and equality conditions for relative entropy~\cite{Pe86,Pe88,Pe03,Je17}, 
approximate quantum error correction~\cite{BaKn02,NgMa10},
operational time-reversal in quantum theory~\cite{Cr08,LePu17,PaBu22}, 
earlier approaches towards extending Bayesian inversion to the quantum setting~\cite{Le06,Le07}, 
coarse-grained/observational entropy~\cite{BSS22}, 
quantum fluctuation relations~\cite{KwKi19,Hu22}, 
the renormalization group in quantum field theory~\cite{FuLaOu22},
entanglement wedge reconstruction and proposals for a resolution of the black hole information paradox~\cite{CHPSSW19,CPS20,PSSY22,KMM22}, 
and many other contexts. 
The Petz recovery map was recently shown to define the only presently known quantum retrodiction functor~\cite{PaBu22}, which in particular says that it is \emph{compositional} in the sense that 
\be
(\mathcal{F}\circ \mathcal{E})_{\rho}^{\star}=\mathcal{E}_{\rho}^{\star}\circ \mathcal{F}_{\mathcal{E}(\rho)}^{\star}
\ee
for a composable pair $\mA\xrightarrow{\mathcal{E}}\mB\xrightarrow{\mathcal{F}}\mC$ of CPTP maps and an initial state $\rho\in\mathcal{S}(\mA)$.

\subsubsection{The quantum Bayes' rule of C.\ Fuchs}

In the case of a POVM $\mathcal{E}:\M_n\to \C^X$ given by $\mathcal{E}=\bigoplus_{x\in X}\tr(M_{x}\;\cdot\;)$ and prior density matrix $\rho\in \M_n$, our quantum Bayes' rule for the Leifer--Spekkens state over time function yields (by similar calculations to the above) 
\be
\bigoplus_{x\in X}p_{x}\rho_{x}=\bigoplus_{x\in X}\sqrt{\rho}M_{x}\sqrt{\rho}, 
\ee
where $\rho_{x}=\mathcal{E}^{\star}_{\rho}(\delta_{x})$ and $p_{x}=\tr(M_{x}\rho)$ for all $x\in X$. The above equation then yields the quantum Bayes' rule of C.\ Fuchs~\cite[Section~5]{Fu01}, namely 
\be
\rho_{x}=\frac{\sqrt{\rho}M_{x}\sqrt{\rho}}{p_{x}},
\ee
which was also derived using the formalism of \emph{conditional states} in Ref.~\cite{LeSp13}.

\subsection{The \texorpdfstring{$t$}{t}-rotated family and rotated Petz recovery maps}
\label{sec:trotsot}

For each $t\in\R$, the assignment sending $(\mathcal{E},\rho)$ to 
\be
(\rho^{1/2-it}\otimes1_{\mB})\mathscr{D}[\mathcal{E}](\rho^{1/2+it}\otimes1_{\mB})
\ee
defines a state over time function, called the \define{$t$-rotated state over time}, that is process-linear, locally positive, and satisfies the classical limit axiom for all $t\in\R$. 
The Bayesian inverse $\mathcal{E}^{\star}_{\rho}$ in this case is given by
\be
\mathcal{E}^{\star}_{\rho}=\Ad_{\rho^{1/2-it}}\circ\mathcal{E}^{*}\circ\Ad_{\mathcal{E}(\rho)^{-1/2+it}},
\ee
which is the \define{rotated Petz recovery map}~\cite{Wilde15}. This map, and averaged versions of it, have become especially important in the context of information recovery and the strengthening of data-processing inequalities~\cite{FaRe15,Wilde15,JRSWW16,JRSWW18,SuToHa16,CaVe20}. 

More generally, for each $\rho\in\mathcal{S}(\mA)$, choose a unitary $U_{\rho}\in\mA$ such that $[U_{\rho},\rho]=0$. Then the assignment sending $(\mathcal{E},\rho)$ to 
\be
\big(U_{\rho}^{\dag}\rho^{1/2}\otimes1_{\mB}\big)\mathscr{D}[\mathcal{E}]\big(\rho^{1/2}U_{\rho}\otimes1_{\mB}\big)
\ee
defines a state over time that is process-linear, locally positive, and satisfies the classical limit axiom. It is referred to as the \define{Sutter--Tomamichel--Harrow (STH) state over time} based on its appearances in Refs.~\cite{SuToHa16,PaBu22}. The Bayesian inverse $\mathcal{E}^{\star}_{\rho}$ in this case is given by
\be
\mathcal{E}^{\star}_{\rho}=\Ad_{U_{\rho}^{\dag}\rho^{1/2}}\circ\mathcal{E}^{*}\circ\Ad_{U_{\mathcal{E}(\rho)}\mathcal{E}(\rho)^{-1/2}},
\ee
which is a generalization of the rotated Petz recovery map (the special case of rotated Petz is $U_{\rho}=\rho^{it}$ for some $t\in\R$). 

Although the above rotated Petz recovery maps can be obtained by a judicious choice of state over time function, an open question is what state over time function has the universal recovery map of M.\ Junge et al.~\cite{JRSWW16,JRSWW18} as the Bayesian inverse? Or more generally, what state over time functions have averaged rotated Petz recovery maps, as defined in Ref.~\cite{PaBu22}, as their Bayesian inverses? Since the universal recovery map of Refs.~\cite{JRSWW16,JRSWW18} was claimed to provide a quantum generalization of Bayes' rule in Ref.~\cite{CHPSSW19} (based on results from Ref.~\cite{SBT17}), we expect that this arises within our framework.

\section{Measurement}

Measurement in quantum mechanics involves hybrid classical-quantum systems. Therefore, our definitions of states over time and Bayesian inverses should specialize to such settings. In this section, we illustrate this with two key examples. The first example involves preparation, evolution, and then measurement. The second example is the state-update rule associated with the measurement of a quantum system in terms of quantum instruments. 

\subsection{Prepare-evolve-measure scenarios}

Let $X$ be a finite set, let $\mathcal{P}:\C^X\to \M_n$ be a CPTP map, and set $\rho=\mathcal{P}(p)$ for some probability distribution $p\in \C^X$. Such data determines a preparation of the state $\rho$ given by
\be
\rho=\sum_{x\in X}p_x\rho_x,
\ee
where $\rho_x=\mathcal{P}(\delta_x)$. Now suppose the state $\rho$ is to evolve according to a CPTP map $\mathcal{E}:\M_n\to \M_m$, and let $\mathcal{M}:\M_m\to \C^Y$ be a POVM on the output $\M_m$ of $\mathcal{E}$, so that $\mathcal{M}=\bigoplus_{y\in Y}\tr(M_{y}\;\cdot\;)$ for a collection of positive operators $M_y$ summing to the identity, indexed by some finite set $Y$. We then refer to the four-tuple $(p,\mathcal{P},\mathcal{E},\mathcal{M})$ as a \define{prepare-evolve-measure (PEM) scenario}.

Given a PEM scenario $(p,\mathcal{P},\mathcal{E},\mathcal{M})$, we can naturally define a classical channel $f:\C^X\to \C^Y$ given by $f=\mathcal{M}\circ \mathcal{E}\circ \mathcal{P}$ (see Figure~\ref{fig:clasdyn}), which we refer to as the \define{classical dynamics} of the PEM scenario $(p,\mathcal{P},\mathcal{E},\mathcal{M})$. The conditional probabilities $f_{yx}$ associated with 
$f$ 
are then interpreted as the probability of measuring outcome $y$ via the POVM $\mathcal{M}$ given that the state $\rho_x$ was prepared and then evolved under the channel $\mathcal{E}$. Conversely, the classical Bayesian inverse $g:\C^Y\to \C^X$ of $f$ is given by
\be
g_{xy}=\frac{f_{yx}q_y}{p_x},
\ee
where $q_{y}:=\tr\big(M_{y}\mathcal{E}(\rho)\big)$ gives the probability of obtaining outcome $y$ assuming the probabilistic preparation from $p$ and $\mathcal{P}$. 
The conditional probability $g_{xy}$ may be interpreted as the probability that $\rho_x$ was prepared given a measurement outcome of $y$. Furthermore, if we set $\sigma=\mathcal{E}(\rho)$ and let $\star$ be the Leifer--Spekkens state over time function, then $(q,\mathcal{M}^{\star}_{\sigma},\mathcal{E}_{\rho}^{\star},\mathcal{P}^{\star}_{p})$ is a PEM scenario and compositionality of the Petz recovery map implies that $g$ is the classical dynamics associated with the PEM scenario $(q,\mathcal{M}^{\star}_{\sigma},\mathcal{E}_{\rho}^{\star},\mathcal{P}^{\star}_{p})$, which may be viewed as an inferential time-reverse of the PEM scenario $(p,\mathcal{P},\mathcal{E},\mathcal{M})$~
\footnote{This is not to be confused with the \emph{operational} time-reverse in the sense defined by M.\ Leifer and M.\ Pusey~\cite{LePu17} because our `control' and `uncontrolled' variables have switched}.
This makes mathematically precise a sense in which measurement and preparation are time-reversals of one another, even when there is non-unitary evolution between the prepared state and the measured state.

\begin{figure}
    \centering    
    \begin{tikzpicture}[scale=0.75]
    \node (CX1) at (-4,0) {$\C^{X}$};
    \node (A1) at (-4,2) {$\matr_{n}$};
    \node (CY1) at (-2,0) {$\C^{Y}$};
    \node (B1) at (-2,2) {$\matr_{m}$};
    \draw[->] (CX1) to node[left]{$\mathcal{P}$} (A1);
    \draw[->] (A1) to node[above]{$\mathcal{E}$} (B1);
    \draw[->] (B1) to node[right]{$\mathcal{M}$} (CY1);
    \draw[->] (CX1) to node[below]{$f$} (CY1);
    \node (CX3) at (0,0) {$p$};
    \node (A3) at (0,2) {$\rho$};
    \node (CY3) at (2,0) {$q$};
    \node (B3) at (2,2) {$\sigma$};
    \draw[|->] (CX3) to node[left]{$\mathcal{P}$} (A3);
    \draw[|->] (A3) to node[above]{$\mathcal{E}$} (B3);
    \draw[|->] (B3) to node[right]{$\mathcal{M}$} (CY3);
    \node (CX2) at (4,0) {$\C^{X}$};
    \node (A2) at (4,2) {$\matr_{n}$};
    \node (CY2) at (6,0) {$\C^{Y}$};
    \node (B2) at (6,2) {$\matr_{m}$};
    \draw[->] (A2) to node[left]{$\mathcal{P}^{\star}_{p}$} (CX2);
    \draw[->] (B2) to node[above]{$\mathcal{E}^{\star}_{\rho}$} (A2);
    \draw[->] (CY2) to node[right]{$\mathcal{M}^{\star}_{\sigma}$} (B2);
    \draw[->] (CY2) to node[below]{$g$} (CX2);
    \end{tikzpicture}
    \caption{Given a PEM scenario $(p,\mathcal{P},\mathcal{E},\mathcal{M})$, we can use the latter three maps to define a classical channel $f:\C^{X}\to\C^{Y}$ via $f=\mathcal{M}\circ\mathcal{E}\circ\mathcal{P}$. This equality is expressed by saying that the diagram on the left \emph{commutes}%
    \footnote{This terminology of a \emph{commuting diagram} should not be confused with commutativity of a set of operators.}.
    The diagram in the middle depicts the evolution of the state $p$ along preparation to $\rho$, evolution to $\sigma$, and measurement to $q$.
    One can then use the probability $p$ and the classical channel $f$ to define the (classical) Bayesian inverse $g:\C^{Y}\to\C^{X}$. However, one can compute \emph{another} map $\C^{Y}\to\C^{X}$ via $\mathcal{P}^{\star}_{p}\circ\mathcal{E}^{\star}_{\rho}\circ\mathcal{M}^{\star}_{\sigma}$ by using the Leifer--Spekkens state over time to Bayesian invert each of the pairs $(\mathcal{M},\sigma)$, $(\mathcal{E},\rho)$, and $(\mathcal{P},p)$, where $\sigma=\mathcal{E}(\rho)$ and $\rho=\mathcal{P}(p)$, to arrive at the CPTP maps $\mathcal{M}^{\star}_{\sigma}$, $\mathcal{E}^{\star}_{\rho}$, and $\mathcal{P}^{\star}_{p}$ (note that the Bayesian inverse of a preparation is a measurement and vice-versa, so that $(q,\mathcal{M}^{\star}_{\sigma},\mathcal{E}_{\rho}^{\star},\mathcal{P}^{\star}_{p})$ defines another PEM scenario). The two stochastic channels $g$ and $\mathcal{P}^{\star}_{p}\circ\mathcal{E}^{\star}_{\rho}\circ\mathcal{M}^{\star}_{\sigma}$, are equal, i.e., the diagram on the right commutes, because the Petz recovery map is compositional. This compositionality can also be viewed a quantum generalization of Jeffrey's probability kinematics~\cite{Je90,PaBu22,BSS22}.}
    \label{fig:clasdyn}
\end{figure}

In the special case where $X=\{1,\dots,n\}$, $Y=\{1,\dots,m\}$, and the sets $\{|i\>\<i|\equiv P_{i}:=\mathcal{P}(\delta_{i})\}_{i\in X}$ and $\{|k\>\<k|\equiv M_{k}\}_{k\in Y}$  are chosen to be the orthogonal rank 1 projections associated with spectral decompositions of $\rho$ and $\sigma$, respectively, then the Bayesian inverses $\mathcal{P}^{\star}_{p}$ and $\mathcal{M}^{\star}_{\sigma}$ are \emph{independent} of the state over time function provided that it satisfies the classical limit axiom. Indeed, the Bayesian inverses are simply given by the Hilbert--Schmidt adjoints in this case, so that $\mathcal{M}^{\star}_{\sigma}=\mathcal{M}^*$ and $\mathcal{P}^{\star}_{p}=\mathcal{P}^*$. However, since $\mathcal{E}$ is not assumed to have a particularly simple structure with respect to these eigenstates, its Hilbert--Schmidt adjoint is not a suitable `reverse' operation in general. By using the Leifer--Spekkens state over time function, we have found such a suitable `reverse' operation $\mathcal{E}^{\star}_{\rho}$ for which the conditional probabilities $\<i|\mathcal{E}^{\star}_{\rho}(M_{k})|i\>$ indeed satisfy Bayes rule (see Figure~\ref{fig:clasdyn}), namely,
\be
\<k|\mathcal{E}(P_{i})|k\>p_{i}=\<i|\mathcal{E}^{\star}_{\rho}(M_{k})|i\>q_{k}.
\ee
We note that this resolves an issue brought up by J.\ Barandes and D.\ Kagan from Ref.~\cite{BaKa21}, where it is stated that the classical Bayes' theorem in this context does not hold in general due to the ``generic irreversibility'' of $\mathcal{E}$. So not only does our notion of Bayesian inversion restore the classical Bayes' theorem in such a context by providing a suitable time-reversal of $\mathcal{E}$, it does so for any PEM scenario, i.e., for arbitrary preparations, evolutions and measurements. 

An interesting result in this context, which is a reformulation in the language of states over time of a result first proved by Leifer \cite{Le06}, is that the classical state over time $f\star p$ associated with the classical dynamics $f$ of the PEM scenario $(p,\mathcal{P},\mathcal{E},\mathcal{M})$ may be obtained from the \emph{quantum} state over time $\mathcal{E}\star \rho$. Indeed, one may show that for all $(x,y)\in X\times Y$, 
\be
(f\star p)(x,y)=\text{tr}\big((N_x\otimes M_y)(\mathcal{E}\star\rho)\big),
\ee
where $N_x$ and $M_y$ are the positive operators corresponding to the POVMs $\mathcal{N}:=\mathcal{P}^{\star}_{p}$ and $\mathcal{M}$, respectively.

Finally, we briefly mention that the quantum Bayes' rule of R.\ Schack, T.\ Brun, and C.\ Caves~\cite{SBC01} can be viewed as a PEM scenario, when one appropriately extends our definitions to  $C^*$-algebras. The need for this extension to $C^*$-algebras is because the probability $p$ is replaced by a prior (Radon) probability measure on $X=\mathcal{S}(\mA)$, the compact Hausdorff space of states in a matrix algebra $\mA$~\cite{FuJa13,Pa17}. Namely, $\C^{X}$ is  replaced with $C(X)$, the $C^*$-algebra of continuous $\C$-valued functions on $X$, while the preparation $\mathcal{P}:C(X)\to\mA$ is replaced by a canonical classical-to-quantum generalization of the \emph{sampling map} from classical statistics~\cite{FGPR20,PaBayes}. Once a measurement outcome is obtained, the prior can then be updated. Such a scenario appears in the context of adaptive strategies for optimal quantum state determination~\cite{Jo90,Jo12}. The details of this will be expounded upon elsewhere.

\subsection{Instruments, measurement, and the state-update rule}

If $\sigma\in\mA$ is an initial quantum state that undergoes measurement associated with some instrument $\mathcal{F}:\mA\to\mB\otimes\C^{X}$ and outcome $x\in X$ is measured, then the \emph{state-update rule} dictates that the quantum state of the system becomes~\cite{HeKr69,DaLe70,Kr71,Kr83,Oz85,Oz97,Oz98,Fu01} 
\be
\sigma\mapsto\frac{\mathcal{F}_{x}(\sigma)}{\tr\big(\mathcal{F}_{x}(\sigma)\big)}
\ee
\emph{after} the measurement has been performed.
Here, $\mathcal{F}_{x}:\mA\to\mB$ is the CP map of the instrument associated with outcome $x$, so that the sum $\sum_{x\in X}\mathcal{F}_{x}$ defines a CPTP map from $\mA$ to $\mB$. It is often the case that $\mathcal{F}_{x}$ is assumed to be a Kraus rank 1 CP map, i.e., it can be written in the form $\mathcal{F}_{x}=\Ad_{V_{x}}$. The special case where the $V_{x}$ form projection operators gives the L\"uders--von~Neumann measurement~\cite{vN18,Lu51,Lu06}.  

This state-update rule has often been called a quantum generalization of Bayes' rule~\cite{Bu77,Oz97,Ko99,Ko01,Te12,Va18,FGS22}. 
Some have also argued that the state-update rule is \emph{not} a quantum generalization Bayesian conditioning, and that it is more a combination of belief propagation and conditioning on a measurement outcome~\cite[Sections~V.A.2.\ and V.B.]{LeSp13}. 
In the present section, we illustrate how one is inevitably led to the state-update rule from Bayesian inverses associated with \emph{any} state over time function that satisfies the classical limit axiom. In other words, our work supports the idea that the state-update rule is indeed a quantum generalization of Bayes' rule. Nevertheless, we still agree with Ref.~\cite{LeSp13} in the statement that this does not produce a \emph{retrodictive} state, since the state-update rule does not tell us what the state of the quantum system might have been \emph{before} the act of measurement. Furthermore, we will see that the state-update rule is indeed conditioning on a measurement, where the concept of conditioning is made precise in terms of operator-algebraic conditional expectations~\cite{Um62,Tak72,Oz85}.

To precisely state our result, we introduce additional notation besides what was specified above. First, set $\mathcal{E}:\mB\otimes\C^{X}\to\C^{X}$ to be the partial trace map that traces out the $\mB$ system, i.e., $\mathcal{E}=\tr_{\mB}\equiv\tr\otimes\id_{\C^{X}}$. Let $\rho\in\mB\otimes\C^{X}$ be any state, which we decompose as 
\be
\rho=\sum_{x\in X}\rho_{x}\otimes\delta_{x},
\ee
where $\rho_{x}\in\mB$ is some positive element for all $x\in X$, and $\delta_{x}$ is the unit vector whose $x$ component is $1$ and other components are $0$. 

\begin{proposition}
\label{thm:subi}
Let $\star$ be any state over time function satisfying the classical limit axiom and assume that $\rho_{x}\ne0$ for all $x\in X$. Then, using the notation of the previous paragraph, a Bayesian inverse of $(\mathcal{E},\rho)$ with respect to $\star$ exists, is unique, and is given by 
\be
\mathcal{E}^{\star}_{\rho}(\delta_{x})=\frac{\rho_{x}}{\tr(\rho_{x})}\otimes\delta_{x}
\ee
for all $x\in X$.
\end{proposition}

We first prove this result before discussing the relevance to the state-change associated with a measurement. 

\begin{proof}
Explicitly computing $\mathscr{D}[\mathcal{E}]\in(\mB\otimes\C^{X})\otimes\C^{X}$ (the parentheses is used to distinguish the input and output algebras) gives
\be
\mathscr{D}[\mathcal{E}]=\sum_{x\in X}1_{\mB}\otimes\delta_{x}\otimes\delta_{x}.
\ee
Therefore, $[\rho\otimes1_{\C^{X}},\mathscr{D}[\mathcal{E}]]=0$, so that 
\be
\mathcal{E}\star\rho
=(\rho\otimes1_{\C^{X}})\mathscr{D}[\mathcal{E}]
=\sum_{x\in X}\rho_{x}\otimes\delta_{x}\otimes\delta_{x}
\ee
and 
\be
\tau(\mathcal{E}\star\rho)=\sum_{x\in X}\delta_{x}\otimes\rho_{x}\otimes\delta_{x}.
\ee
Furthermore, since $\mathcal{E}(\rho)\in\C^{X}$ lives in a commutative $C^*$-algebra, $\mathcal{E}(\rho)\otimes1_{\mB\otimes\C^{X}}$  necessarily commutes with every element of $\C^{X}\otimes\mB\otimes\C^{X}$. In particular, $\big[\mathcal{E}(\rho)\otimes1_{\mB\otimes\C^{X}},\mathscr{D}[\mathcal{E}^{\star}_{\rho}]\big]=0$, so that 
\begin{align}
\mathcal{E}^{\star}_{\rho}\star\mathcal{E}(\rho)
&=\big(\mathcal{E}(\rho)\otimes1_{\mB}\otimes1_{\C^{X}}\big)\mathscr{D}[\mathcal{E}^{\star}_{\rho}] \nonumber\\
&=\left(\sum_{x\in X}\tr(\rho_{x})\delta_{x}\otimes1_{\mB}\otimes1_{\C^{X}}\right)\mathscr{D}[\mathcal{E}^{\star}_{\rho}].
\end{align}
Bayes' rule $\tau(\mathcal{E}\star\rho)=\mathcal{E}^{\star}_{\rho}\star\mathcal{E}(\rho)$ then yields 
\be
\mathscr{D}[\mathcal{E}^{\star}_{\rho}]=\sum_{x\in X}\delta_{x}\otimes\frac{\rho_{x}}{\tr(\rho_{x})}\otimes\delta_{x},
\ee
so that the Bayesian inverse $\mathcal{E}^{\star}_{\rho}$ is given by the map that sends $\delta_{x}$, representing the outcome $x$, to 
\be
\mathcal{E}^{\star}_{\rho}(\delta_{x})=\frac{\rho_{x}}{\tr(\rho_{x})}\otimes\delta_{x}. 
\ee
\end{proof}

Using this result, we can relate the state-update rule to state-preserving conditional expectations and Bayesian inverses associated with a state over time function satisfying the classical limit axiom. 

\begin{theorem}
Let $\sigma\in\mA$ be a state, $\mathcal{F}:\mA\to\mB\otimes\C^{X}$ an instrument, $\mathcal{E}:\mB\otimes\C^{X}\to\C^{X}$ the partial trace, and set $\rho=\mathcal{F}(\sigma)$. Suppose $\tr(\mathcal{F}_{x}(\sigma))\ne0$ for all $x\in X$. Furthermore, define the following maps.
\begin{enumerate}
\item
\label{item:su}
Let $\Psi:\C^{X}\to\mB\otimes\C^{X}$ be the state-update map defined as the unique linear extension of
\be
\Psi(\delta_{x})=\frac{\mathcal{F}_{x}(\sigma)}{\tr\big(\mathcal{F}_{x}(\sigma)\big)}\otimes\delta_{x}.
\ee
\item
\label{item:bi}
Let $\mathcal{E}^{\star}_{\rho}:\C^{X}\to\mB\otimes\C^{X}$ be the Bayesian inverse of $(\mathcal{E},\rho)$ associated with \emph{any} state over time function $\star$ satisfying the classical limit axiom (cf.\ Proposition~\ref{thm:subi}). 
\item
\label{item:spcd}
Let $\Omega:\C^{X}\to\mB\otimes\C^{X}$ be the Hilbert--Schmidt adjoint of a state-preserving conditional expectation associated with the map $\mathcal{E}$ and state $\rho$, i.e., $\Omega(\mathcal{E}(\rho))=\rho$ and $\mathcal{E}\circ\Omega=\id_{\C^{X}}$ (cf.\ Ref.~\cite{Al95}). 
\end{enumerate}
Then 
$\Omega=\Psi=\mathcal{E}^{\star}_{\rho}$.
\end{theorem}

\begin{proof}
The equivalence between items~\ref{item:su} and~\ref{item:bi}, and hence the equality $\Psi=\mathcal{E}^{\star}_{\rho}$, follows from Proposition~\ref{thm:subi} by setting $\rho_{x}=\mathcal{F}_{x}(\sigma)$.
As for the relationship to state-preserving conditional expectations, first note that 
\begin{align}
\Psi\big(\mathcal{E}(\rho)\big)
&=\Psi\left(\sum_{x}\tr(\rho_{x})\delta_{x}\right)
=\sum_{x\in X}\tr(\rho_{x})\Psi(\delta_{x}) \nonumber\\
&=\sum_{x\in X}\tr(\rho_{x})\left(\frac{\rho_{x}}{\tr(\rho_{x})}\otimes\delta_{x}\right)
=\sum_{x}\rho_{x}\otimes\delta_{x} \nonumber \\
&=\rho.
\end{align}
Second, we have 
\be
(\mathcal{E}\circ\Psi)(\delta_{x})
=\mathcal{E}\left(\frac{\rho_{x}}{\tr(\rho_{x})}\otimes\delta_{x}\right)
=\frac{\tr(\rho_{x})}{\tr(\rho_{x})}\delta_{x}
=\delta_{x}
\ee
for each $x\in X$. By linear extension, this shows $\mathcal{E}\circ\Psi=\id_{\C^{X}}$. Thus, $\Psi$ satisfies the definition of a state-preserving conditional expectation. By the uniqueness of state-preserving conditional expectations for faithful states~\cite{Tak72,OhPe93,PaRu19}, $\Omega=\Psi=\mathcal{E}^{\star}_{\rho}$.
\end{proof}

Note that by definition of $\Psi$ by extending it linearly, if one acquires soft evidence from the measurement in the form of a probability distribution $r\in\C^{X}$, then 
\be
\Psi(r)=\sum_{x\in X}\frac{r_{x}\mathcal{F}_{x}(\sigma)}{\tr\big(\mathcal{F}_{x}(\sigma)\big)}\otimes\delta_{x}.
\ee
Note that the $\otimes\delta_{x}$ term is merely used as a book-keeping device to separate the possible state-updates depending on each outcome.
Namely, by tracing out over $\C^{X}$, one obtains the barycenter 
\be
(\id_{\mB}\otimes\tr_{\C^{X}})\big(\Psi(r)\big)=\sum_{x\in X}\frac{r_{x}\mathcal{F}_{x}(\sigma)}{\tr\big(\mathcal{F}_{x}(\sigma)\big)}
\ee
as the updated density matrix in $\mB$ after the measurement has been performed and the outcome is only given by soft evidence $r$. This is a quantum generalization of Jeffrey's rule~\cite{Je90,Ja19}. 
Nevertheless, this is a special case of Bayesian inverses from our perspective due to the hybrid classical/quantum nature of the channels involved. The full quantum generalization of Bayes' rule is the main definition provided in this work. The point here is that the Bayesian inverse for \emph{every} state over time function satisfying the classical limit axiom \emph{must} reproduce this special version of quantum Bayes' rule.

\section{Non-positive Bayesian inverses}

To incorporate even more instances of Bayes' rule, we will extend state over time functions to a larger domain. One reason to do this is because there may exist solutions $\mathcal{E}^{\star}_{\rho}$ to Bayes' rule $\mathcal{E}\star\rho=\tau(\mathcal{E}^{\star}_{\rho}\star\mathcal{E}(\rho))$ that are not necessarily completely positive, and these solutions may nevertheless have useful applications and/or interpretations. A second reason is to incorporate a wider variety of examples in the literature as instances of states over time and Bayes' rule. These include the two-state formalism~\cite{Wat55,ABL64,ReAh95}, time-dependent correlators~\cite{BDOV13,BHR17}, generalized conditional expectations~\cite{Ts22}, and many more. 

We will first consider a solution $\mathcal{E}^{\star}_{\rho}$ to Bayes' rule for a state over time function where $\mathcal{E}^{\star}_{\rho}$ is not necessarily CP, but is \emph{$\dag$-preserving}, in the sense that $\mathcal{E}^{\star}_{\rho}(B^{\dag})=(\mathcal{E}^{\star}_{\rho}(B))^{\dag}$ for all inputs $B$. In this case, we allow more flexibility in our earlier definition of a state over time function by requesting that it is a family of functions each of the form $\star:\HPTP(\mA,\mB)\times\mathcal{S}(\mA)\to\mA\otimes\mB$, where $\HPTP(\mA,\mB)$ denotes the set of $\dag$-preserving and trace-preserving (HPTP) maps from $\mA$ to $\mB$. 
An HPTP solution $\mathcal{E}^{\star}_{\rho}$ to Bayes' rule will be called a \emph{Bayes' map}.
Our primary example will be the Jordan product state over time function.

\subsection{The Jordan product state over time}
\label{sec:SBsot}

The \define{symmetric bloom}/\define{Jordan product} state over time function is defined on $(\mathcal{E},\rho)\in\HPTP(\mA,\mB)\times\mathcal{S}(\mA)$ by 
\be
\mathcal{E}\star\rho=\frac{1}{2}\big\{\rho\otimes 1_{\mB},\mathscr{D}[\mathcal{E}]\big\},
\ee
where $\{X,Y\}=XY+YX$ denotes the Jordan product (anti-commutator).
The symmetric bloom state over time function is hermitian, bilinear, associative, and satisfies the classical limit axiom (these statements are all proved in Ref.~\cite{FuPa22}). The symmetric bloom was recently shown to define a quantum state over time function that bypasses the no-go theorem of Ref.~\cite{HHPBS17}, which claimed the non-existence of a state over time function satisfying these properties. While the no-go theorem of Ref.~\cite{HHPBS17} is mathematically correct, it assumes a larger domain for state over time functions than what is physically necessary. The symmetric bloom then by-passes the statement in Ref.~\cite{HHPBS17} by being defined on a smaller, more physically relevant, domain. 

A Bayes map $\mathcal{E}^{\star}_{\rho}$ for the symmetric bloom must satisfy the equation 
\be
\tau\Big(\big\{\rho\otimes1_{\mB},\mathscr{D}[\mathcal{E}]\big\}\Big)
=\big\{\mathcal{E}(\rho)\otimes1_{\mA},\mathscr{D}[\mathcal{E}^{\star}_{\rho}]\big\},
\ee
which is equivalent to 
\be
\big\{1_{\mB}\otimes\rho,\mathscr{D}[\mathcal{E}^*]\big\}=\big\{\mathcal{E}(\rho)\otimes1_{\mA},\mathscr{D}[\mathcal{E}^{\star}_{\rho}]\big\}
\ee
because $\rho$ is self-adjoint and $\mathcal{E}$ is CP.
This is a linear algebra problem of the form $B=\{A,X\},$ where $A$ and $B$ are known and $X$ is desired.

To avoid cumbersome indices, we now restrict to the case where $\mA=\matr_{m}$ and $\mB=\matr_{n}$ are matrix algebras. In addition, to avoid a discussion of measure zero and non-commutative almost everywhere equivalence~\cite{PaRu19}, we will restrict our attention to the cases where $\mathcal{E}(\rho)$ is non-singular (has only non-zero eigenvalues). 
Let 
\be
\mathcal{E}(\rho)=\sum_{k=1}^{n}q_{k}|w_{k}\>\<w_{k}|
\ee
be a spectral decomposition into one-dimensional projections (so that some of the $q_{k}$ may repeat).

Using this and the above equation in terms of the Jordan product, one obtains
\be
\mathcal{E}^{\star}_{\rho}\big(|w_{k}\>\<w_{l}|\big)=(q_{k}+q_{l})^{-1}\Big\{\rho,\mathcal{E}^*\big(|w_{k}\>\<w_{l}|\big)\Big\}.
\ee
Since the $\{|w_{k}\>\}$ form a basis, this determines $\mathcal{E}^{\star}_{\rho}$ as a linear operator. In fact, $\mathcal{E}^{\star}_{\rho}$ is $\dag$-preserving because $\rho$ is self-adjoint, the $\{q_{k}\}$ are real, the anti-commutator is $\dag$-preserving, and $\mathcal{E}$ is $\dag$-preserving. It is not presently known to us what necessary and sufficient conditions guarantee complete positivity of $\mathcal{E}^{\star}_{\rho}$ in the case of the symmetric bloom.

\subsubsection{The symmetric bloom approximates Leifer--Spekkens}

Interestingly, the symmetric bloom state over time provides the linear approximation to the Leifer--Spekkens state over time near the maximally mixed state~\cite{FuPa22}. Let $\rho_{0}:=\frac{1_{\mA}}{\tr(1_{\mA})}$ denote the uniform state in a multi-matrix algebra $\mA$, let $\mathcal{E}:\mA\to\mB$ be a CPTP map, and let $\star_{\LS}$ and $\star_{\SB}$ denote the Leifer--Spekkens and symmetric bloom state over time functions, respectively. Then for any $A\in\mA^{\sa}$ such that $\tr(A)=0$, one has
\be
\mathcal{E}\star_{\LS}(\rho_{0}+\epsilon A)-\mathcal{E}\star_{\LS}\rho_{0}=\epsilon(\mathcal{E}\star_{\SB}A)+\mathcal{O}(\epsilon^2)
\ee
for sufficiently small $\epsilon$. 
Note that we have linearly extended $\star_{\mathrm{J}}$ from $\mathcal{S}(\mA)$ to $\mA^{\sa}$ in the second argument to make sense of the right-hand-side of this identity.
Equivalently, 
\be
\lim_{\epsilon\to0}\left(\frac{\mathcal{E}\star_{\LS}(\rho_{0}+\epsilon A)-\mathcal{E}\star_{\LS}\rho_{0}}{\epsilon}\right)=\mathcal{E}\star_{\SB}A.
\ee

This, combined with the work of Ref.~\cite{FuPa22}, shows that the symmetric bloom state over time maintains some of the physically relevant features of the Leifer--Spekkens state over time, but in addition satisfies a larger number of convenient properties, including associativity. At present, we do not know what this says about the relationship between the Petz recovery map and the Bayes map for the symmetric bloom.

\subsubsection{An operational interpretation of the symmetric bloom}
\label{sec:opsymb}

In Refs.~\cite{BDOV13,BDOV14}, the symmetric bloom state over time, in the special case where the channel is the identity, appears as the real part of the ideal two-point quantum correlator (see the later section on the right bloom for more details). The symmetric bloom state over time, when viewed as a function of its second argument, defines a map $\id_{\mA}\star\;\cdot\;:\mathcal{S}(\mA)\to\mA\otimes\mA$. This map is given an operational interpretation by decomposing it via a generalized Kraus decomposition into a quantum instrument, which in turn allows one to compute expectation values associated with the input state and instrument~\cite[Proposition~1]{BDOV13}. A unique such decomposition into a quantum instrument can be chosen by specifying that the sum of the absolute values of the associated statistical errors is minimized~\cite[Proposition~2]{BDOV13}. In this special case of the symmetric bloom, this unique decomposition is expressed in terms of symmetric and anti-symmetric optimal cloners~\cite[Proposition~3]{BDOV13}. In fact, Ref.~\cite{BDOV13} proposed an optical experiment with polarized photons to operationally determine the symmetric bloom.

\subsection{Bayes' rule and linear Bayes maps}

We may even go beyond $\dag$-preserving maps by finding \emph{linear} solutions to Bayes' rule. To make sense of this, we therefore need to extend state over time functions to be defineable on arbitrary trace-preserving linear maps. We will also see that the quantum time-reversal symmetry $\tau$ for states over time, given by the swap map $\gamma$ and adjoint $\dag$, is not enough for a robust Bayes' rule in this more general setting. 
In all that follows, let $\TP(\mA,\mB)$ denote the space of trace-preserving  linear maps from $\mA$ to $\mB$. 

\begin{definition}
\label{defn:ESOT}
An \define{extended state over time function} associates every pair $(\mA,\mB)$ of multi-matrix algebras with a map $\star:\TP(\mA,\mB)\times\mathcal{S}(\mA)\to\mA\otimes\mB$, 
whose value on $(\mathcal{E},\rho)$ is denoted by $\mathcal{E}\star\rho$,
such that 
$\star$ preserves marginals in the sense that 
\be
\tr_{\mB}\left(\mathcal{E}\star\rho\right)=\rho
\quad\text{and}\quad
\tr_{\mA}\left(\mathcal{E}\star\rho\right)=\mathcal{E}(\rho).
\ee
In such a case, the element $\mathcal{E}\star\rho\in \mA\otimes\mB$ 
will be referred to as the \define{state over time} associated with $\star$ and the input $(\mathcal{E},\rho)$.
\end{definition}

As before, the word `state' is abusive, since the element $\mathcal{E}\star\rho$ need not be a state. 
One can extend the properties from Definition~\ref{defn:saxs} to extended state over time functions as follows: Property~\ref{item:hermitian} allows $\mathcal{E}\in\HPTP(\mA,\mB)$, Properties~\ref{item:slin} and \ref{item:classicallimit} allow $\mathcal{E}\in\TP(\mA,\mB)$, Property~\ref{item:plin} allows $\lambda\in\C$ and $\mathcal{E},\mathcal{F}\in\TP(\mA,\mB)$, and Properties~\ref{item:locpos}, \ref{item:pos}, and \ref{item:bilin} are the same as originally stated.
For brevity, we will also henceforth drop the word `extended' from our phrasing unless emphasis is needed. 

\begin{definition}
\label{defn:Bayesmap}
Let $\star$ be a state over time function. Given a density matrix $\rho\in\St(\mA)$ and a CPTP map $\mathcal{E}:\mA\to\mB$, 
a \define{Bayes map} associated with $(\mathcal{E},\rho)$ is a trace-preserving linear map $\mathcal{E}^{\star}_{\rho}:\mB\to\mA$ such that
\be
\mathcal{E}\star\rho
=\tau\left(\widetilde{\mathcal{E}^{\star}_{\rho}}\star\mathcal{E}(\rho)\right),
\ee
where $\widetilde{\mathcal{E}^{\star}_{\rho}}:=\dag\circ\mathcal{E}^{\star}_{\rho}\circ\dag$ is defined by~%
\footnote{
One could equivalently define a new state over time function $\star^{\dag}:\TP(\mA,\mB)\times\mathcal{S}(\mA)\to\mA\otimes\mB$ whose value on $(\mathcal{E},\rho)$ is given by 
$\mathcal{E}\star^{\dag}\rho:=\big((\dag\circ\mathcal{E}\circ\dag)\star\rho\big)^{\dag}$. In this case, we could call $\star^{\dag}$ the \define{reverse orientation} state over time function associated with $\star$. Then, an equivalent formulation of Bayes' rule reads 
$\mathcal{E}\star\rho=\gamma\big(\mathcal{E}^{\star}_{\rho}\star^{\dag}\mathcal{E}(\rho)\big)$, where only the swap map $\gamma$ is used, but with the additional modification of using the reverse orientation state over time on the right. 
}
\be
\widetilde{\mathcal{E}^{\star}_{\rho}}(B):=\big(\mathcal{E}^{\star}_{\rho}(B^{\dag})\big)^{\dag}.
\ee
\end{definition}

Note that when $\mathcal{E}^{\star}_{\rho}$ is $\dag$-preserving, then $\widetilde{\mathcal{E}^{\star}_{\rho}}=\mathcal{E}^{\star}_{\rho}$. We have already seen a special case of this when discussing the symmetric bloom. 

\begin{remark}
\label{rmk:marginB}
Note that by taking the partial trace $\tr_{\mB}$ of both sides of our generalized Bayes' rule, we find
\begin{align}
\rho&=\tr_{\mB}\big(\tau(\widetilde{\mathcal{E}^{\star}_{\rho}}\star\mathcal{E}(\rho)\big)
=\Big(\widetilde{\mathcal{E}^{\star}_{\rho}}\big(\mathcal{E}(\rho)\big)\Big)^{\dag} \nonumber\\
&=\mathcal{E}^{\star}_{\rho}\big(\mathcal{E}(\rho)^{\dag}\big)=\mathcal{E}^{\star}_{\rho}\big(\mathcal{E}(\rho)\big),
\end{align}
so that $\mathcal{E}^{\star}_{\rho}$ takes the prediction $\mathcal{E}(\rho)$ back to the prior $\rho$. This would not necessarily be the case had we used $\mathcal{E}^{\star}_{\rho}$ instead of $\widetilde{\mathcal{E}^{\star}_{\rho}}$ in our definition of Bayes' rule. More reasons for the importance of using $\widetilde{\mathcal{E}^{\star}_{\rho}}$ as opposed to just $\mathcal{E}^{\star}_{\rho}$ in Bayes' rule will be provided in the upcoming examples.
\end{remark}

\subsection{The \texorpdfstring{$(r,s)$}{(r,s)}-parametrized family}
\label{sec:rsfamsot}

We will now simultaneously generalize both the Leifer--Spekkens and Jordan state over time functions by viewing them as special cases inside a family. 
For each $r,s\in[0,1]$ the assignment sending $(\mathcal{E},\rho)\in\TP(\mA,\mB)\times\mathcal{S}(\mA)$ to 
\be
s(\rho^{r}\otimes 1_{\mB})\mathscr{D}[\mathcal{E}](\rho^{1-r}\otimes 1_{\mB})+(1-s)(\rho^{1-r}\otimes 1_{\mB})\mathscr{D}[\mathcal{E}](\rho^{r}\otimes 1_{\mB})
\ee
defines a state over time function that is process-linear, and satisfies the classical limit axiom for all $r,s\in[0,1]$. 
Process-linearity follows from the linearity of the channel state in its argument. As for the classical limit axiom, 
note that by the functional calculus for matrices, 
if $[\rho\otimes1_{\mB},\mathscr{D}[\mathcal{E}]]=0$, 
then $[\rho^{r}\otimes1_{\mB},\mathscr{D}[\mathcal{E}]]=0$.
Temporarily introducing the notation $s^{\perp}:=1-s$ and $r^{\perp}:=1-r$, 
this 
implies 
\begin{align}
&s(\rho^{r}\otimes 1_{\mB})\mathscr{D}[\mathcal{E}](\rho^{r^{\perp}}\otimes 1_{\mB})+s^{\perp}(\rho^{r^{\perp}}\otimes 1_{\mB})\mathscr{D}[\mathcal{E}](\rho^{r}\otimes 1_{\mB}) \nonumber\\
&=s(\rho\otimes1_{\mB})\mathscr{D}[\mathcal{E}]+(1-s)(\rho\otimes1_{\mB})\mathscr{D}[\mathcal{E}] \nonumber\\
&=(\rho\otimes1_{\mB})\mathscr{D}[\mathcal{E}], 
\end{align}
so that the classical limit axiom holds. 

This parametrized family specializes to the Leifer--Spekkens state over time for $r=\frac{1}{2}$ and $s$ arbitrary, the symmetric bloom for $r\in\{0,1\}$ and $s=\frac{1}{2}$, and many other cases that have appeared in the literature. These include the \emph{left} and \emph{right blooms}, which specialize to the two-state formalism and time-dependent correlators, as will be described next.

\subsection{The right bloom}
\label{sec:rbsot}

The cases $(r,s)=(1,1)$ or $(r,s)=(0,0)$ yield $\mathcal{E}\star\rho=(\rho\otimes1_{\mB})\mathscr{D}[\mathcal{E}]$, the \emph{bloom} from Ref.~\cite{FuPa22}, which will be referred to as the \emph{right bloom} here~
\footnote{Using notation from Ref.~\cite{FuPa22}, this is because the functional associated with $\mathcal{E}\star\rho$ sends $A\otimes B\in\mA\otimes\mB$ to \unexpanded{
\[
\<\mathcal{E}\star\rho,A\otimes B\>=\tr\big((\mathcal{E}\star\rho)^{\dag}(A\otimes B)\big)=\tr(\rho A\mathcal{E}^*(B)\big),
\]
}
which agrees with $\omega\big(A F(B)\big)$ upon setting $F=\mathcal{E}^{*}$ and $\rho=\mathscr{D}[\omega]$.
}. 
The right bloom is also state-linear, so that it is in fact bilinear, and associative.

A Bayes map $\mathcal{E}^{\star}_{\rho}$ for the right bloom must satisfy the equation 
\be
\mathcal{E}\star\rho=\gamma\Big(\big((\dag\circ\mathcal{E}^{\star}_{\rho}\circ\dag)\star\mathcal{E}(\rho)\big)^{\dag}\Big),
\ee
which by 
applying $\dag$ to both sides and using the fact that $\gamma$ is $\dag$-preserving gives
\begin{align}
\mathscr{D}[\mathcal{E}](\rho\otimes1_{\mB})&=\gamma\Big(\big(\mathcal{E}(\rho)\otimes1_{\mA}\big)\mathscr{D}[\dag\circ\mathcal{E}^{\star}_{\rho}\circ\dag]\Big) \nonumber \\
&=\big(1_{\mA}\otimes\mathcal{E}(\rho)\big)\mathscr{D}\big[(\mathcal{E}^{\star}_{\rho})^{*}\big]
\end{align}
by Lemma~\ref{lem:EgamDE}. If we assume that $\mathcal{E}(\rho)$ is strictly positive, this gives the unique solution
\begin{align}
\mathscr{D}\big[\mathcal{E}^{\star *}_{\rho}\big]
&=\big(1_{\mA}\otimes\mathcal{E}(\rho)^{-1}\big)\mathscr{D}[\mathcal{E}](\rho\otimes1_{\mB})\nonumber \\
&=\mathscr{D}\big[\mathscr{L}_{\mathcal{E}(\rho)^{-1}}\circ\mathcal{E}\circ\mathscr{L}_{\rho}\big]
\end{align}
by Lemma~\ref{lem:EgamDE} (as in Lemma~\ref{lem:gamDE}, $\mathscr{L}$ denotes left multiplication). 

Since the channel state assignment $\mathscr{D}$ is a linear isomorphism, the inputs are equal, i.e., $\mathcal{E}^{\star *}_{\rho}=\mathscr{L}_{\mathcal{E}(\rho)^{-1}}\circ\mathcal{E}\circ\mathscr{L}_{\rho}$. Taking the Hilbert--Schmidt adjoint therefore gives the solution 
\be
\mathcal{E}^{\star}_{\rho}
=\mathscr{L}_{\rho}^*\circ\mathcal{E}^*\circ\mathscr{L}_{\mathcal{E}(\rho)^{-1}}^{*}
=\mathscr{L}_{\rho}\circ\mathcal{E}^*\circ\mathscr{L}_{\mathcal{E}(\rho)^{-1}}.
\ee
where we have used the general fact that $\mathscr{L}_{A}^*=\mathscr{L}_{A^{\dag}}$. Explicitly, this means $\mathcal{E}^{\star}_{\rho}$ is given by the formula
\begin{equation}
\label{eq:LBBM}
\mB\ni B\xmapsto{\mathcal{E}^{\star}_{\rho}}\rho\mathcal{E}^*\big(\mathcal{E}(\rho)^{-1}B\big),
\end{equation}
which agrees with the Bayes map of Refs.~\cite{PaRuBayes,PaBayes,PaQPL21,GPRR21}. 

We make two important remarks regarding the right bloom state over time function and the associated Bayes map.

\begin{remark}
\label{rmk:LBInv}
If $\mathcal{E}$ is a $*$-isomorphism (equivalently, $\mathcal{E}$ is an invertible quantum channel, in the sense that $\mathcal{E}^{-1}$ exists and is also a quantum channel), then the Bayes map satisfies $\mathcal{E}^{\star}_{\rho}=\mathcal{E}^{-1}$. This property would generally \emph{fail} if we had instead defined Bayes' rule naively as $\mathcal{E}\star\rho
=\gamma\left(\mathcal{E}^{\star}_{\rho}\star\mathcal{E}(\rho)\right)$ without using $\tau$ on the right-hand-side. 

Indeed, if $\mathcal{E}=\Ad_{U}$ is invertible and represented by a unitary $U$, then the Bayes map given by Equation~\eqref{eq:LBBM} yields
\begin{align}
\mathcal{E}^{\star}_{\rho}(B)
&=\rho\Ad_{U^{\dag}}\Big(\big(\Ad_{U}(\rho)\big)^{-1}B\Big)
=\rho U^{\dag}\big(U\rho U^{\dag}\big)^{-1}B U \nonumber \\
&=\rho U^{\dag}(U^{\dag})^{-1}\rho^{-1}U^{-1}B U
=\rho\rho^{-1}U^{\dag}BU \nonumber\\
&=U^{\dag}BU
\end{align}
since $U^{-1}=U^{\dag}$. Thus, $\mathcal{E}^{\star}_{\rho}=\mathcal{E}^{-1}$.

However, if the alternative proposal $\mathcal{E}\star\rho
=\gamma\left(\mathcal{E}^{\star}_{\rho}\star\mathcal{E}(\rho)\right)$ for Bayes' rule had been used, then the unique linear solution would be given by $\mathcal{E}^{\star}_{\rho}(B)=\rho\mathcal{E}^{*}\big(B\mathcal{E}(\rho)^{-1}\big)$. Plugging in $\mathcal{E}=\Ad_{U}$ into \emph{this} expression would instead yield
\be
\mathcal{E}^{\star}_{\rho}(B)
=\rho\Ad_{U^{\dag}}\Big(B\big(\Ad_{U}(\rho)\big)^{-1}\Big)
=\rho U^{\dag}BU\rho^{-1},
\ee
which is not the inverse of $\mathcal{E}$ in general. 
We will see many other reasons supporting the usage of $\tau$ (and $\widetilde{\mathcal{E}^{\star}_{\rho}}$), as opposed to just $\gamma$, in our definition of Bayes' rule in later examples.
\end{remark}

\begin{remark}
\label{rmk:BaTh}
Although $\mathcal{E}^{\star}_{\rho}$ is not in general CP, \cite[Proposition~3.2]{GPRR21} (see also Ref.~\cite{PaRuBayes}) shows the following are equivalent 
\begin{enumerate}
\item
$\mathcal{E}^{\star}_{\rho}$ is $\dag$-preserving
\item
$\rho\mathcal{E}^*\big(\mathcal{E}(\rho)^{-1}B\big)=\mathcal{E}^*\big(B\mathcal{E}(\rho)^{-1}\big)\rho$
for all $B\in\mB$
\item
$\mathcal{E}^{\star}_{\rho}$ is CP
\item
$\Ad_{\mathcal{E}(\rho)^{it}}\circ\mathcal{E}=\mathcal{E}\circ\Ad_{\rho^{it}}$ for all $t\in\R$.
\end{enumerate}
In this case, one can rewrite the formula for $\mathcal{E}^{\star}_{\rho}$ as 
\be
\mathcal{E}^{\star}_{\rho}=\Ad_{\rho^{1/2}}\circ\mathcal{E}^{*}\circ\Ad_{\mathcal{E}(\rho)^{-1/2}},
\ee
which agrees with the Petz recovery map. 

Note that the last equivalent condition appeared originally in the work of Accardi and Cecchini in Ref.~\cite{AcCe82}, specifically in the context of GNS-symmetric dynamics and Tomita--Takesaki theory~\cite{GPRR21,BCM16,An06}. However, it also appears in the physics literature as time-symmetric covariant quantum channels~\cite{Ma12,LKJR15} (see also Ref.~\cite{PaBu22}), where the time evolution symmetry is generated by the modular Hamiltonians associated with the initial and final states. 

Interestingly, if we had ignored the dagger in our formulation of Bayes' rule, the covariance condition would have instead been $\Ad_{\mathcal{E}(\rho)^{it}}\circ\mathcal{E}=\mathcal{E}\circ\Ad_{\rho^{-it}}$ for all $t\in\R$. This would suggest that the modular flow goes forwards in time for one state but backwards in time for the other state, which seems to be at odds with the natural directionality of time suggested by the modular flow~\cite{CoRo94}.
\end{remark}

\subsubsection{Weak values and the two-state formalism}

In the special case where $\mA=\matr_{m}$, $\mB=\C^{X}$, and $\mE:\mA\to\mB$ describes a POVM with $x$ component $\mathcal{E}_{x}=\tr(M_{x}\;\cdot\;)$ for some positive operator $M_{x}$, the state over time associated with the right bloom generalizes the \emph{two-state} from Refs.~\cite{Wat55,ABL64,ReAh95}, which has also appeared recently in the context of holography~\cite{NTTTW21}. In this case, $\mA\otimes\mB\cong\bigoplus_{x\in X}\matr_{m}$ so that 
\begin{align}
\mathcal{E}\star\rho
&=(\rho\otimes1_{\C^{X}})\left(\sum_{i,j}E_{ij}^{(m)}\otimes\left(\sum_{x\in X}\tr\big(M_{x}E_{ji}^{(m)}\big)\delta_{x}\right)\right) \nonumber \\
&=\sum_{x\in X}(\rho\otimes1_{\C^{X}})\left(\sum_{i,j}\big((M_{x})_{ij}E_{ij}^{(m)}\big)\otimes\delta_{x}\right) \nonumber \\
&=\sum_{x\in X}(\rho M_{x})\otimes\delta_{x}\cong\bigoplus_{x\in X}\rho M_{x}.
\end{align}
In the special case where $\rho=|\psi\>\<\psi|$ is the one-dimensional projection corresponding to a pure state $|\psi\>$ and similarly for $M_{x}=|\phi_{x}\>\<\phi_{x}|$, this state over time becomes
\be
\mathcal{E}\star\rho\cong\bigoplus_{x\in X}\<\psi|\phi_{x}\> |\psi\>\<\phi_{x}|, 
\ee
which is the two-state appearing in Ref.~\cite{ReAh95} (combining all possibilities indexed by $x$ together with the overlap coefficient $\<\psi|\phi_{x}\>$) when the Hamiltonian evolution is trivial. This setup is similar to a PEM scenario in that a single state is prepared and the evolution is trivial. The main difference, however, is the choice of state over time function, which in this case yields the two-state and describes pre- and post-selection~\cite{ABL64}.

\begin{figure}
    \centering    
    \begin{tikzpicture}
    \draw[thick] (0,0) -- (0,4);
    \draw[thick] (-0.1,0) node[left]{$t_{0}$\;} -- (0.1,0);
    \draw[thick] (-0.1,2) node[left]{$t_{1}$\;} -- (0.1,2);
    \draw[thick] (-0.1,4) node[left]{$t_{2}$\;} -- (0.1,4);
    \node (t1L) at (-2,0) {$\mA_0$};
    \node (tL) at (-2,2) {$\mA_1$};
    \node (t2Lb) at (-2,3.5) {$\mA_2$};
    \node (t2Lt) at (-2,4.5) {$\mB$};
    \draw[->] (t1L) -- node[left]{$\Ad_{U_{t_{1}\leftarrow t_{0}}}$} (tL);
    \draw[->] (tL) -- node[left]{$\Ad_{U_{t_{2}\leftarrow t_{1}}}$} (t2Lb);
    \draw[->] (t2Lb) -- node[left]{$\mathcal{E}$} (t2Lt);
    \node (psi) at (2,0) {$|\psi\>$};
    \node (psit) at (2,2) {$|\psi'\>$};
    \node (phi) at (4,4) {$|\phi_{x}\>$};
    \node (phit) at (4,2) {$|\phi_{x}'\>$};
    \draw[|->] (psi) -- node[left]{$U_{t_{1}\leftarrow t_{0}}$} (psit);
    \draw[|->] (phi) -- node[left]{$U_{t_{2}\leftarrow t_{1}}^{\dag}$} (phit);
    \draw[dotted] (t2Lt) -- (-0.75,4);
    \draw[dotted] (t2Lb) -- (-0.75,4);
    \draw[dotted] (tL) -- (-0.75,2);
    \draw[dotted] (t1L) -- (-0.75,0);
    \end{tikzpicture}
    \caption{In this figure, the algebras $\mA_0,\mA_1,$ and $\mA_2$ are all equal to some fiducial $\mA$, and the subscript is meant to label the time. A state $|\psi\>$ is initially prepared at time $t_{0}$. Then, it evolves from $t_{0}$ to $t_{2}$ via $U_{t_{2}\leftarrow t_{0}}$. Finally, a state $|\phi_{x}\>$ is measured at time $t_{2}$ via the POVM $\mathcal{E}$. The two-state of Refs.~\cite{ReAh95,AhVa08} at some intermediate time $t_{1}$ is obtained by forward-propagating $|\psi\>$ to $t_{1}$ via $U_{t_{1}\leftarrow t_{0}}$ and back-propagating $|\phi_{x}\>$ to $t_{1}$ via $U^{\dag}_{t_{2}\leftarrow t_{1}}$. Taking the outer product of these two defines the two-state $|\psi'\>\<\phi'_{x}|$. This two-state is precisely the state over time associated with our right bloom $(\mathcal{E}\circ\Ad_{U_{t_{2}\leftarrow t_{0}}})\star\rho\in\mA_{0}\otimes\mB$ after forward-propagating the latter via $U_{t_{1}\leftarrow t_{0}}$ to get an element of $\mA_{1}\otimes\mB$.}
    \label{fig:twostate}
\end{figure}

In the case where the Hamiltonian is not trivial and there are three times $t_{0}<t_{1}<t_{2}$ with $U_{t_{1}\leftarrow t_{0}}$, $U_{t_{2}\leftarrow t_{1}}$, and $U_{t_{2}\leftarrow t_{0}}$ describing the unitary evolution from $t_{0}$ to $t_{1}$, $t_{1}$ to $t_{2}$, and $t_{0}$ to $t_{2}$, respectively, (cf.\ Figure~\ref{fig:twostate}) then, upon setting
$\mathcal{E}':=\mathcal{E}\circ\Ad_{U_{t_{2}\leftarrow t_{0}}},$
$|\psi'\>:=U_{t_{1}\leftarrow t_{0}}|\psi\>$, and 
$|\phi_{x}'\>:=U^{\dag}_{t_{2}\leftarrow t_{1}}|\phi_{x}\>$, we obtain
\begin{align}
\mathcal{E}'\star\rho
&=(\rho\otimes1_{\C^{X}})\mathscr{D}[\mathcal{E}\circ\Ad_{U_{t_{2}\leftarrow t_{0}}}] \nonumber \\
&=(\rho U_{t_{2}\leftarrow t_{0}}^{\dag}\otimes 1_{\C^{X}})\mathscr{D}[\mathcal{E}](U_{t_{2}\leftarrow t_{0}}\otimes1_{\C^{X}}) \nonumber \\
&\cong\bigoplus_{x\in X}\rho U_{t_{2}\leftarrow t_{0}}^{\dag}M_{x}U_{t_{2}\leftarrow t_{0}} \nonumber \\
&=\bigoplus_{x\in X}|\psi\>\<\psi|U_{t_{1}\leftarrow t_{0}}^{\dag}U_{t_{2}\leftarrow t_{1}}^{\dag}|\phi_{x}\>\<\phi_{x}|U_{t_{2}\leftarrow t_{1}}U_{t_{1}\leftarrow t_{0}} \nonumber \\
&=\bigoplus_{x\in X}|\psi\>\<\psi'|\phi'_{x}\>\<\phi'_{x}|U_{t_{1}\leftarrow t_{0}},
\end{align}
where we have used Lemma~\ref{lem:EgamDE} in the second equality.
Therefore, 
\be
\left(\Ad_{U_{t_{1}\leftarrow t_{0}}}\otimes\id_{\C^{X}}\right)\big(\mathcal{E}'\star\rho\big)\cong\bigoplus_{x\in X}\<\psi'|\phi_{x}'\>|\psi'\>\<\phi_{x}'|,
\ee
The intuitive reason for the $\Ad_{U_{t_{1}\leftarrow t_{0}}}$ in front of $\mathcal{E}'\star\rho$ is because the two-state of Ref.~\cite{ReAh95} is viewed at time $t_{1}$ rather than the initial time $t_{0}$ or the final time $t_{2}$ (cf.\ Figure~\ref{fig:twostate}). As such, it is necessary to propagate our state over time from time $t_{0}$ to $t_{1}$ in order to obtain the two-state of Refs.~\cite{ReAh95,AhVa08}. The $\id_{\C^{X}}$ is needed to incorporate all the possible outcomes due to the measurement.

We now examine what Bayes maps look like for the previous setup when the unitary evolution is trivial.  
Since a Bayes map $\mathcal{E}^{\star}_{\rho}:\mB\to\mA$ must be trace-preserving, it must define an $m\times m$ matrix $\rho_{x}:=\mathcal{E}^{\star}_{\rho}(\delta_{x})$ such that $\tr(\rho_{x})=1$ for every $x\in X$. Write $p:=\mathcal{E}(\rho)\equiv\bigoplus_{x}p_{x}\equiv\sum_{x}p_{x}\delta_{x}$ as the associated probability distribution on $X$, which is given by $p_{x}=\tr(M_{x}\rho)$. Then  
\begin{align}
\rho_{x}
&=\rho\mathcal{E}^{*}\big(\mathcal{E}(\rho)^{-1}\delta_{x}\big)
=\rho\mathcal{E}^{*}\big(\mathcal{E}_{x}(\rho)^{-1}\delta_{x}\big) \nonumber \\
&=\rho\left(\frac{1}{\tr(M_{x}\rho)}\right)\mathcal{E}^*(\delta_{x})
=\frac{\rho M_{x}}{p_{x}}
\end{align}
whenever $p_x\neq 0$. 
Note that $\mathcal{E}^{\star}_{\rho}$ is not $\dag$-preserving, and it was necessary to use the version of Bayes' rule from Definition~\ref{defn:Bayesmap} to derive this result. 
Also, since $\mathcal{E}^{\star}_{\rho}:\C^{X}\to\matr_{m}$ is a linear trace-preserving map from a classical algebra to a matrix algebra, it can be viewed as a sort of ensemble (though not technically since each $\rho_{x}$ need not be a density matrix). 

In the special case where $\rho=|\psi\>\<\psi|$ corresponds to a pure state $|\psi\>$ and similarly $M_{x}=|\phi_{x}\>\<\phi_{x}|$ is a one-dimensional projection, this becomes
\be
\rho_{x}=\frac{|\psi\>\<\psi|\phi_{x}\>\<\phi_{x}|}{|\<\psi|\phi_{x}\>|^2}
=\frac{|\psi\>\<\phi_{x}|}{\<\phi_{x}|\psi\>},
\ee
which agrees with the \emph{normalized two-state} from~\cite[Equation~(5)]{ReAh95} and the \emph{transition matrix} from~\cite[equation (1.3)]{NTTTW21}.  
The \emph{expectation value} 
\be
\tr(\rho_{x}^{\dag}A)=\frac{\<\psi|A|\phi_{x}\>}{\<\psi|\phi_{x}\>}
\ee
of this normalized two-state on an observable $A\in\matr_{m}$ then agrees with the \emph{weak value} of Y.\ Aharonov, D.\ Albert, and L.\ Vaidman~\cite[Equation~(6)]{AAV88}~%
\footnote{When going from the Schr{\"o}dinger picture to the Heisenberg picture, the two-state $\rho_{x}$ gets sent to its weak value expectation functional sending an observable $A$ to $\tr(\rho_{x}^{\dag}A)$. This duality, which is obtained by using the Hilbert--Schmidt dual, is a conjugate-linear isomorphism with our convention. This is why $\rho_{x}^{\dag}$, as opposed to $\rho_{x}$, appears in the expression for weak values. Technically, $\tr(\rho_{x}A)$ reproduces the expression~\cite[Equation~(6)]{AAV88} exactly. This is due to a matter of convention, since, if we use the right bloom (to be defined in the next section), the two-state would instead be given by $\rho_{x}=\frac{M\rho_{x}}{\tr(M\rho_{x})}=\frac{|\phi_{x}\>\<\psi|}{\<\psi|\phi_{x}\>}$. The associated weak values for this two-state would then equal $\tr(\rho_{x}^{\dag}A)$, in complete agreement with \cite[Equation~(6)]{AAV88}.}.
To contrast this with the earlier PEM scenario where the Bayesian inverse (using the Leifer--Spekkens state over time) of a measurement was a preparation, the current choice of state over time function provides a Bayes map that in general lacks positivity and leads to a different interpretation in terms of weak values.

\subsection{The left bloom}
\label{sec:lbsot}

The cases $(r,s)=(1,0)$ and $(r,s)=(0,1)$ yield the bloom $\mathcal{E}\star\rho=\mathscr{D}[\mathcal{E}](\rho\otimes1_{\mB})$ from~\cite{FuPa21,FuPa22}, which will be referred to as the \emph{left bloom} here. Besides satisfying the classical limit axiom, the left bloom is also bilinear and associative. 

The Bayes maps associated with the left bloom are similar to those for the right bloom, so only the solution for the Bayes maps will be provided. The unique linear map solving Bayes's rule is given by 
\be
\mB\ni B\xmapsto{\mathcal{E}^{\star}_{\rho}}\mathcal{E}^*\big(B\mathcal{E}(\rho)^{-1}\big)\rho.
\ee
This map is CPTP under exactly the same conditions as in the right bloom case.

\subsubsection{The two-time correlator}

In the special case where $\mathcal{E}$ is the identity channel, the left bloom appears in the context of two-point quantum correlation functions~\cite{BDOV13}, where it is equal to the ideal two-point quantum correlator applied to $\rho$~%
\footnote{Indeed, when $\mA=\matr_{m}$, the channel state $\mathscr{D}[\id_{\mA}]$ equals the swap operator sending 
\unexpanded{$|i\>\otimes|j\>$} to  \unexpanded{$|j\>\otimes|i\>$} in Ref.~\cite{BDOV13}.}.
Although not a positive (nor $\dag$-preserving) map, the two-point correlator was given an operational interpretation in Refs.~\cite{BDOV13,BDOV14} by decomposing it into the sum of two $\dag$-preserving operations, one of which is the symmetric bloom state over time, and the other of which is proportional to a commutator (see Section~\ref{sec:opsymb} on the symmetric bloom for more details). 

In the slightly more general case where $\mathcal{E}$ describes unitary evolution after some time $t$, i.e., $\mathcal{E}=\Ad_{e^{-iHt}}$ for some Hamiltonian $H\in\mA^{\sa}$, then the left-bloom defines two-time correlators~\cite{BHR17}. 
Indeed, the \emph{time-dependent correlation} between two observables $A,B\in\mA^{\sa}$, is 
\be
\big\<B(t)A(0)\big\>_{\rho}:=
\tr\big(e^{iHt}Be^{-iHt}A\rho\big),
\ee
where $A(s):=\Ad_{e^{isH}}(A)$ for $s\in\R$ and similarly for $B$. 
If $\mathcal{E}\star\rho$ denotes the left bloom state over time, then
\be
\big\<B(t)A(0)\big\>_{\rho}
=
\tr\big((\mathcal{E}\star\rho)^{\dag}(A\otimes B)\big)
\ee
for all $A,B\in\mA^{\sa}$. Indeed, 
\begin{align}
\tr\big((\mathcal{E}\star\rho)^{\dag}(A\otimes& B)\big)
=\tr\big((\rho\otimes1_{\mA})\mathscr{D}[\Ad_{U^{\dag}}](A\otimes B)\big) \nonumber \\
&=\tr\Big(\big((\id\otimes\Ad_{U^{\dag}})\mu_{\mA}^{*}(1_{\mA})\big)(A\rho\otimes B)\Big) \nonumber \\
&=\tr\big(\mu_{\mA}^{*}(1_{\mA})(A\rho\otimes UBU^{\dag})\big) \nonumber \\
&=\tr(A\rho UBU^{\dag}) \nonumber \\
&=\tr(UBU^{\dag} A\rho),
\end{align}
where we have temporarily set $U=e^{itH}$. The third equality follows from self-adjointness of $\mathscr{D}[\Ad_{U^{\dag}}]$ and $\mu_{\mA}^{*}(1_{\mA})$, while the fourth equality follows from self-adjointness of $\mu_{\mA}^{*}(1_{\mA})$ and $1_{\mA}$.

\subsubsection{Restoring the symmetry between left and right blooms}

Note that if we denote the left and right bloom state over time functions by $\star_{\mathrm{L}}$ and $\star_{\mathrm{R}}$, respectively, then
\be
\mathcal{E}\star_{\mathrm{L}}\rho
=\tau\big(\widetilde{\mathcal{E}^{\star}_{\rho}}\star_{\mathrm{L}}\mathcal{E}(\rho)\big)
=\gamma\big(\mathcal{E}^{\star}_{\rho}\star_{\mathrm{R}}\mathcal{E}(\rho)\big),
\ee
where the first equality is Bayes' rule and the second equality is the relationship between the left and right bloom states over time~%
\footnote{In terms of reverse orientation state over time functions, 
this reads $\star_{L}^{\dag}=\star_{R}$, and a proof is given by
\begingroup
\allowdisplaybreaks
\[
\begin{split}
\mathcal{E}\star_{\mathrm{L}}^{\dag}\rho
&=\big((\dag\circ\mathcal{E}\circ\dag)\star_{\mathrm{L}}\rho\big)^{\dag}
=\big((\rho\otimes1_{\mB})\mathscr{D}[\dag\circ\mathcal{E}\circ\dag]\big)^{\dag}\\
&=\mathscr{D}[\dag\circ\mathcal{E}\circ\dag]^{\dag}(\rho\otimes 1_{\mB})
=\mathscr{D}[\mathcal{E}](\rho\otimes 1_{\mB})
=\mathcal{E}\star_{\mathrm{R}}\rho,
\end{split}
\]
\endgroup
where we used the second identity in Equation~\eqref{eq:EgammaDE} of Lemma~\ref{lem:EgamDE} in the second-last equality.}.
This relationship between left and right bloom and the connection to Bayes' rule resolves the open question of Leifer and Spekkens with regard to the apparent asymmetry between these two states over time (cf.\ last paragraph in~\cite[Section~VII.B.1]{LeSp13}). In fact, it \emph{is} a symmetry.

\subsection{Bayes maps for the \texorpdfstring{$(r,s)$}{(r,s)} family}
\label{sec:rsfambayes}

After going through several examples of the $(r,s)$ family, here we derive the general formula for the Bayes map. 
Since a Bayes map $\mathcal{E}^{\star}_{\rho}$ must satisfy 
$\gamma(\mathcal{E}\star\rho)=\big((\dag\circ\mathcal{E}^{\star}_{\rho}\circ\dag)\star\mathcal{E}(\rho)\big)^{\dag}$,
one can show (by similar calculations to the above) that Bayes' rule is equivalent to
\begin{align}
&s(1_{\mB}\otimes\rho^{r})\mathscr{D}[\mathcal{E}^{*}](1_{\mB}\otimes\rho^{r^{\perp}})+s^{\perp}(1_{\mB}\otimes\rho^{r^{\perp}})\mathscr{D}[\mathcal{E}^*](1_{\mB}\otimes\rho^{r}) \nonumber \\
&=s\big(\mathcal{E}(\rho)^{r^{\perp}}\otimes1_{\mA}\big)\mathscr{D}[\mathcal{E}^{\star}_{\rho}]\big(\mathcal{E}(\rho)^{r}\otimes1_{\mA}\big) \nonumber \\
&\quad+s^{\perp}\big(\mathcal{E}(\rho)^{r}\otimes1_{\mA}\big)\mathscr{D}[\mathcal{E}^{\star}_{\rho}]\big(\mathcal{E}(\rho)^{r^{\perp}}\otimes1_{\mA}\big)
\end{align}
by Equation~\eqref{eq:EgammaDE} of Lemma~\ref{lem:EgamDE}.
Introducing the same notation $q_{k}$ and $|w_{k}\>$ as in the case of the symmetric bloom, Bayes' rule is equivalent to 
\begin{align}
&\sum_{k,l}e_{kl}\otimes\Big(s\rho^{r}\mathcal{E}^{*}(e_{lk})\rho^{r^{\perp}}+s^{\perp}\rho^{r^{\perp}}\mathcal{E}^{*}(e_{lk})\rho^{r}\Big) \nonumber \\
&=\sum_{k,l}\Big(\big(sq_{k}^{r^{\perp}}q_{l}^{r}+s^{\perp}q_{k}^{r}q_{l}^{r^{\perp}}\big)e_{kl}\Big)\otimes\mathcal{E}^{\star}_{\rho}(e_{lk}),
\end{align}
where 
$
e_{kl}:=|w_{k}\>\<w_{l}|.
$
Since the $e_{kl}$ are linearly independent, this gives the solution
\be
\mathcal{E}^{\star}_{\rho}(e_{kl})=\frac{s\rho^{r}\mathcal{E}^{*}(e_{kl})\rho^{1-r}+(1-s)\rho^{1-r}\mathcal{E}^{*}(e_{kl})\rho^{r}}{sq_{k}^{r}q_{l}^{1-r}+(1-s)q_{k}^{1-r}q_{l}^{r}},
\ee
which can be linearly extended to define a Bayes map for the $(r,s)$ family. Note that $\mathcal{E}^{\star}_{\rho}$ is $\dag$-preserving when $s=\frac{1}{2}$. It is unclear to us if there is a manifestly basis-independent expression for this Bayes map.

\subsection{Generalized conditional expectations}
\label{sec:GCE}

In Ref.~\cite{Ts22}, M.\ Tsang argued for the interpretation of generalized conditional expectations as maps defining retrodiction, and hence as quantum analogues of Bayes' theorem. Generalized conditional expectations include the Petz recovery maps, the Bayes maps for the entire $(r,s)$ family, and many others (see also~\cite[Section 6.1]{Ha17} and the references therein). In the present section, we will show that \emph{all} of the generalized conditional expectations described by Tsang are indeed included in our framework of states over time and Bayes maps. Namely, from the data used in Refs.~\cite{Ha17,Ts22} to define generalized conditional expectations, we construct state over time functions, and then we show that (the Hilbert--Schmidt adjoint of) these generalized conditional expectation are Bayes maps for those states over time. 

Beyond what we have already mentioned above, our work also achieves several new insights.
\begin{enumerate}
\item
To derive our Bayes map, we do not need to minimize or extremize any distance measures as is done in~\cite[Section~III.B]{Ts22}. Namely, we only require a few consistency conditions and our notion of time-reversal symmetry.
\item
We complete Tsang's open remark/question (posed at the end of \cite[Section~III.B]{Ts22}) of relating his formalism to that of Horsman et al.~\cite{HHPBS17} and Leifer--Spekkens~\cite{LeSp13}. We are able to do this precisely because of our proposed relationship between states over time and Bayes' rules (the usages of both $\tau$ and $\widetilde{\mathcal{E}^{\star}_{\rho}}$ from Definition~\ref{defn:Bayesmap} are crucial here). 
\item
We provided an explicit formula for the Bayes map associated with the $(r,s)$ family. 
\end{enumerate}

In what follows, our notation will differ from Refs.~\cite{Ts22,Ha17} to avoid conflicting notation. For every multi-matrix algebra $\mA$, let $\Theta:\mathcal{S}(\mA)\to\Map(\mA,\mA)$ send a state $\rho$ to a linear map $\Theta_{\rho}:\mA\to\mA$ that satisfies the following axioms (see~\cite[Section~III.A]{Ts22} and~\cite[Section~6.1]{Ha17}) %
\footnote{Axiom~(T\ref{item:T3}) is weaker than what appears in Refs.~\cite{Ha17,Ts22}, where the latter demanded $\Theta_{\rho\otimes\rho'}=\Theta_{\rho}\otimes\Theta_{\rho'}$. See the next footnote for the significance of this.}.
\begin{enumerate}[(T1)]
\item
\label{item:T1}
$\Theta_{\rho}(A)=\rho A$ whenever $A\in\mA$ satisfies $[\rho,A]=0$.
\item
\label{item:T2}
$\Theta_{\rho}\circ\mathcal{E}^{-1}=\mathcal{E}^{-1}\circ\Theta_{\mathcal{E}(\rho)}$ whenever $\mathcal{E}$ is a $*$-isomorphism (eg.\ $\Theta_{\rho}\circ\Ad_{U^{\dag}}=\Ad_{U^{\dag}}\circ\Theta_{U\rho U^{\dag}}$ for all unitaries $U$)
\item
\label{item:T3}
If $\mA'$ is another multi-matrix algebra, then $\Theta_{\rho\otimes\rho'}=\Theta_{\rho}\otimes\Theta_{\rho'}$ for all states $\rho\in\mathcal{S}(\mA)$ and $\rho'\in\mathcal{S}(\mA')$, provided that at least one of $\rho$ or $\rho'$ are in the center of $\mA$ or $\mA'$, respectively.
\item
\label{item:T4}
$\Theta_{\rho}$ is self-adjoint and positive semi-definite with respect to the Hilbert--Schmidt inner product. 
\end{enumerate}
We call such a map $\Theta$ a \define{state-rendering map}. 
Note that in our notation we have suppressed the dependence of $\Theta$ on the algebra $\mA$, as it should be evident from the context.

\begin{remark}
The fact that a map $\Phi:\mA\to\mA$ is self-adjoint with respect to the Hilbert--Schmidt inner product \emph{does not} imply that $\Phi$ is $\dag$-preserving.  
An example is the map $\Phi(A)=\rho A$ for a state $\rho\in\mathcal{S}(\mA)$ that is not in the center of $\mA$ (when $\mA=\matr_{m}$, this means $\rho$ is not the maximally mixed state).
Conversely, if $\Phi$ is $\dag$-preserving, this does not imply $\Phi$ is self-adjoint. 
An example is the map $\Phi=\Ad_{U}$ for a  unitary $U\in\mA$ not proportional to the identity.
Similarly, the fact that $\Phi$ is positive semi-definite with respect to the Hilbert--Schmidt inner product \emph{does not} mean that $\Phi$ is a positive map in the sense that it takes positive elements to positive elements.
An example is $\Phi(A)=\rho A$ for a state $\rho$ that is not maximally mixed.
Conversely, if $\Phi$ is positive, this does not imply $\Phi$ is positive semi-definite. 
An example is $\Phi(A)=A^{\mathrm{T}}$, the transpose map on (complex) matrix algebras.
\end{remark}

\begin{theorem}
\label{thm:Tsot}
Given a state-rendering map $\Theta$, the assignment $\star$ sending a trace-preserving linear map $\mathcal{E}:\mA\to\mB$ and $\rho\in\mathcal{S}(\mA)$ to 
\be
\mathcal{E}\star\rho:=(\Theta_{\rho}\otimes\id_{\mB})\big(\mathscr{D}[\mathcal{E}]\big)
\ee
defines a state over time function that is process-linear and satisfies the classical limit axiom. 
\end{theorem}

We include the proof of this theorem here to illustrate how the axioms of state-rendering maps are related to those of state over time functions.

\begin{proof} 
We first check that $\mathcal{E}\star\rho$ has the correct marginals, and then we prove the latter two claims. 

The marginal on $\mA$ is 
\begin{align}
\tr_{\mB}(\mathcal{E}\star\rho)&=\big(\Theta_{\rho}\otimes(\tr\circ\mathcal{E})\big)\big(\mu_{\mA}^{*}(1_{\mA})\big) \nonumber \\
&=\big(\Theta_{\rho}\otimes\tr\big)\big(\mu_{\mA}^{*}(1_{\mA})\big) \nonumber \\
&=\Theta_{\rho}(1_{\mA}) \nonumber \\
&=\rho.
\end{align}
The first equality follows from the definition of the channel state and the definition of the partial trace $\tr_{\mB}$. The second equality follows from the fact that $\mathcal{E}$ is trace-preserving. The third equality is an identity due to the fact that $\mu_{A}^*$ is a (non-positive) broadcasting map (see the axioms of a quantum Markov category in Refs.~\cite{PaBayes,PaQPL21}). The last equality follows from axiom~(T\ref{item:T1}). 

In the case where $\mA=\matr_{m}$, the marginal on $\mB$ is 
\begin{align}
\tr_{\mA}(\mathcal{E}\star\rho)
&=\sum_{i,j}\tr\big(\Theta_{\rho}(E_{ij}^{(m)})\big)\mathcal{E}(E_{ji}^{(m)}) \nonumber \\
&=\sum_{i,j}\tr\big(\Theta_{\rho}^{*}(1_{\mA})^{\dag}E_{ij}^{(m)}\big)\mathcal{E}(E_{ji}^{(m)}) \nonumber \\
&=\sum_{i,j}\tr\big(\Theta_{\rho}(1_{\mA})^{\dag}E_{ij}^{(m)}\big)\mathcal{E}(E_{ji}^{(m)}) \nonumber \\
&=\sum_{i,j}\tr\big(\rho E_{ij}^{(m)})\mathcal{E}(E_{ji}^{(m)}) \nonumber \\
&=\mathcal{E}(\rho).
\end{align}
The second equality follows from the definition of the Hilbert--Schmidt inner product.
The third equality follows from the self-adjointness of $\Theta_{\rho}$, which is part of axiom~(T\ref{item:T4}). The fourth equality follows from axiom~(T\ref{item:T1}) and the fact that $\rho$ is hermitian. 
The last equality follows from the fact that $\tr\big(\rho E_{ij}^{(m)})$ is precisely the $ji$-matrix entry of $\rho$. The proof when $\mA$ is an arbitrary multi-matrix algebra is left as an exercise. 

Process-linearity of $\star$ follows  from linearity of $\mathscr{D}$ and $\Theta_{\rho}$. 

Finally, to see that the classical limit axiom for $\star$ holds, suppose $[\rho\otimes1_{\mB},\mathscr{D}[\mathcal{E}]]=0$. Then 
\begin{equation}
\label{eq:rbe}
\left[\rho\otimes\frac{1_{\mB}}{\tr(1_{\mB})},\mathscr{D}[\mathcal{E}]\right]=0
\end{equation}
and 
\begin{align}
\mathcal{E}\star\rho&=\left(\Theta_{\rho}\otimes\Big(\tr(1_{\mB})\Theta_{\frac{1_{\mB}}{\tr(1_{\mB})}}\Big)\right)\big(\mathscr{D}[\mathcal{E}]\big) \nonumber \\
&=\tr(1_{\mB})\Theta_{\rho\otimes\frac{1_{\mB}}{\tr(1_{\mB})}}\big(\mathscr{D}[\mathcal{E}]\big) \nonumber \\
&=(\rho\otimes1_{\mB})\mathscr{D}[\mathcal{E}].
\end{align}
The first equality follows from axiom~(T\ref{item:T1}). 
The second equality follows from linearity of the tensor product and axiom~(T\ref{item:T3}).
The third equality follows from Equation~\eqref{eq:rbe} and axiom~(T\ref{item:T1}).
\end{proof}

Many of the examples of state over time functions we gave earlier are special cases of this construction 
\footnote{As mentioned briefly in the previous footnote, our axiom~(T\ref{item:T3}) is weaker than in Refs.~\cite{Ha17,Ts22}. If we had used $\Theta_{\rho\otimes\rho'}=\Theta_{\rho}\otimes\Theta_{\rho'}$ as in Refs.~\cite{Ha17,Ts22}, then the symmetric bloom and $(r,s)$ family, with $s\in(0,1)$, would \emph{not} satisfy this axiom (however, it does hold if $s\in\{0,1\}$). It suffices to illustrate this in the case of the symmetric bloom. A counter-example can be obtained for $\mA=\matr_{2}=\mA'$ by taking $\rho=\left[\begin{smallmatrix}p&0\\0&1-p\end{smallmatrix}\right]$, $\rho'=\left[\begin{smallmatrix}p'&0\\0&1-p'\end{smallmatrix}\right]$, $A=\left[\begin{smallmatrix}0&1\\1&0\end{smallmatrix}\right]=A'$, and $p,p'\in[0,\frac{1}{2})\cup(\frac{1}{2},1]$. Indeed, $\Theta_{\rho\otimes\rho'}^{\mathrm{J}}(A\otimes A')\ne\Theta_{\rho}^{\mathrm{J}}(A)\otimes\Theta_{\rho'}^{\mathrm{J}}(A')$.}.
\begin{enumerate}
\item
The right and left blooms are obtained from $\Theta_{\rho}^{\mathrm{R}}(A):=\rho A$ and $\Theta_{\rho}^{\mathrm{L}}(A):=A\rho$, respectively. 
\item
The Leifer--Spekkens state over time is obtained from $\Theta_{\rho}^{\mathrm{LS}}(A):=\sqrt{\rho}A\sqrt{\rho}$. 
\item
The symmetric bloom state over time is obtained from $\Theta_{\rho}^{\mathrm{J}}(A):=\frac{1}{2}\{\rho,A\}\equiv\frac{1}{2}(\rho A+A\rho)$. 
\item
More generally, the $(r,s)$ family is obtained from $\Theta_{\rho}^{(r,s)}(A):=s\rho^{r}A\rho^{1-r}+(1-s)\rho^{1-r}A\rho^{r}$. 
\end{enumerate}

We now introduce the generalized conditional expectation from~\cite[Equation~(3.22)]{Ts22} (see also~\cite[Equation~(6.21)]{Ha17} and Ref.~\cite{Pe88b}), which was derived in Ref.~\cite{Ts22} by minimizing a certain inner product associated with a state-rendering map $\Theta$.

Given a channel $\mathcal{E}:\mA\to\mB$, a state $\rho\in\mathcal{S}(\mA)$, and a state-rendering map $\Theta$, a \define{generalized conditional expectation} is a linear map $\mathcal{E}_{\Theta,\rho}:\mA\to\mB$ such that 
\begin{equation}
\label{eq:GCE}
\mathcal{E}\circ\Theta_{\rho}=\Theta_{\mathcal{E}(\rho)}\circ\mathcal{E}_{\Theta,\rho}.
\end{equation}

\begin{theorem}
\label{thm:GCEandBayes}
Let $\mathcal{E}_{\Theta,\rho}$ be a generalized conditional expectation as in Equation~\eqref{eq:GCE}. Then $\mathcal{E}_{\Theta,\rho}^*$ is a Bayes map for $(\mathcal{E},\rho)$ associated with the state over time $\star$ defined in Theorem~\ref{thm:Tsot}.
\end{theorem}

The proof of Theorem~\ref{thm:GCEandBayes} is provided in Appendix~\ref{app:B}. Interestingly, we made no use of axiom~(T\ref{item:T2}) in any of our results thus far. The importance of this axiom is due to the next fact, the proof of which is deferred to Appendix~\ref{app:B}. It shows that if a quantum channel is a $*$-isomorphism, then the inverse is a Bayesian inverse. 

\begin{proposition}
\label{prop:TsIn}
Let $\Theta$ be a state-rendering map and $\star$ its associated state over time function. 
If $\mathcal{E}:\mA\to\mB$ is a $*$-isomorphism, 
then $\mathcal{E}^{-1}$ is a Bayesian inverse of $(\mathcal{E},\rho)$ for every $\rho\in\mathcal{S}(\mA)$.
\end{proposition}

\begin{remark}
\label{rmk:TsInv}
Proposition~\ref{prop:TsIn} and Theorem~\ref{thm:GCEandBayes} provide additional illustrations of the importance of using the dagger and $\widetilde{\mathcal{E}}^{\star}_{\rho}$ in the definition of Bayes' rule. This can be seen more clearly in their proofs.
\end{remark}

\section{Discussion and Conclusions}

In this work, we provided a rigorous definition for a Bayesian inverse using the notion of a state over time~\cite{HHPBS17,FuPa22} in terms of a universal time-reversal symmetry map that is independent of the input channel and state. In particular, this answers an open question of Leifer and Spekkens~\cite{LeSp13} on the question about the relationship between Bayes' rules and time-reversal symmetry. It also answers a question posed by Tsang regarding the connection between states over time and generalized conditional expectations~\cite{Ts22}. Tsang also attempted to unify many notions of generalized conditional expectations, and our work includes many of the special cases considered in Ref.~\cite{Ts22} along with several new ones. We also showed how our definition of a state over time function reproduces those of Leifer and Spekkens~\cite{LeSp13}, the two-state vector formalism~\cite{Wat55,ABL64,ReAh95}, the symmetric bloom of the present authors~\cite{FuPa22}, and others, many of which are summarized in Table~\ref{table:SOTex}. 

\begin{table*}[t]
  \centering
  \begin{tabular}{c|c|ccccccc|c}
    Name (page ref) & state over time $\mathcal{E}\star\rho$ & P\ref{item:hermitian} & P\ref{item:locpos} & P\ref{item:pos} & P\ref{item:slin} & P\ref{item:plin} & P\ref{item:classicallimit} & A & Bayes map $\mathcal{E}^{\star}_{\rho}$ \\
    \hline
    uncorrelated  (\pageref{sec:uncsot}) & $\rho\otimes\mathcal{E}(\rho)$  & \cmark & \cmark & \cmark & \xmark & \cmark & \xmark & \xmark & any CPTP such that $\mathcal{E}^{\star}_{\rho}\big(\mathcal{E}(\rho)\big)=\rho$ \\
    Ohya compound  (\pageref{sec:Ohyasot}) & $\sum_{\alpha}\lambda_{\alpha}P_{\alpha}\otimes\mathcal{E}\left(\frac{P_{\alpha}}{\tr(P_{\alpha})}\right)$  & \cmark & \cmark & \cmark & \xmark & \cmark & $\ast$ & ? & not computed here  \\
    Leifer--Spekkens (\pageref{sec:LSsot}) & $\big(\sqrt{\rho}\otimes 1_{\mB}\big)\mathscr{D}[\mathcal{E}]\big(\sqrt{\rho}\otimes 1_{\mB}\big)$  & \cmark & \cmark & \xmark & \xmark & \cmark & \cmark & \xmark & Petz map $\mathscr{R}_{\rho,\mathcal{E}}:=\Ad_{\rho^{1/2}}\circ\mathcal{E}^{*}\circ\Ad_{\mathcal{E}(\rho)^{-1/2}}$ \\
    $t$-rotated (\pageref{sec:trotsot}) &  $(\rho^{1/2-it}\otimes1_{\mB})\mathscr{D}[\mathcal{E}](\rho^{1/2+it}\otimes1_{\mB})$  & \cmark & \cmark & \xmark & \xmark & \cmark & \cmark & \xmark & rotated Petz map $\Ad_{\rho^{-it}}\circ\mathscr{R}_{\rho,\mathcal{E}}\circ\Ad_{\mathcal{E}(\rho)^{it}}$ \\
    STH (\pageref{sec:trotsot}) & $\big(U_{\rho}^{\dag}\rho^{1/2}\otimes1_{\mB}\big)\mathscr{D}[\mathcal{E}]\big(\rho^{1/2}U_{\rho}\otimes1_{\mB}\big)$  & \cmark & \cmark & \xmark & \xmark & \cmark & \cmark & \xmark & $\Ad_{U_{\rho}^{\dag}}\circ\mathscr{R}_{\rho,\mathcal{E}}\circ\Ad_{U_{\mathcal{E}(\rho)}}$ \\
    symmetric bloom (\pageref{sec:SBsot}) & $\frac{1}{2}\big\{\rho\otimes 1_{\mB},\mathscr{D}[\mathcal{E}]\big\}$ & \cmark & \xmark & \xmark & \cmark & \cmark & \cmark & \cmark & $|w_{k}\>\<w_{l}|\mapsto(q_{k}+q_{l})^{-1}\big\{\rho,\mathcal{E}^*\big(|w_{k}\>\<w_{l}|\big)\big\}$ \\
    right bloom (\pageref{sec:rbsot}) & $(\rho\otimes1_{\mB})\mathscr{D}[\mathcal{E}]$ (ex.\ two-state) & \xmark & \xmark & \xmark & \cmark & \cmark & \cmark & \cmark & $B\mapsto\rho\mathcal{E}^*\big(\mathcal{E}(\rho)^{-1}B\big)$ (ex.\ weak values) \\
    left bloom (\pageref{sec:lbsot}) & $\mathscr{D}[\mathcal{E}](\rho\otimes1_{\mB})$ (ex.\ correlator)  & \xmark & \xmark & \xmark & \cmark & \cmark & \cmark & \cmark & $B\mapsto\mathcal{E}^*\big(B\mathcal{E}(\rho)^{-1}\big)\rho$ \\
  \end{tabular}
  \caption{A table of the many state over time functions appearing in this work, along with their formulas, properties satisfied, and associated Bayes maps. The axioms are hermiticity P\ref{item:hermitian}, block-positivity P\ref{item:locpos}, positivity P\ref{item:pos}, state-linearity P\ref{item:slin}, process-linearity P\ref{item:plin}, classical limit P\ref{item:classicallimit}, and associativity A (bi-linearity was removed from the table to avoid redundancy). The $\ast$ for Ohya's compound state over time is because the classical limit is satisfied for density matrices with no repeating eigenvalues. Note that we have not fully defined Ohya's compound state over time for arbitrary CPTP maps between multi-matrix algebras (this  will be addressed in future work, along with additional examples of state over time functions). The question mark represents that we have not yet determined whether the given axiom is satisfied.}
  \label{table:SOTex}
\end{table*}

In addition, using our proposed definitions of Bayesian inverses, and more generally Bayes maps, associated with a state over time function, we have obtained the normalized two-states of Reznik and Aharonov~\cite{ReAh95}, all the rotated Petz recovery maps, two-point quantum correlators, the quantum Bayes' rule of Fuchs~\cite{Fu01}, and many more concepts arising in various scenarios. Furthermore, we explained how one would unambiguously arrive at the quantum state-update rule associated with instruments~\cite{Kr71} for \emph{any} state over time satisfying the classical limit axiom. This shows how the state-update rule can be viewed as a quantum Bayes' rule, which was advocated by Bub~\cite{Bu77}, Ozawa~\cite{Oz97}, Tegmark~\cite{Te12}, and many others. However, we remark that this version of quantum Bayes' rule is not as general as our Bayes' rule, the latter of which applies to \emph{arbitrary} maps, not necessarily special kinds of instruments.

Several open questions should be addressed. Although some of these were mentioned earlier, we summarize these questions for convenience, and add new ones for further thought. 
\begin{enumerate}
\item
Are there any physically relevant examples of state over time functions that are not process-linear? 
\item
In regard to Ohya's compound state over time~\cite{Oh83a}, is there a canonical way of isolating an extremal decomposition of every state $\rho$ given a pair $(\mathcal{E},\rho)$, perhaps by extremizing some information measure? If so, can the compound state over time be improved to define a state over time function with more desirable properties?
\item
Is there a state over time function whose Bayesian inverses are the universal recovery maps of M.\ Junge et al.~\cite{JRSWW16,JRSWW18}? What about the more general averaged rotated Petz recovery maps~\cite{PaBu22} or the optimal state retrieval maps of J.\ Surace and M.\ Scandi~\cite{SuSc22}? 
\item
Which state over time functions give Bayesian inverses that define retrodiction functors in the sense of Ref.~\cite{PaBu22}? Namely, which properties of state over time functions correspond to which axioms for retrodiction families? For example, what properties of a state over time function imply that the Bayesian inverses are compositional?
This is an important question to address so that one can choose a state over time function appropriately to achieve the desired properties of retrodiction.
\item
Just as the two-point correlator can be decomposed into hermitian and anti-hermitian parts, with a minimal Kraus decompositions whose coefficients can be given an operational meaning~\cite{BDOV13}, can arbitrary non-positive Bayes maps also be decomposed in a similar way to provide them with operational meanings? Some progress has recently been made in this direction~\cite{Ts22b}.
\item
Another approach towards generalizing quantum states associated with acausally-related regions to causally-related ones is the superdensity formalism~\cite{CJQW18,CHQY19}. One major difference between this line of development and ours is that we demand our objects to be defined on the tensor product of the algebras of the associated regions, whereas Refs.~\cite{CJQW18,CHQY19} `double' the algebra. For example, for the case of two causally-related regions with algebras $\mA$ and $\mB$, our states over time are elements of $\mA\otimes\mB$, whereas  superdensity operators are elements of $\mA^*\otimes\mB^*\otimes\mA\otimes\mB$, where the duals are with respect to the Hilbert--Schmidt inner products. As such, we have not yet been able to meaningfully compare our constructions.
\item
Our present work focused primarily on the setting of two times, and hence two algebras mediated by some channel. The associativity axiom, which we alluded to briefly, is closely related to consistently defining \emph{multi-time} states over time where one might have a sequence of composable evolutions $\mA\xrightarrow{\mathcal{E}}\mB\xrightarrow{\mathcal{F}}\mC\xrightarrow{\;\;\;}\cdots$ together with an initial state $\rho\in\mathcal{S}(\mA)$. It is expected that these ideas are closely related to multi-time correlators and multi-time measurements~\cite{APTV09,HuGu22}, though the details have not yet been worked out. 
\item
Just as genuine states in the tensor product $\mA\otimes\mB$ can be given important information measures, such as entanglement entropy (see Ref.~\cite{Wi18} for a review with applications to quantum field theory), what information measures can be assigned to states over time for causally related regions? Some attempts to extend entanglement entropy to such scenarios have been made in Euclidean field theories for the special case of the right bloom under the name pseudo-entropy~\cite{NTTTW21,MSTTW21,MSTTW21b,DHMTT22}, or in the context of superdensity operators~\cite{CJQW18,CHQY19}. Much work needs to be done to better understand such dynamical information measures. What can these teach us about temporal correlations~\cite{LeGa85,NiSh96,FJV15,ChMa20,HuGu22} or information extraction from black hole evaporation~\cite{Ha75,Ha76,Ha82}?  
\end{enumerate}

We hope to address these questions, particularly some from the last item, in subsequent work. We hope that by answering such questions, we may eventually provide new methods for detecting and computing temporal correlations that distinguish between classical and quantum systems. More provocatively, we suspect that these differences may provide some insight towards our understanding of quantum information theory and the structure of space and time, and hence quantum gravity. 

\section*{Acknowledgements}
We thank Clive Cenxin Aw, Francesco Buscemi, John DeBrota, Markus Frembs, Marie Fries, Chris Fuchs, Xiao-Kan Guo, Sachin Gupta, Shintaro Minagawa, Jacques Pienaar, David Schmid, Blake Stacey, Jacopo Surace, Tadashi Takayanagi, Mankei Tsang, Mark Wilde, Edward Witten, Yil\`e Y\={\i}ng, and Beata Zjawin for discussions. 
We thank Carlo Rovelli for bringing Ref.~\cite{BHR17} to our attention.
We thank two anonymous reviewers for their helpful feedback. 
AJP thanks the Institut des Hautes \'Etudes Scientifiques and the Yukawa Institute for Theoretical Physics at Kyoto University 
for their hospitality. 
This work is supported by MEXT-JSPS Grant-in-Aid for Transformative Research Areas (A) ``Extreme Universe'', No.\ 21H05183.

\appendix

\section{Useful facts about the channel state}
\label{app:A}

Before getting to the examples, we first state an important lemma that will be used in several calculations.

\begin{lemma}
\label{lem:DEM}
Let $B$ be an $m\times n$ matrix and let $C$ be an $n\times m$ matrix. Then 
\be
\sum_{i,j}^{m}E_{ij}^{(m)}\otimes \left(CE_{ji}^{(m)}B\right)=\sum_{k,l}^{n}\left(BE_{kl}^{(n)}C\right)\otimes E_{lk}^{(n)}.
\ee
\end{lemma}

\begin{proof}
Writing $E_{ij}^{(m)}=|i\>\<j|$ and using the completeness relations $\mathds{1}_{n}=\sum_{k=1}^{n}|k\>\<k|$ and $\mathds{1}_{m}=\sum_{i=1}^{m}|i\>\<i|$, we obtain 
\begin{align}
\sum_{i,j}^{m}&E_{ij}^{(m)}\otimes \left(CE_{ji}^{(m)}B\right)
 \nonumber \\
&=\sum_{i,j}^{m}\sum_{k,l}^{n} 
|i\>\<j|\otimes\Big(|l\>\<l|C|j\>\<i|B|k\>\<k|\Big) \nonumber \\
&=\sum_{i,j}^{m}\sum_{k,l}^{n} 
\Big(|i\>\<i|B|k\>\<l|C|j\>\<j|\Big)\otimes|l\>\<k|, 
\end{align}
which equals the required expression. 
\end{proof}

\begin{lemma} 
\label{lem:gamDE}
If $\mathcal{E}:\mA\to\mB$ is a 
CP map, 
then 
\begin{equation}
\label{eq:EdtoEs}
(\id_{\mA}\otimes\mathcal{E})\big(\mu_{\mA}^{*}(1_{\mA})\big)
=(\mathcal{E}^*\otimes\id_{\mB})\big(\mu_{\mB}^{*}(1_{\mB})\big),
\end{equation}
\begin{equation}
\label{eq:gammaDE}
\gamma\big(\mathscr{D}[\mathcal{E}]\big)=\mathscr{D}[\mathcal{E}^*],
\end{equation}
and 
\begin{equation}
\label{eq:ABDE}
(A\otimes B)\mathscr{D}[\mathcal{E}](A'\otimes B')
=\mathscr{D}[\mathscr{L}_{B}\circ\mathscr{R}_{B'}\circ\mathcal{E}\circ\mathscr{L}_{A'}\circ\mathscr{R}_{A}]
\end{equation}
for all $A,A'\in\mA$ and $B,B'\in\mB$.
Here $\mathscr{L}_{A'}$ and $\mathscr{R}_{A}$ are left and right multiplication by $A$, namely $\mathscr{L}_{A'}(X)=A'X$ and $\mathscr{R}_{A}(X)=XA$ for all $X\in\mA$ (and similarly for $\mB$).
\end{lemma}

\begin{proof}
In this proof, we will assume $\mA=\matr_{m}$ and $\mB=\matr_{n}$ for simplicity. Then Equation~\eqref{eq:EdtoEs} reads
\be
\mathscr{D}[\mathcal{E}]=\sum_{k,l}\mathcal{E}^*\big(E_{kl}^{(n)}\big)\otimes E_{lk}^{(n)}.
\ee
To see that this holds, note that 
since $\mathcal{E}$ is CP, 
there exist Kraus operators $\{V_{\alpha}\}$ such that $\mathcal{E}=\sum_{\alpha}\Ad_{V_{\alpha}}$~\cite{Ch75}. 
Hence, 
\begin{align}
\mathscr{D}[\mathcal{E}]&=\sum_{\alpha}\sum_{i,j}E_{ij}^{(m)}\otimes\left(V_{\alpha}E_{ji}^{(m)}V_{\alpha}^{\dag}\right) \nonumber \\
&=\sum_{\alpha}\sum_{k,l}\left(V_{\alpha}^{\dag}E_{kl}^{(n)}V_{\alpha}\right)\otimes E_{lk}^{(n)} \nonumber \\
&=\gamma\big(\mathscr{D}[\mathcal{E}^*]\big)
\end{align}
by Lemma~\ref{lem:DEM}. Equation~\eqref{eq:gammaDE} follows directly from this. 
For the last identity, and using Lemma~\ref{lem:DEM} as well, we obtain
\begin{align}
&(A\otimes B)\mathscr{D}[\mathcal{E}](A'\otimes B') \nonumber \\
&=\sum_{\alpha}(1_{\mA}\otimes BV_{\alpha})\left(\sum_{i,j}AE_{ij}^{(m)}A'\otimes E_{ji}^{(m)}\right)(1_{\mA}\otimes V_{\alpha}^{\dag}B') \nonumber \\
&=\sum_{\alpha}(1_{\mA}\otimes BV_{\alpha})\left(\sum_{i,j}E_{ij}^{(m)}\otimes A'E_{ji}^{(m)}A\right)(1_{\mA}\otimes V_{\alpha}^{\dag}B') \nonumber \\
&=\sum_{i,j}E_{ij}^{(m)}\otimes\big(B\mathcal{E}(A'E_{ji}^{(m)}A)B'\big),
\end{align}
which equals~\eqref{eq:ABDE}.
\end{proof}

We now generalize Lemma~\ref{lem:gamDE} to allow for arbitrary linear maps.

\begin{lemma} 
\label{lem:EgamDE}
If $\mathcal{E}:\mA\to\mB$ is a linear map, then
\begin{equation}
\label{eq:EdtodEds}
(\id_{\mA}\otimes\mathcal{E})\big(\mu_{\mA}^{*}(1_{\mA})\big)=\big((\dag\circ\mathcal{E}^{*}\circ\dag)\otimes\id_{\mB}\big)\big(\mu_{\mB}^{*}(1_{\mB})\big),
\end{equation}
\begin{equation}
\label{eq:EgammaDE}
\gamma\big(\mathscr{D}[\mathcal{E}]\big)
=\mathscr{D}[\dag\circ\mathcal{E}^*\circ\dag]
=\mathscr{D}[\mathcal{E}^*]^{\dag},
\end{equation}
\begin{equation}
\label{eq:EABDE}
(A\otimes B)\mathscr{D}[\mathcal{E}](A'\otimes B')
=\mathscr{D}[\mathscr{L}_{B}\circ\mathscr{R}_{B'}\circ\mathcal{E}\circ\mathscr{L}_{A'}\circ\mathscr{R}_{A}],
\end{equation}
and
\begin{equation}
\label{eq:VEW}
(A^{\dag}\otimes B)\mathscr{D}[\mathcal{E}](A\otimes B^{\dag})=\mathscr{D}[\Ad_{B}\circ\mathcal{E}\circ\Ad_{A}]
\end{equation}
for all $A,A'\in\mA$ and $B,B'\in\mB$.
\end{lemma}

\begin{proof}
The first equality can be seen by first taking the Hilbert--Schmidt adjoints of both sides ($1_{\mA}$ is viewed as the unique linear map $\C\to\mA$ sending $\lambda\in\C$ to $\lambda1_{\mA}$, and similarly for $\mB$). The claim is then equivalent to~%
\footnote{The Hilbert--Schmidt adjoint $T^*$ of a conjugate-linear map $T:\mathcal{H}\to\mathcal{K}$ between Hilbert spaces $\mathcal{H}$ and $\mathcal{K}$ is defined by  \unexpanded{$\<v,T^*w\>_{\mathcal{H}}=\<w,Tv\>_{\mathcal{K}}$} for all $v\in\mathcal{H}$ and $w\in\mathcal{K}$. In this case, the conjugate-linear map in question is $\dag$, which is self-adjoint with respect to this definition.}
\be
\tr\circ\mu_{\mA}\circ
(\id_{\mA}\otimes\mathcal{E}^*)
=
\tr\circ\mu_{\mB}\circ
\big((\dag\circ\mathcal{E}\circ\dag)\otimes\id_{\mB}\big).
\ee
But this identity follows immediately from 
applying the left map to $A\otimes B\in\mA\otimes\mB$, which results in 
\be
\tr\big(A\mathcal{E}^*(B)\big)=\tr\big((A^{\dag})^{\dag}\mathcal{E}^*(B)\big)=\tr\big(\mathcal{E}(A^{\dag})^{\dag}B\big).
\ee
The first equality in Equation~\eqref{eq:EgammaDE} follows from Equation~\eqref{eq:EdtodEds}.
The second equality in Equation~\eqref{eq:EgammaDE} follows from
\begin{align}
\mathscr{D}[\dag\circ\mathcal{E}\circ\dag]^{\dag}
&=\Big(\big(\id_{\mA}\otimes(\dag\circ\mathcal{E}\circ\dag)\big)\big(\mu^{*}_{\mA}(1_{\mA})\big)\Big)^{\dag} \nonumber \\
&=\big((\id_{\mA}\otimes\mathcal{E})\circ(\dag\otimes\dag)\big)\big(\mu^{*}_{\mA}(1_{\mA})\big) \nonumber \\
&=(\id_{\mA}\otimes\mathcal{E})\big(\mu^{*}_{\mA}(1_{\mA})\big) \nonumber \\
&=\mathscr{D}[\mathcal{E}]
\end{align}
(the string-diagram language of quantum Markov categories from Ref.~\cite{PaBayes} makes these identities immediately apparent). 
Finally, Equation~\eqref{eq:EABDE} follows from Equation~\eqref{eq:ABDE} in Lemma~\ref{lem:gamDE} together with linearity of the channel state with respect to its input. This is because every linear map $\mathcal{E}$ is a complex combination of at most four CP maps~\cite{dePi67,Hi73,Ch75}. Indeed, one can first decompose the Choi matrix $\mathscr{C}[\mathcal{E}]$ of $\mathcal{E}$ into a complex combination of self-adjoint elements
\be
\mathscr{C}[\mathcal{E}]=\frac{1}{2}\big(\mathscr{C}[\mathcal{E}]+\mathscr{C}[\mathcal{E}]^{\dag}\big)+\frac{1}{2i}\big(i\mathscr{C}[\mathcal{E}]-i\mathscr{C}[\mathcal{E}]^{\dag}\big)
\ee
and then further decompose each of these self-adjoint elements into linear combinations of positive elements. The end result is a complex combination of the form 
\be
\mathscr{C}[\mathcal{E}]
=\mathscr{C}[\mathcal{E}]_{\mathfrak{R}}^{+}
-\mathscr{C}[\mathcal{E}]_{\mathfrak{R}}^{-}
+i\mathscr{C}[\mathcal{E}]_{\mathfrak{I}}^{+}
-i\mathscr{C}[\mathcal{E}]_{\mathfrak{I}}^{-}, 
\ee
where $\mathscr{C}[\mathcal{E}]_{\mathfrak{R}}^{+},\mathscr{C}[\mathcal{E}]_{\mathfrak{R}}^{-},\mathscr{C}[\mathcal{E}]_{\mathfrak{I}}^{+},\mathscr{C}[\mathcal{E}]_{\mathfrak{I}}^{-}$ are all positive. Then, by using Choi's theorem, each of these correspond to CP maps with Kraus decomposition 
\be
\mathcal{E}
=\sum_{\alpha}\Ad_{V_{\alpha}}
-\sum_{\beta}\Ad_{W_{\beta}}
+i\sum_{\gamma}\Ad_{X_{\gamma}}
-i\sum_{\delta}\Ad_{Y_{\delta}}.
\ee
Since the channel state is linear in its argument, this decomposition and Equation~\eqref{eq:ABDE} from  Lemma~\ref{lem:gamDE} imply Equation~\eqref{eq:EABDE}. Equation~\eqref{eq:VEW} is a special case of Equation~\eqref{eq:EABDE}.
\end{proof}

\section{Proofs of some main results}
\label{app:B}

This appendix contains proofs of our main theorems that were not included in the body. 
Before proving Theorem~\ref{thm:GCEandBayes} about the relationship between Bayes maps and generalized conditional expectations, we compute the reverse orientation state over time in the following lemma.

\begin{lemma}
\label{lem:dualTsang}
Let $\Theta$ be a state-rendering map and let $\star$ denote the associated state over time function from Theorem~\ref{thm:Tsot}. Then 
\be
\big(\widetilde{\mathcal{E}}\star\rho\big)^{\dag}
=\big((\dag\circ\Theta_{\rho}\circ\dag)\otimes\id_{\mB}\big)\mathscr{D}[\mathcal{E}]
\ee
for all linear maps $\mathcal{E}:\mA\to\mB$ and states $\rho\in\mathcal{S}(\mA)$.
\end{lemma}

\begin{proof}
By definition of the left-hand-side, we have
\begin{align}
\big(\widetilde{\mathcal{E}}\star\rho\big)^{\dag}&=\big((\dag\circ\mathcal{E}\circ\dag)\star\rho\big)^{\dag} \nonumber \\
&=\big((\Theta_{\rho}\otimes\id_{\mB})\mathscr{D}[\dag\circ\mathcal{E}\circ\dag]\big)^{\dag} \nonumber \\
&=\big((\dag\circ\Theta_{\rho}\circ\dag)\otimes\id_{\mB}\big)\big(\mathscr{D}[\dag\circ\mathcal{E}\circ\dag]^{\dag}\big) \nonumber \\
&=\big((\dag\circ\Theta_{\rho}\circ\dag)\otimes\id_{\mB}\big)\mathscr{D}[\mathcal{E}],
\end{align}
where  
the third equality follows from the properties of the involution $\dag$~
\footnote{As with some earlier proofs, the sequence of calculations from this proof are much more easily visualized using the string diagrams of quantum Markov categories~\cite{PaBayes}.}
and
the fourth equality follows from 
Equation~\eqref{eq:EgammaDE} of Lemma~\ref{lem:EgamDE}.
\end{proof}

\begin{proof}
[Proof of Theorem~\ref{thm:GCEandBayes}]
The goal is to prove $\gamma(\mathcal{E}\star\rho)
=\big(\widetilde{\mathcal{E}^{*}_{\Theta,\rho}}\star\mathcal{E}(\rho)\big)^{\dag}$. 
Taking the Hilbert--Schmidt adjoint of both sides of Equation~\eqref{eq:GCE} and using axiom~(T\ref{item:T4}) gives the equivalent condition 
\begin{equation}
\label{eq:GCEd}
\Theta_{\rho}\circ\mathcal{E}^*=\mathcal{E}_{\Theta,\rho}^{*}\circ\Theta_{\mathcal{E}(\rho)}.
\end{equation}
Using this, we obtain
\begingroup
\allowdisplaybreaks
\begin{align}
\gamma(\mathcal{E}&\star\rho)
=\gamma\Big(\big(\Theta_{\rho}\otimes\id_{\mB}\big)\big(\mathscr{D}[\mathcal{E}]\big)\Big) \nonumber \\
&=\big(\id_{\mB}\otimes\Theta_{\rho}\big)\gamma\big(\mathscr{D}[\mathcal{E}]\big) \nonumber \\
&=\big(\id_{\mB}\otimes\Theta_{\rho}\big)\big(\mathscr{D}[\mathcal{E}^*]\big) \nonumber \\
&=\big(\id_{\mB}\otimes(\Theta_{\rho}\circ\mathcal{E}^*)\big)\big(\mu_{\mB}^{*}(1_{\mB})\big) \nonumber \\
&=\big(\id_{\mB}\otimes(\mathcal{E}_{\Theta,\rho}^{*}\circ\Theta_{\mathcal{E}(\rho)})\big)\big(\mu_{\mB}^{*}(1_{\mB})\big) \nonumber \\
&=\Big(\big(\id_{\mB}\otimes\mathcal{E}_{\Theta,\rho}^{*}\big)\circ\big(\id_{\mB}\otimes\Theta_{\mathcal{E}(\rho)}\big)\Big)\big(\mu_{\mB}^{*}(1_{\mB})\big) \nonumber \\
&=\Big(\big(\id_{\mB}\otimes\mathcal{E}_{\Theta,\rho}^{*}\big)\circ\big((\dag\circ\Theta_{\mathcal{E}(\rho)}\circ\dag)\otimes\id_{\mB}\big)\Big)\big(\mu_{\mB}^{*}(1_{\mB})\big) \nonumber \\
&=\big((\dag\circ\Theta_{\mathcal{E}(\rho)}\circ\dag)\otimes\id_{\mB}\big)\mathscr{D}[\mathcal{E}^{*}_{\Theta,\rho}] \nonumber \\
&=\big(\widetilde{\mathcal{E}^{*}_{\Theta,\rho}}\star\mathcal{E}(\rho)\big)^{\dag}.
\end{align}
\endgroup
The second equality follows from the property of the swap map $\gamma$.
The third equality follows from Lemma~\ref{lem:gamDE}. 
The fourth equality is the definition of the channel state. 
The fifth equality follows from Equation~\eqref{eq:GCEd}. 
The sixth equality follows from the interchange law relating $\circ$ and $\otimes$. 
The seventh equality follows from Equation~\eqref{eq:EdtodEds} of Lemma~\ref{lem:EgamDE} and axiom~(T\ref{item:T4}).
The eighth equality follows from the interchange law and the definition of the channel state. 
The ninth equality follows from Lemma~\ref{lem:dualTsang}.
\end{proof}

\begin{proof}
[Proof of Proposition~\ref{prop:TsIn}]
By similar calculations to those we have seen before, 
\be
\tau\big(\widetilde{\mathcal{E}^{\star}_{\rho}}\star\mathcal{E}(\rho)\big)=\big((\mathcal{E}^{\star}_{\rho}\circ\Theta_{\mathcal{E}(\rho)})\otimes\id_{\mA}\big)\big(\mu_{\mA}^*(1_{\mA})\big).
\ee
Meanwhile, 
\begin{align}
\mathcal{E}\star\rho&=(\Theta_{\rho}\otimes\id_{\mA})\mathscr{D}[\mathcal{E}] \nonumber \\
&=\big((\Theta_{\rho}\otimes\id_{\mB})\circ(\mathcal{E}^*\otimes\id_{\mB})\big)\big(\mu_{\mB}^{*}(1_{\mB})\big) \nonumber \\
&=\big((\Theta_{\rho}\circ\mathcal{E}^{*})\otimes\id_{\mB}\big)\big(\mu_{\mB}^{*}(1_{\mB})\big) \nonumber \\
&=\big((\mathcal{E}^{*}\circ\Theta_{\mathcal{E}(\rho)})\otimes\id_{\mB}\big)\big(\mu_{\mB}^{*}(1_{\mB})\big),
\end{align}
where we used $\mathcal{E}^*=\mathcal{E}^{-1}$ and axiom~(T\ref{item:T2}) in the fourth equality. 
Comparing the two expressions shows that $\mathcal{E}^{-1}$ is a valid Bayesian inverse.
\end{proof}

\bibliographystyle{apsrev4-1} 
\bibliography{references} 

\begin{thebibliography}{153}%
\makeatletter
\providecommand \@ifxundefined [1]{%
 \@ifx{#1\undefined}
}%
\providecommand \@ifnum [1]{%
 \ifnum #1\expandafter \@firstoftwo
 \else \expandafter \@secondoftwo
 \fi
}%
\providecommand \@ifx [1]{%
 \ifx #1\expandafter \@firstoftwo
 \else \expandafter \@secondoftwo
 \fi
}%
\providecommand \natexlab [1]{#1}%
\providecommand \enquote  [1]{``#1''}%
\providecommand \bibnamefont  [1]{#1}%
\providecommand \bibfnamefont [1]{#1}%
\providecommand \citenamefont [1]{#1}%
\providecommand \href@noop [0]{\@secondoftwo}%
\providecommand \href [0]{\begingroup \@sanitize@url \@href}%
\providecommand \@href[1]{\@@startlink{#1}\@@href}%
\providecommand \@@href[1]{\endgroup#1\@@endlink}%
\providecommand \@sanitize@url [0]{\catcode `\\12\catcode `\$12\catcode
  `\&12\catcode `\#12\catcode `\^12\catcode `\_12\catcode `\%12\relax}%
\providecommand \@@startlink[1]{}%
\providecommand \@@endlink[0]{}%
\providecommand \url  [0]{\begingroup\@sanitize@url \@url }%
\providecommand \@url [1]{\endgroup\@href {#1}{\urlprefix }}%
\providecommand \urlprefix  [0]{URL }%
\providecommand \Eprint [0]{\href }%
\providecommand \doibase [0]{http://dx.doi.org/}%
\providecommand \selectlanguage [0]{\@gobble}%
\providecommand \bibinfo  [0]{\@secondoftwo}%
\providecommand \bibfield  [0]{\@secondoftwo}%
\providecommand \translation [1]{[#1]}%
\providecommand \BibitemOpen [0]{}%
\providecommand \bibitemStop [0]{}%
\providecommand \bibitemNoStop [0]{.\EOS\space}%
\providecommand \EOS [0]{\spacefactor3000\relax}%
\providecommand \BibitemShut  [1]{\csname bibitem#1\endcsname}%
\let\auto@bib@innerbib\@empty
\bibitem [{\citenamefont {Pearl}(1988)}]{Pearl88}%
  \BibitemOpen
  \bibfield  {author} {\bibinfo {author} {\bibfnamefont {J.}~\bibnamefont
  {Pearl}},\ }\href {\doibase 10.1016/C2009-0-27609-4} {\emph {\bibinfo {title}
  {{Probabilistic Reasoning in Intelligent Systems: Networks of Plausible
  Inference}}}}\ (\bibinfo  {publisher} {Elsevier},\ \bibinfo {year}
  {1988})\BibitemShut {NoStop}%
\bibitem [{\citenamefont {Balasubramanian}(1997)}]{Bala97}%
  \BibitemOpen
  \bibfield  {author} {\bibinfo {author} {\bibfnamefont {V.}~\bibnamefont
  {Balasubramanian}},\ }\href {\doibase 10.1162/neco.1997.9.2.349} {\bibfield
  {journal} {\bibinfo  {journal} {Neural Comput.}\ }\textbf {\bibinfo {volume}
  {9}},\ \bibinfo {pages} {349} (\bibinfo {year} {1997})},\ \Eprint
  {http://arxiv.org/abs/cond-mat/9601030} {arXiv:cond-mat/9601030 [cond-mat]}
  \BibitemShut {NoStop}%
\bibitem [{\citenamefont {Berry}(2006)}]{Be06}%
  \BibitemOpen
  \bibfield  {author} {\bibinfo {author} {\bibfnamefont {D.~A.}\ \bibnamefont
  {Berry}},\ }\href {\doibase 10.1038/nrd1927} {\bibfield  {journal} {\bibinfo
  {journal} {Nat. Rev. Drug Discov.}\ }\textbf {\bibinfo {volume} {5}},\
  \bibinfo {pages} {27} (\bibinfo {year} {2006})}\BibitemShut {NoStop}%
\bibitem [{\citenamefont {Friston}(2010)}]{Fr10}%
  \BibitemOpen
  \bibfield  {author} {\bibinfo {author} {\bibfnamefont {K.}~\bibnamefont
  {Friston}},\ }\href {\doibase 10.1038/nrn2787} {\bibfield  {journal}
  {\bibinfo  {journal} {Nat Rev Neurosci}\ }\textbf {\bibinfo {volume} {11}},\
  \bibinfo {pages} {127} (\bibinfo {year} {2010})}\BibitemShut {NoStop}%
\bibitem [{\citenamefont {Barnett}(2014)}]{Ba14}%
  \BibitemOpen
  \bibfield  {author} {\bibinfo {author} {\bibfnamefont {S.~M.}\ \bibnamefont
  {Barnett}},\ }\enquote {\bibinfo {title} {Quantum retrodiction},}\ in\ \href
  {\doibase 10.1007/978-3-319-04063-9_1} {\emph {\bibinfo {booktitle} {Quantum
  Information and Coherence}}},\ \bibinfo {editor} {edited by\ \bibinfo
  {editor} {\bibfnamefont {E.}~\bibnamefont {Andersson}}\ and\ \bibinfo
  {editor} {\bibfnamefont {P.}~\bibnamefont {{\"O}hberg}}}\ (\bibinfo
  {publisher} {Springer International Publishing},\ \bibinfo {address} {Cham},\
  \bibinfo {year} {2014})\ pp.\ \bibinfo {pages} {1--30}\BibitemShut {NoStop}%
\bibitem [{\citenamefont {Ghahramani}(2015)}]{Gh15}%
  \BibitemOpen
  \bibfield  {author} {\bibinfo {author} {\bibfnamefont {Z.}~\bibnamefont
  {Ghahramani}},\ }\href {\doibase 10.1038/nature14541} {\bibfield  {journal}
  {\bibinfo  {journal} {Nature}\ }\textbf {\bibinfo {volume} {521}},\ \bibinfo
  {pages} {452} (\bibinfo {year} {2015})}\BibitemShut {NoStop}%
\bibitem [{\citenamefont {Jacobs}(2019)}]{Ja19}%
  \BibitemOpen
  \bibfield  {author} {\bibinfo {author} {\bibfnamefont {B.}~\bibnamefont
  {Jacobs}},\ }\href {\doibase 10.1613/jair.1.11349} {\bibfield  {journal}
  {\bibinfo  {journal} {J. Artificial Intelligence Res.}\ }\textbf {\bibinfo
  {volume} {65}},\ \bibinfo {pages} {783} (\bibinfo {year} {2019})},\ \Eprint
  {http://arxiv.org/abs/1807.05609} {arXiv:1807.05609 [cs.AI]} \BibitemShut
  {NoStop}%
\bibitem [{\citenamefont {{Buscemi}}\ and\ \citenamefont
  {{Scarani}}(2021)}]{BuSc21}%
  \BibitemOpen
  \bibfield  {author} {\bibinfo {author} {\bibfnamefont {F.}~\bibnamefont
  {{Buscemi}}}\ and\ \bibinfo {author} {\bibfnamefont {V.}~\bibnamefont
  {{Scarani}}},\ }\href {\doibase 10.1103/PhysRevE.103.052111} {\bibfield
  {journal} {\bibinfo  {journal} {Phys. Rev. E}\ }\textbf {\bibinfo {volume}
  {103}},\ \bibinfo {eid} {052111} (\bibinfo {year} {2021})},\ \Eprint
  {http://arxiv.org/abs/2009.02849} {arXiv:2009.02849 [quant-ph]} \BibitemShut
  {NoStop}%
\bibitem [{\citenamefont {Aw}\ \emph {et~al.}(2021)\citenamefont {Aw},
  \citenamefont {Buscemi},\ and\ \citenamefont {Scarani}}]{AwBuSc21}%
  \BibitemOpen
  \bibfield  {author} {\bibinfo {author} {\bibfnamefont {C.~C.}\ \bibnamefont
  {Aw}}, \bibinfo {author} {\bibfnamefont {F.}~\bibnamefont {Buscemi}}, \ and\
  \bibinfo {author} {\bibfnamefont {V.}~\bibnamefont {Scarani}},\ }\href
  {\doibase 10.1116/5.0060893} {\bibfield  {journal} {\bibinfo  {journal}
  {{AVS} Quantum Science}\ }\textbf {\bibinfo {volume} {3}},\ \bibinfo {pages}
  {045601} (\bibinfo {year} {2021})},\ \Eprint
  {http://arxiv.org/abs/2106.08589} {arXiv:2106.08589 [cond-mat.stat-mech]}
  \BibitemShut {NoStop}%
\bibitem [{\citenamefont {{St. Clere Smithe}}(2021)}]{Sm21}%
  \BibitemOpen
  \bibfield  {author} {\bibinfo {author} {\bibfnamefont {T.}~\bibnamefont {{St.
  Clere Smithe}}},\ }\href {\doibase 10.48550/arXiv.2109.04461} {\enquote
  {\bibinfo {title} {{Compositional Active Inference I: Bayesian Lenses.
  Statistical Games}},}\ } (\bibinfo {year} {2021}),\ \Eprint
  {http://arxiv.org/abs/2109.04461} {arXiv:2109.04461 [math.ST]} \BibitemShut
  {NoStop}%
\bibitem [{\citenamefont {Bub}(1977)}]{Bu77}%
  \BibitemOpen
  \bibfield  {author} {\bibinfo {author} {\bibfnamefont {J.}~\bibnamefont
  {Bub}},\ }\href {\doibase 10.1007/BF00262075} {\bibfield  {journal} {\bibinfo
   {journal} {J. Philos. Log.}\ }\textbf {\bibinfo {volume} {6}},\ \bibinfo
  {pages} {381} (\bibinfo {year} {1977})}\BibitemShut {NoStop}%
\bibitem [{\citenamefont {L{\"u}ders}(2006)}]{Lu06}%
  \BibitemOpen
  \bibfield  {author} {\bibinfo {author} {\bibfnamefont {G.}~\bibnamefont
  {L{\"u}ders}},\ }\href@noop {} {\bibfield  {journal} {\bibinfo  {journal}
  {Annalen der Physik}\ }\textbf {\bibinfo {volume} {15}},\ \bibinfo {pages}
  {663} (\bibinfo {year} {2006})}\BibitemShut {NoStop}%
\bibitem [{\citenamefont {Ozawa}(1997)}]{Oz97}%
  \BibitemOpen
  \bibfield  {author} {\bibinfo {author} {\bibfnamefont {M.}~\bibnamefont
  {Ozawa}},\ }in\ \href {\doibase 10.1007/978-1-4615-5923-8_25} {\emph
  {\bibinfo {booktitle} {Quantum Communication, Computing, and Measurement}}}\
  (\bibinfo  {publisher} {Springer},\ \bibinfo {year} {1997})\ pp.\ \bibinfo
  {pages} {233--241},\ \Eprint {http://arxiv.org/abs/quant-ph/9705030}
  {arXiv:quant-ph/9705030 [quant-ph]} \BibitemShut {NoStop}%
\bibitem [{\citenamefont {Korotkov}(1999)}]{Ko99}%
  \BibitemOpen
  \bibfield  {author} {\bibinfo {author} {\bibfnamefont {A.~N.}\ \bibnamefont
  {Korotkov}},\ }\href {\doibase 10.1103/PhysRevB.60.5737} {\bibfield
  {journal} {\bibinfo  {journal} {Phys. Rev. B}\ }\textbf {\bibinfo {volume}
  {60}},\ \bibinfo {pages} {5737} (\bibinfo {year} {1999})}\BibitemShut
  {NoStop}%
\bibitem [{\citenamefont {Korotkov}(2001)}]{Ko01}%
  \BibitemOpen
  \bibfield  {author} {\bibinfo {author} {\bibfnamefont {A.~N.}\ \bibnamefont
  {Korotkov}},\ }\href {\doibase 10.1103/PhysRevB.63.115403} {\bibfield
  {journal} {\bibinfo  {journal} {Phys. Rev. B}\ }\textbf {\bibinfo {volume}
  {63}},\ \bibinfo {pages} {115403} (\bibinfo {year} {2001})}\BibitemShut
  {NoStop}%
\bibitem [{\citenamefont {Schack}\ \emph {et~al.}(2001)\citenamefont {Schack},
  \citenamefont {Brun},\ and\ \citenamefont {Caves}}]{SBC01}%
  \BibitemOpen
  \bibfield  {author} {\bibinfo {author} {\bibfnamefont {R.}~\bibnamefont
  {Schack}}, \bibinfo {author} {\bibfnamefont {T.~A.}\ \bibnamefont {Brun}}, \
  and\ \bibinfo {author} {\bibfnamefont {C.~M.}\ \bibnamefont {Caves}},\ }\href
  {\doibase 10.1103/PhysRevA.64.014305} {\bibfield  {journal} {\bibinfo
  {journal} {Phys. Rev. A}\ }\textbf {\bibinfo {volume} {64}},\ \bibinfo
  {pages} {014305} (\bibinfo {year} {2001})},\ \Eprint
  {http://arxiv.org/abs/quant-ph/0008113} {arXiv:quant-ph/0008113 [quant-ph]}
  \BibitemShut {NoStop}%
\bibitem [{\citenamefont {Fuchs}(2001)}]{Fu01}%
  \BibitemOpen
  \bibfield  {author} {\bibinfo {author} {\bibfnamefont {C.~A.}\ \bibnamefont
  {Fuchs}},\ }in\ \href@noop {} {\emph {\bibinfo {booktitle} {Decoherence and
  its implications in quantum computation and information transfer}}},\ Vol.\
  \bibinfo {volume} {182},\ \bibinfo {editor} {edited by\ \bibinfo {editor}
  {\bibfnamefont {G.}~\bibnamefont {{Tony}}}\ and\ \bibinfo {editor}
  {\bibfnamefont {P.~E.}\ \bibnamefont {{Turchi}}}}\ (\bibinfo  {publisher}
  {IOS Press},\ \bibinfo {year} {2001})\ pp.\ \bibinfo {pages} {38--82},\
  \Eprint {http://arxiv.org/abs/quant-ph/0106166} {arXiv:quant-ph/0106166
  [quant-ph]} \BibitemShut {NoStop}%
\bibitem [{\citenamefont {Warmuth}(2005)}]{Wa05}%
  \BibitemOpen
  \bibfield  {author} {\bibinfo {author} {\bibfnamefont {M.~K.}\ \bibnamefont
  {Warmuth}},\ }in\ \href
  {https://proceedings.neurips.cc/paper/2005/file/4191ef5f6c1576762869ac49281130c9-Paper.pdf}
  {\emph {\bibinfo {booktitle} {Advances in Neural Information Processing
  Systems}}},\ Vol.~\bibinfo {volume} {18},\ \bibinfo {editor} {edited by\
  \bibinfo {editor} {\bibfnamefont {Y.}~\bibnamefont {Weiss}}, \bibinfo
  {editor} {\bibfnamefont {B.}~\bibnamefont {Sch\"{o}lkopf}}, \ and\ \bibinfo
  {editor} {\bibfnamefont {J.}~\bibnamefont {Platt}}}\ (\bibinfo  {publisher}
  {MIT Press},\ \bibinfo {year} {2005})\BibitemShut {NoStop}%
\bibitem [{\citenamefont {Warmuth}\ and\ \citenamefont
  {Kuzmin}(2010)}]{WaKu10}%
  \BibitemOpen
  \bibfield  {author} {\bibinfo {author} {\bibfnamefont {M.~K.}\ \bibnamefont
  {Warmuth}}\ and\ \bibinfo {author} {\bibfnamefont {D.}~\bibnamefont
  {Kuzmin}},\ }\href {\doibase 10.1007/s10994-009-5133-7} {\bibfield  {journal}
  {\bibinfo  {journal} {Mach Learn}\ }\textbf {\bibinfo {volume} {78}},\
  \bibinfo {pages} {63} (\bibinfo {year} {2010})}\BibitemShut {NoStop}%
\bibitem [{\citenamefont {Leifer}(2006)}]{Le06}%
  \BibitemOpen
  \bibfield  {author} {\bibinfo {author} {\bibfnamefont {M.~S.}\ \bibnamefont
  {Leifer}},\ }\href {\doibase 10.1103/PhysRevA.74.042310} {\bibfield
  {journal} {\bibinfo  {journal} {Phys. Rev. A}\ }\textbf {\bibinfo {volume}
  {74}},\ \bibinfo {pages} {042310} (\bibinfo {year} {2006})},\ \Eprint
  {http://arxiv.org/abs/0606022} {arXiv:0606022 [quant-ph]} \BibitemShut
  {NoStop}%
\bibitem [{\citenamefont {{Leifer}}(2007)}]{Le07}%
  \BibitemOpen
  \bibfield  {author} {\bibinfo {author} {\bibfnamefont {M.~S.}\ \bibnamefont
  {{Leifer}}},\ }in\ \href {\doibase 10.1063/1.2713456} {\emph {\bibinfo
  {booktitle} {Foundations of Probability and Physics - 4}}},\ \bibinfo
  {series} {American Institute of Physics Conference Series}, Vol.\ \bibinfo
  {volume} {889},\ \bibinfo {editor} {edited by\ \bibinfo {editor}
  {\bibfnamefont {G.}~\bibnamefont {{Adenier}}}, \bibinfo {editor}
  {\bibfnamefont {C.}~\bibnamefont {{Fuchs}}}, \ and\ \bibinfo {editor}
  {\bibfnamefont {A.~Y.}\ \bibnamefont {{Khrennikov}}}}\ (\bibinfo {year}
  {2007})\ pp.\ \bibinfo {pages} {172--186},\ \Eprint
  {http://arxiv.org/abs/quant-ph/0611233} {arXiv:quant-ph/0611233 [quant-ph]}
  \BibitemShut {NoStop}%
\bibitem [{\citenamefont {Coecke}\ and\ \citenamefont
  {Spekkens}(2012)}]{CoSp12}%
  \BibitemOpen
  \bibfield  {author} {\bibinfo {author} {\bibfnamefont {B.}~\bibnamefont
  {Coecke}}\ and\ \bibinfo {author} {\bibfnamefont {R.~W.}\ \bibnamefont
  {Spekkens}},\ }\href {\doibase 10.1007/s11229-011-9917-5} {\bibfield
  {journal} {\bibinfo  {journal} {Synthese}\ }\textbf {\bibinfo {volume}
  {186}},\ \bibinfo {pages} {651} (\bibinfo {year} {2012})}\BibitemShut
  {NoStop}%
\bibitem [{\citenamefont {Tegmark}(2012)}]{Te12}%
  \BibitemOpen
  \bibfield  {author} {\bibinfo {author} {\bibfnamefont {M.}~\bibnamefont
  {Tegmark}},\ }\href {\doibase 10.1103/PhysRevD.85.123517} {\bibfield
  {journal} {\bibinfo  {journal} {Phys. Rev. D}\ }\textbf {\bibinfo {volume}
  {85}},\ \bibinfo {pages} {123517} (\bibinfo {year} {2012})},\ \Eprint
  {http://arxiv.org/abs/1108.3080} {arXiv:1108.3080 [hep-th]} \BibitemShut
  {NoStop}%
\bibitem [{\citenamefont {Farenick}\ and\ \citenamefont
  {Kozdron}(2012)}]{FaKo12}%
  \BibitemOpen
  \bibfield  {author} {\bibinfo {author} {\bibfnamefont {D.}~\bibnamefont
  {Farenick}}\ and\ \bibinfo {author} {\bibfnamefont {M.~J.}\ \bibnamefont
  {Kozdron}},\ }\href {\doibase 10.1063/1.3703069} {\bibfield  {journal}
  {\bibinfo  {journal} {J. Math. Phys}\ }\textbf {\bibinfo {volume} {53}},\
  \bibinfo {pages} {042201} (\bibinfo {year} {2012})},\ \Eprint
  {http://arxiv.org/abs/1111.5638} {arXiv:1111.5638 [math.PR]} \BibitemShut
  {NoStop}%
\bibitem [{\citenamefont {Leifer}\ and\ \citenamefont
  {Spekkens}(2013)}]{LeSp13}%
  \BibitemOpen
  \bibfield  {author} {\bibinfo {author} {\bibfnamefont {M.~S.}\ \bibnamefont
  {Leifer}}\ and\ \bibinfo {author} {\bibfnamefont {R.~W.}\ \bibnamefont
  {Spekkens}},\ }\href {\doibase 10.1103/PhysRevA.88.052130} {\bibfield
  {journal} {\bibinfo  {journal} {Phys. Rev. A}\ }\textbf {\bibinfo {volume}
  {88}},\ \bibinfo {pages} {052130} (\bibinfo {year} {2013})},\ \Eprint
  {http://arxiv.org/abs/1107.5849} {arXiv:1107.5849 [quant-ph]} \BibitemShut
  {NoStop}%
\bibitem [{\citenamefont {Cotler}\ \emph
  {et~al.}(2019{\natexlab{a}})\citenamefont {Cotler}, \citenamefont {Hayden},
  \citenamefont {Penington}, \citenamefont {Salton}, \citenamefont {Swingle},\
  and\ \citenamefont {Walter}}]{CHPSSW19}%
  \BibitemOpen
  \bibfield  {author} {\bibinfo {author} {\bibfnamefont {J.}~\bibnamefont
  {Cotler}}, \bibinfo {author} {\bibfnamefont {P.}~\bibnamefont {Hayden}},
  \bibinfo {author} {\bibfnamefont {G.}~\bibnamefont {Penington}}, \bibinfo
  {author} {\bibfnamefont {G.}~\bibnamefont {Salton}}, \bibinfo {author}
  {\bibfnamefont {B.}~\bibnamefont {Swingle}}, \ and\ \bibinfo {author}
  {\bibfnamefont {M.}~\bibnamefont {Walter}},\ }\href {\doibase
  10.1103/PhysRevX.9.031011} {\bibfield  {journal} {\bibinfo  {journal} {Phys.
  Rev. X}\ }\textbf {\bibinfo {volume} {9}},\ \bibinfo {pages} {031011}
  (\bibinfo {year} {2019}{\natexlab{a}})}\BibitemShut {NoStop}%
\bibitem [{\citenamefont {Parzygnat}\ and\ \citenamefont
  {Russo}(2019)}]{PaRu19}%
  \BibitemOpen
  \bibfield  {author} {\bibinfo {author} {\bibfnamefont {A.~J.}\ \bibnamefont
  {Parzygnat}}\ and\ \bibinfo {author} {\bibfnamefont {B.~P.}\ \bibnamefont
  {Russo}},\ }\href@noop {} {\enquote {\bibinfo {title} {Non-commutative
  disintegrations: existence and uniqueness in finite dimensions},}\ }
  (\bibinfo {year} {2019}),\ \bibinfo {note} {to appear in J. Noncommut.
  Geom.},\ \Eprint {http://arxiv.org/abs/1907.09689} {arXiv:1907.09689
  [quant-ph]} \BibitemShut {NoStop}%
\bibitem [{\citenamefont {Parzygnat}(2020)}]{PaBayes}%
  \BibitemOpen
  \bibfield  {author} {\bibinfo {author} {\bibfnamefont {A.~J.}\ \bibnamefont
  {Parzygnat}},\ }\href@noop {} {\enquote {\bibinfo {title} {Inverses,
  disintegrations, and {B}ayesian inversion in quantum {M}arkov categories},}\
  } (\bibinfo {year} {2020}),\ \Eprint {http://arxiv.org/abs/2001.08375}
  {arXiv:2001.08375 [quant-ph]} \BibitemShut {NoStop}%
\bibitem [{\citenamefont {Parzygnat}\ and\ \citenamefont
  {Russo}(2022)}]{PaRuBayes}%
  \BibitemOpen
  \bibfield  {author} {\bibinfo {author} {\bibfnamefont {A.~J.}\ \bibnamefont
  {Parzygnat}}\ and\ \bibinfo {author} {\bibfnamefont {B.~P.}\ \bibnamefont
  {Russo}},\ }\href {\doibase https://doi.org/10.1016/j.laa.2022.02.030}
  {\bibfield  {journal} {\bibinfo  {journal} {Linear Algebra Its Appl.}\
  }\textbf {\bibinfo {volume} {644}},\ \bibinfo {pages} {28} (\bibinfo {year}
  {2022})},\ \Eprint {http://arxiv.org/abs/2005.03886} {arXiv:2005.03886
  [quant-ph]} \BibitemShut {NoStop}%
\bibitem [{\citenamefont {Giorgetti}\ \emph {et~al.}(2021)\citenamefont
  {Giorgetti}, \citenamefont {Parzygnat}, \citenamefont {Ranallo},\ and\
  \citenamefont {Russo}}]{GPRR21}%
  \BibitemOpen
  \bibfield  {author} {\bibinfo {author} {\bibfnamefont {L.}~\bibnamefont
  {Giorgetti}}, \bibinfo {author} {\bibfnamefont {A.~J.}\ \bibnamefont
  {Parzygnat}}, \bibinfo {author} {\bibfnamefont {A.}~\bibnamefont {Ranallo}},
  \ and\ \bibinfo {author} {\bibfnamefont {B.~P.}\ \bibnamefont {Russo}},\
  }\href@noop {} {\enquote {\bibinfo {title} {Bayesian inversion and the
  {T}omita--{T}akesaki modular group},}\ } (\bibinfo {year} {2021}),\ \Eprint
  {http://arxiv.org/abs/2112.03129} {arXiv:2112.03129 [math.OA]} \BibitemShut
  {NoStop}%
\bibitem [{\citenamefont {Chru{\'s}ci{\'n}ski}\ and\ \citenamefont
  {Matsuoka}(2020)}]{ChMa20}%
  \BibitemOpen
  \bibfield  {author} {\bibinfo {author} {\bibfnamefont {D.}~\bibnamefont
  {Chru{\'s}ci{\'n}ski}}\ and\ \bibinfo {author} {\bibfnamefont
  {T.}~\bibnamefont {Matsuoka}},\ }\href {\doibase
  10.1016/S0034-4877(20)30060-4} {\bibfield  {journal} {\bibinfo  {journal}
  {Rep. Math. Phys.}\ }\textbf {\bibinfo {volume} {86}},\ \bibinfo {pages}
  {115} (\bibinfo {year} {2020})}\BibitemShut {NoStop}%
\bibitem [{\citenamefont {Tsang}(2022{\natexlab{a}})}]{Ts22}%
  \BibitemOpen
  \bibfield  {author} {\bibinfo {author} {\bibfnamefont {M.}~\bibnamefont
  {Tsang}},\ }\href {\doibase 10.1103/PhysRevA.105.042213} {\bibfield
  {journal} {\bibinfo  {journal} {Phys. Rev. A}\ }\textbf {\bibinfo {volume}
  {105}},\ \bibinfo {pages} {042213} (\bibinfo {year} {2022}{\natexlab{a}})},\
  \Eprint {http://arxiv.org/abs/1912.02711} {arXiv:1912.02711 [quant-ph]}
  \BibitemShut {NoStop}%
\bibitem [{\citenamefont {Li}(2022)}]{Liu22}%
  \BibitemOpen
  \bibfield  {author} {\bibinfo {author} {\bibfnamefont {H.}~\bibnamefont
  {Li}},\ }\href@noop {} {\enquote {\bibinfo {title} {Quantum {B}ayesian
  statistical inference},}\ } (\bibinfo {year} {2022}),\ \Eprint
  {http://arxiv.org/abs/2204.08845} {arXiv:2204.08845 [quant-ph]} \BibitemShut
  {NoStop}%
\bibitem [{\citenamefont {Matsuoka}\ and\ \citenamefont
  {Chru{\'s}ci{\'n}ski}(2022)}]{MaCh22}%
  \BibitemOpen
  \bibfield  {author} {\bibinfo {author} {\bibfnamefont {T.}~\bibnamefont
  {Matsuoka}}\ and\ \bibinfo {author} {\bibfnamefont {D.}~\bibnamefont
  {Chru{\'s}ci{\'n}ski}},\ }in\ \href {\doibase 10.1007/978-3-031-06170-7_7}
  {\emph {\bibinfo {booktitle} {International Conference on Quantum Probability
  \& Related Topics}}}\ (\bibinfo {organization} {Springer, Cham},\ \bibinfo
  {year} {2022})\ pp.\ \bibinfo {pages} {135--150}\BibitemShut {NoStop}%
\bibitem [{\citenamefont {Bera}\ and\ \citenamefont {Bera}(2022)}]{BeBe22}%
  \BibitemOpen
  \bibfield  {author} {\bibinfo {author} {\bibfnamefont {M.~L.}\ \bibnamefont
  {Bera}}\ and\ \bibinfo {author} {\bibfnamefont {M.~N.}\ \bibnamefont
  {Bera}},\ }\href@noop {} {\enquote {\bibinfo {title} {Quantum {B}ayes' rule
  affirms consistency in measurement inferences in quantum mechanics},}\ }
  (\bibinfo {year} {2022}),\ \Eprint {http://arxiv.org/abs/2207.08623}
  {arXiv:2207.08623 [quant-ph]} \BibitemShut {NoStop}%
\bibitem [{\citenamefont {Fitzsimons}\ \emph {et~al.}(2015)\citenamefont
  {Fitzsimons}, \citenamefont {Jones},\ and\ \citenamefont {Vedral}}]{FJV15}%
  \BibitemOpen
  \bibfield  {author} {\bibinfo {author} {\bibfnamefont {J.~F.}\ \bibnamefont
  {Fitzsimons}}, \bibinfo {author} {\bibfnamefont {J.~A.}\ \bibnamefont
  {Jones}}, \ and\ \bibinfo {author} {\bibfnamefont {V.}~\bibnamefont
  {Vedral}},\ }\href {\doibase 10.1038/srep18281} {\bibfield  {journal}
  {\bibinfo  {journal} {Sci. Rep.}\ }\textbf {\bibinfo {volume} {5}},\ \bibinfo
  {pages} {18281} (\bibinfo {year} {2015})},\ \Eprint
  {http://arxiv.org/abs/1302.2731} {arXiv:1302.2731 [quant-ph]} \BibitemShut
  {NoStop}%
\bibitem [{\citenamefont {Fullwood}\ and\ \citenamefont
  {Parzygnat}(2022)}]{FuPa22}%
  \BibitemOpen
  \bibfield  {author} {\bibinfo {author} {\bibfnamefont {J.}~\bibnamefont
  {Fullwood}}\ and\ \bibinfo {author} {\bibfnamefont {A.~J.}\ \bibnamefont
  {Parzygnat}},\ }\href {\doibase 10.1098/rspa.2022.0104} {\bibfield  {journal}
  {\bibinfo  {journal} {Proc. R. Soc. A}\ }\textbf {\bibinfo {volume} {478}}
  (\bibinfo {year} {2022}),\ 10.1098/rspa.2022.0104},\ \Eprint
  {http://arxiv.org/abs/2202.03607} {arXiv:2202.03607 [quant-ph]} \BibitemShut
  {NoStop}%
\bibitem [{\citenamefont {{Horsman}}\ \emph {et~al.}(2017)\citenamefont
  {{Horsman}}, \citenamefont {{Heunen}}, \citenamefont {{Pusey}}, \citenamefont
  {{Barrett}},\ and\ \citenamefont {{Spekkens}}}]{HHPBS17}%
  \BibitemOpen
  \bibfield  {author} {\bibinfo {author} {\bibfnamefont {D.}~\bibnamefont
  {{Horsman}}}, \bibinfo {author} {\bibfnamefont {C.}~\bibnamefont {{Heunen}}},
  \bibinfo {author} {\bibfnamefont {M.~F.}\ \bibnamefont {{Pusey}}}, \bibinfo
  {author} {\bibfnamefont {J.}~\bibnamefont {{Barrett}}}, \ and\ \bibinfo
  {author} {\bibfnamefont {R.~W.}\ \bibnamefont {{Spekkens}}},\ }\href
  {\doibase 10.1098/rspa.2017.0395} {\bibfield  {journal} {\bibinfo  {journal}
  {Proc. R. Soc. A}\ }\textbf {\bibinfo {volume} {473}},\ \bibinfo {eid}
  {20170395} (\bibinfo {year} {2017})},\ \Eprint
  {http://arxiv.org/abs/1607.03637} {arXiv:1607.03637 [quant-ph]} \BibitemShut
  {NoStop}%
\bibitem [{\citenamefont {von Neumann}(2018)}]{vN18}%
  \BibitemOpen
  \bibfield  {author} {\bibinfo {author} {\bibfnamefont {J.}~\bibnamefont {von
  Neumann}},\ }\href {\doibase 10.1515/9781400889921} {\emph {\bibinfo {title}
  {Mathematical foundations of quantum mechanics: {N}ew edition}}}\ (\bibinfo
  {publisher} {Princeton university press},\ \bibinfo {year}
  {2018})\BibitemShut {NoStop}%
\bibitem [{\citenamefont {Kraus}(1983)}]{Kr83}%
  \BibitemOpen
  \bibfield  {author} {\bibinfo {author} {\bibfnamefont {K.}~\bibnamefont
  {Kraus}},\ }\href@noop {} {\emph {\bibinfo {title} {States, Effects, and
  Operations: Fundamental Notions of Quantum Theory}}},\ Lecture Notes in
  Physics Vol. 190\ (\bibinfo  {publisher} {Springer Berlin Heidelberg},\
  \bibinfo {year} {1983})\BibitemShut {NoStop}%
\bibitem [{\citenamefont {Watanabe}(1955)}]{Wat55}%
  \BibitemOpen
  \bibfield  {author} {\bibinfo {author} {\bibfnamefont {S.}~\bibnamefont
  {Watanabe}},\ }\href {\doibase 10.1103/RevModPhys.27.179} {\bibfield
  {journal} {\bibinfo  {journal} {Rev. Mod. Phys.}\ }\textbf {\bibinfo {volume}
  {27}},\ \bibinfo {pages} {179} (\bibinfo {year} {1955})}\BibitemShut
  {NoStop}%
\bibitem [{\citenamefont {Reznik}\ and\ \citenamefont
  {Aharonov}(1995)}]{ReAh95}%
  \BibitemOpen
  \bibfield  {author} {\bibinfo {author} {\bibfnamefont {B.}~\bibnamefont
  {Reznik}}\ and\ \bibinfo {author} {\bibfnamefont {Y.}~\bibnamefont
  {Aharonov}},\ }\href {\doibase 10.1103/PhysRevA.52.2538} {\bibfield
  {journal} {\bibinfo  {journal} {Phys. Rev. A}\ }\textbf {\bibinfo {volume}
  {52}},\ \bibinfo {pages} {2538} (\bibinfo {year} {1995})},\ \Eprint
  {http://arxiv.org/abs/quant-ph/9501011} {arXiv:quant-ph/9501011 [quant-ph]}
  \BibitemShut {NoStop}%
\bibitem [{\citenamefont {Di~Biagio}\ \emph {et~al.}(2021)\citenamefont
  {Di~Biagio}, \citenamefont {Don{\`{a}}},\ and\ \citenamefont
  {Rovelli}}]{DDR21}%
  \BibitemOpen
  \bibfield  {author} {\bibinfo {author} {\bibfnamefont {A.}~\bibnamefont
  {Di~Biagio}}, \bibinfo {author} {\bibfnamefont {P.}~\bibnamefont
  {Don{\`{a}}}}, \ and\ \bibinfo {author} {\bibfnamefont {C.}~\bibnamefont
  {Rovelli}},\ }\href {\doibase 10.22331/q-2021-08-09-520} {\bibfield
  {journal} {\bibinfo  {journal} {Quantum}\ }\textbf {\bibinfo {volume} {5}}
  (\bibinfo {year} {2021}),\ 10.22331/q-2021-08-09-520},\ \Eprint
  {http://arxiv.org/abs/2010.05734} {arXiv:2010.05734 [quant-ph]} \BibitemShut
  {NoStop}%
\bibitem [{\citenamefont {{Parzygnat}}\ and\ \citenamefont
  {{Buscemi}}(2022)}]{PaBu22}%
  \BibitemOpen
  \bibfield  {author} {\bibinfo {author} {\bibfnamefont {A.~J.}\ \bibnamefont
  {{Parzygnat}}}\ and\ \bibinfo {author} {\bibfnamefont {F.}~\bibnamefont
  {{Buscemi}}},\ }\href@noop {} {\enquote {\bibinfo {title} {{Axioms for
  retrodiction: achieving time-reversal symmetry with a prior}},}\ } (\bibinfo
  {year} {2022}),\ \Eprint {http://arxiv.org/abs/2210.13531} {arXiv:2210.13531
  [quant-ph]} \BibitemShut {NoStop}%
\bibitem [{\citenamefont {{Coecke}}\ \emph {et~al.}(2017)\citenamefont
  {{Coecke}}, \citenamefont {{Gogioso}},\ and\ \citenamefont
  {{Selby}}}]{CGS17}%
  \BibitemOpen
  \bibfield  {author} {\bibinfo {author} {\bibfnamefont {B.}~\bibnamefont
  {{Coecke}}}, \bibinfo {author} {\bibfnamefont {S.}~\bibnamefont {{Gogioso}}},
  \ and\ \bibinfo {author} {\bibfnamefont {J.~H.}\ \bibnamefont {{Selby}}},\
  }\href@noop {} {\enquote {\bibinfo {title} {{The time-reverse of any causal
  theory is eternal noise}},}\ } (\bibinfo {year} {2017}),\ \Eprint
  {http://arxiv.org/abs/1711.05511} {arXiv:1711.05511 [quant-ph]} \BibitemShut
  {NoStop}%
\bibitem [{\citenamefont {Chiribella}\ \emph {et~al.}(2021)\citenamefont
  {Chiribella}, \citenamefont {Aurell},\ and\ \citenamefont
  {\.{Z}yczkowski}}]{CAZ21}%
  \BibitemOpen
  \bibfield  {author} {\bibinfo {author} {\bibfnamefont {G.}~\bibnamefont
  {Chiribella}}, \bibinfo {author} {\bibfnamefont {E.}~\bibnamefont {Aurell}},
  \ and\ \bibinfo {author} {\bibfnamefont {K.}~\bibnamefont {\.{Z}yczkowski}},\
  }\href {\doibase 10.1103/PhysRevResearch.3.033028} {\bibfield  {journal}
  {\bibinfo  {journal} {Phys. Rev. Research}\ }\textbf {\bibinfo {volume}
  {3}},\ \bibinfo {pages} {033028} (\bibinfo {year} {2021})},\ \Eprint
  {http://arxiv.org/abs/2101.04962} {arXiv:2101.04962 [quant-ph]} \BibitemShut
  {NoStop}%
\bibitem [{\citenamefont {Selby}\ \emph {et~al.}(2022)\citenamefont {Selby},
  \citenamefont {Stasinou}, \citenamefont {Gogioso},\ and\ \citenamefont
  {Coecke}}]{SSGC22}%
  \BibitemOpen
  \bibfield  {author} {\bibinfo {author} {\bibfnamefont {J.~H.}\ \bibnamefont
  {Selby}}, \bibinfo {author} {\bibfnamefont {M.~E.}\ \bibnamefont {Stasinou}},
  \bibinfo {author} {\bibfnamefont {S.}~\bibnamefont {Gogioso}}, \ and\
  \bibinfo {author} {\bibfnamefont {B.}~\bibnamefont {Coecke}},\ }\href
  {\doibase 10.48550/ARXIV.2209.07867} {\enquote {\bibinfo {title} {Time
  symmetry in quantum theories and beyond},}\ } (\bibinfo {year} {2022}),\
  \Eprint {http://arxiv.org/abs/2209.07867} {arXiv:2209.07867 [quant-ph]}
  \BibitemShut {NoStop}%
\bibitem [{\citenamefont {Aharonov}\ \emph {et~al.}(1964)\citenamefont
  {Aharonov}, \citenamefont {Bergmann},\ and\ \citenamefont
  {Lebowitz}}]{ABL64}%
  \BibitemOpen
  \bibfield  {author} {\bibinfo {author} {\bibfnamefont {Y.}~\bibnamefont
  {Aharonov}}, \bibinfo {author} {\bibfnamefont {P.~G.}\ \bibnamefont
  {Bergmann}}, \ and\ \bibinfo {author} {\bibfnamefont {J.~L.}\ \bibnamefont
  {Lebowitz}},\ }\href {\doibase 10.1103/PhysRev.134.B1410} {\bibfield
  {journal} {\bibinfo  {journal} {Phys. Rev.}\ }\textbf {\bibinfo {volume}
  {134}},\ \bibinfo {pages} {B1410} (\bibinfo {year} {1964})}\BibitemShut
  {NoStop}%
\bibitem [{\citenamefont {Buscemi}\ \emph {et~al.}(2013)\citenamefont
  {Buscemi}, \citenamefont {Dall'Arno}, \citenamefont {Ozawa},\ and\
  \citenamefont {Vedral}}]{BDOV13}%
  \BibitemOpen
  \bibfield  {author} {\bibinfo {author} {\bibfnamefont {F.}~\bibnamefont
  {Buscemi}}, \bibinfo {author} {\bibfnamefont {M.}~\bibnamefont {Dall'Arno}},
  \bibinfo {author} {\bibfnamefont {M.}~\bibnamefont {Ozawa}}, \ and\ \bibinfo
  {author} {\bibfnamefont {V.}~\bibnamefont {Vedral}},\ }\href {\doibase
  10.48550/ARXIV.1312.4240} {\enquote {\bibinfo {title} {Direct observation of
  any two-point quantum correlation function},}\ } (\bibinfo {year} {2013}),\
  \Eprint {http://arxiv.org/abs/1312.4240} {arXiv:1312.4240 [quant-ph]}
  \BibitemShut {NoStop}%
\bibitem [{\citenamefont {Bianchi}\ \emph {et~al.}(2017)\citenamefont
  {Bianchi}, \citenamefont {Haggard},\ and\ \citenamefont {Rovelli}}]{BHR17}%
  \BibitemOpen
  \bibfield  {author} {\bibinfo {author} {\bibfnamefont {E.}~\bibnamefont
  {Bianchi}}, \bibinfo {author} {\bibfnamefont {H.~M.}\ \bibnamefont
  {Haggard}}, \ and\ \bibinfo {author} {\bibfnamefont {C.}~\bibnamefont
  {Rovelli}},\ }\href {\doibase 10.1007/s10714-017-2263-2} {\bibfield
  {journal} {\bibinfo  {journal} {Gen. Relativ. Gravit}\ }\textbf {\bibinfo
  {volume} {49}} (\bibinfo {year} {2017}),\ 10.1007/s10714-017-2263-2},\
  \Eprint {http://arxiv.org/abs/1306.5206} {arXiv:1306.5206 [gr-qc]}
  \BibitemShut {NoStop}%
\bibitem [{\citenamefont {Petz}(1984)}]{Pe84}%
  \BibitemOpen
  \bibfield  {author} {\bibinfo {author} {\bibfnamefont {D.}~\bibnamefont
  {Petz}},\ }\href@noop {} {\bibfield  {journal} {\bibinfo  {journal} {Q. J.
  Math.}\ }\textbf {\bibinfo {volume} {35}},\ \bibinfo {pages} {475} (\bibinfo
  {year} {1984})}\BibitemShut {NoStop}%
\bibitem [{\citenamefont {Ohya}(1983{\natexlab{a}})}]{Oh83a}%
  \BibitemOpen
  \bibfield  {author} {\bibinfo {author} {\bibfnamefont {M.}~\bibnamefont
  {Ohya}},\ }\href {\doibase 10.1007/BF02789599} {\bibfield  {journal}
  {\bibinfo  {journal} {Lett. Nuovo Cimento}\ }\textbf {\bibinfo {volume}
  {38}},\ \bibinfo {pages} {402} (\bibinfo {year}
  {1983}{\natexlab{a}})}\BibitemShut {NoStop}%
\bibitem [{\citenamefont {Ohya}(1983{\natexlab{b}})}]{Oh83b}%
  \BibitemOpen
  \bibfield  {author} {\bibinfo {author} {\bibfnamefont {M.}~\bibnamefont
  {Ohya}},\ }\href {\doibase 10.1109/TIT.1983.1056719} {\bibfield  {journal}
  {\bibinfo  {journal} {IEEE Trans. Inf. Theory}\ }\textbf {\bibinfo {volume}
  {29}},\ \bibinfo {pages} {770} (\bibinfo {year}
  {1983}{\natexlab{b}})}\BibitemShut {NoStop}%
\bibitem [{\citenamefont {Ozawa}(1998)}]{Oz98}%
  \BibitemOpen
  \bibfield  {author} {\bibinfo {author} {\bibfnamefont {M.}~\bibnamefont
  {Ozawa}},\ }\href {\doibase
  10.1002/(SICI)1521-3978(199811)46:6/8<615::AID-PROP615>3.0.CO;2-D} {\bibfield
   {journal} {\bibinfo  {journal} {Fortschr. Phys.}\ }\textbf {\bibinfo
  {volume} {46}},\ \bibinfo {pages} {615} (\bibinfo {year} {1998})},\ \Eprint
  {http://arxiv.org/abs/quant-ph/9711006} {arXiv:quant-ph/9711006 [quant-ph]}
  \BibitemShut {NoStop}%
\bibitem [{\citenamefont {Farenick}(2001)}]{Fa01}%
  \BibitemOpen
  \bibfield  {author} {\bibinfo {author} {\bibfnamefont {D.~R.}\ \bibnamefont
  {Farenick}},\ }\href {\doibase 10.1007/978-1-4613-0097-7} {\emph {\bibinfo
  {title} {Algebras of linear transformations}}},\ Universitext\ (\bibinfo
  {publisher} {Springer-Verlag, New York},\ \bibinfo {year} {2001})\ pp.\
  \bibinfo {pages} {xiv+238}\BibitemShut {NoStop}%
\bibitem [{\citenamefont {Goodman}\ \emph {et~al.}(1989)\citenamefont
  {Goodman}, \citenamefont {de~la Harpe},\ and\ \citenamefont
  {Jones}}]{GdHJ89}%
  \BibitemOpen
  \bibfield  {author} {\bibinfo {author} {\bibfnamefont {F.~M.}\ \bibnamefont
  {Goodman}}, \bibinfo {author} {\bibfnamefont {P.}~\bibnamefont {de~la
  Harpe}}, \ and\ \bibinfo {author} {\bibfnamefont {V.~F.~R.}\ \bibnamefont
  {Jones}},\ }\href {\doibase 10.1007/978-1-4613-9641-3} {\emph {\bibinfo
  {title} {{Coxeter graphs and towers of algebras}}}},\ \bibinfo {series}
  {{Mathematical Sciences Research Institute Publications}}, Vol.~\bibinfo
  {volume} {14}\ (\bibinfo  {publisher} {Springer-Verlag},\ \bibinfo {address}
  {New York},\ \bibinfo {year} {1989})\ pp.\ \bibinfo {pages}
  {x+288}\BibitemShut {NoStop}%
\bibitem [{\citenamefont {Ohya}\ and\ \citenamefont {Petz}(1993)}]{OhPe93}%
  \BibitemOpen
  \bibfield  {author} {\bibinfo {author} {\bibfnamefont {M.}~\bibnamefont
  {Ohya}}\ and\ \bibinfo {author} {\bibfnamefont {D.}~\bibnamefont {Petz}},\
  }\href {\doibase 10.1007/978-3-642-57997-4} {\emph {\bibinfo {title} {Quantum
  entropy and its use}}},\ Texts and Monographs in Physics\ (\bibinfo
  {publisher} {Springer-Verlag, Berlin},\ \bibinfo {year} {1993})\ pp.\
  \bibinfo {pages} {viii+335}\BibitemShut {NoStop}%
\bibitem [{\citenamefont {Davies}\ and\ \citenamefont {Lewis}(1970)}]{DaLe70}%
  \BibitemOpen
  \bibfield  {author} {\bibinfo {author} {\bibfnamefont {E.~B.}\ \bibnamefont
  {Davies}}\ and\ \bibinfo {author} {\bibfnamefont {J.~T.}\ \bibnamefont
  {Lewis}},\ }\href {\doibase 10.1007/BF01647093} {\bibfield  {journal}
  {\bibinfo  {journal} {Comm. Math. Phys.}\ }\textbf {\bibinfo {volume} {17}},\
  \bibinfo {pages} {239} (\bibinfo {year} {1970})}\BibitemShut {NoStop}%
\bibitem [{\citenamefont {Wick}\ \emph {et~al.}(1952)\citenamefont {Wick},
  \citenamefont {Wightman},\ and\ \citenamefont {Wigner}}]{WWW52}%
  \BibitemOpen
  \bibfield  {author} {\bibinfo {author} {\bibfnamefont {G.~C.}\ \bibnamefont
  {Wick}}, \bibinfo {author} {\bibfnamefont {A.}~\bibnamefont {Wightman}}, \
  and\ \bibinfo {author} {\bibfnamefont {E.}~\bibnamefont {Wigner}},\ }\href
  {\doibase 10.1103/PhysRev.88.101} {\bibfield  {journal} {\bibinfo  {journal}
  {Phys. Rev.}\ }\textbf {\bibinfo {volume} {88}},\ \bibinfo {pages} {101}
  (\bibinfo {year} {1952})}\BibitemShut {NoStop}%
\bibitem [{\citenamefont {Streater}(1975)}]{St75}%
  \BibitemOpen
  \bibfield  {author} {\bibinfo {author} {\bibfnamefont {R.~F.}\ \bibnamefont
  {Streater}},\ }\href {\doibase 10.1088/0034-4885/38/7/001} {\bibfield
  {journal} {\bibinfo  {journal} {Rep. Prog. Phys.}\ }\textbf {\bibinfo
  {volume} {38}},\ \bibinfo {pages} {771} (\bibinfo {year} {1975})}\BibitemShut
  {NoStop}%
\bibitem [{Note1()}]{Note1}%
  \BibitemOpen
  \bibinfo {note} {The Jamio{\l }kowski state is commonly defined as the
  partial transpose of the \protect \emph {Choi matrix}~\cite {Ch75}.
  Equation~\protect \textup {\hbox {\mathsurround \z@ \protect \normalfont
  (\ignorespaces \ref {eq:Jcs}\unskip \@@italiccorr )}} provides a manifestly
  basis-independent formula for the Jamio{\l }kowski state that is also
  applicable to arbitrary multi-matrix algebras.}\BibitemShut {Stop}%
\bibitem [{Note2()}]{Note2}%
  \BibitemOpen
  \bibinfo {note} {The notation $\protect \mathcal {E}\star \rho $ is based on
  Refs.~\cite {LeSp13,HHPBS17}.}\BibitemShut {Stop}%
\bibitem [{Note3()}]{Note3}%
  \BibitemOpen
  \bibinfo {note} {The condition $[\protect \mathscr {D}[\protect \mathcal
  {E}],\rho \otimes 1_{{{\protect \mathcal {B}}}}]=0$ is equivalent to $[\rho
  ,\protect \mathcal {E}^*(B)]=0$ for all $B\in {{\protect \mathcal
  {B}}}$.}\BibitemShut {Stop}%
\bibitem [{Note4()}]{Note4}%
  \BibitemOpen
  \bibinfo {note} {For reference, and to provide a formula in the
  Schr{\"o}dinger picture, a state over time function $\star $ is \protect
  \emph {\protect \textbf {associative}} iff \protect \[ \begin {split}
  &\protect \mathscr {D}_{{{\protect \mathcal {A}}},{{\protect \mathcal
  {B}}}\otimes {{\protect \mathcal {C}}}}^{-1}\left [{\protect \rm tr}
  (1_{{{\protect \mathcal {A}}}})\left ((\protect \mathcal {F}\circ {\protect
  \rm tr} _{{{\protect \mathcal {A}}}})\star \left (\protect \frac {\protect
  \mathscr {D}_{{{\protect \mathcal {A}}},{{\protect \mathcal {B}}}}[\protect
  \mathcal {E}]}{{\protect \rm tr} (1_{{{\protect \mathcal {A}}}})}\right
  )\right )\right ]\star \rho \\ &=(\protect \mathcal {F}\circ {\protect \rm
  tr} _{{{\protect \mathcal {A}}}})\star (\protect \mathcal {E}\star \rho )
  \end {split} \protect \] for all $\rho \in \protect \mathcal {S}({{\protect
  \mathcal {A}}})$ and composable pairs ${{\protect \mathcal {A}}}\protect
  \xrightarrow {\protect \mathcal {E}}{{\protect \mathcal {B}}}\protect
  \xrightarrow {\protect \mathcal {F}}{{\protect \mathcal {C}}}$ of CPTP maps.
  Note that one must be careful about domains to even make sense of this axiom.
  In particular, we point out that the formula presented here is a more
  appropriate axiom for associativity than the one considered in Ref.~\cite
  {FuPa22} due to the fact that $\star $ need not be state-linear (this is why
  we have included the factors of ${\protect \rm tr} (1_{{{\protect \mathcal
  {A}}}})$ in the present formulation). Indeed, if $\star $ is state-linear,
  then these factors cancel, which need not happen if $\star $ is not
  state-linear. Note that this change does not alter the main theorem of
  Ref.~\cite {FuPa22}, since the state over time function constructed there is
  state-linear.}\BibitemShut {Stop}%
\bibitem [{\citenamefont {Asorey}\ \emph {et~al.}(2005)\citenamefont {Asorey},
  \citenamefont {Kossakowski}, \citenamefont {Marmo},\ and\ \citenamefont
  {Sudarshan}}]{AKMS06}%
  \BibitemOpen
  \bibfield  {author} {\bibinfo {author} {\bibfnamefont {M.}~\bibnamefont
  {Asorey}}, \bibinfo {author} {\bibfnamefont {A.}~\bibnamefont {Kossakowski}},
  \bibinfo {author} {\bibfnamefont {G.}~\bibnamefont {Marmo}}, \ and\ \bibinfo
  {author} {\bibfnamefont {E.~C.~G.}\ \bibnamefont {Sudarshan}},\ }\href
  {\doibase 10.1007/s11080-005-4482-3} {\bibfield  {journal} {\bibinfo
  {journal} {Open Systems \& Information Dynamics}\ }\textbf {\bibinfo {volume}
  {12}},\ \bibinfo {pages} {319} (\bibinfo {year} {2005})},\ \Eprint
  {http://arxiv.org/abs/quant-ph/0602228} {arXiv:quant-ph/0602228 [quant-ph]}
  \BibitemShut {NoStop}%
\bibitem [{\citenamefont {{Girard}}\ \emph {et~al.}(2021)\citenamefont
  {{Girard}}, \citenamefont {{Pl{\'a}vala}},\ and\ \citenamefont
  {{Sikora}}}]{GPS21}%
  \BibitemOpen
  \bibfield  {author} {\bibinfo {author} {\bibfnamefont {M.}~\bibnamefont
  {{Girard}}}, \bibinfo {author} {\bibfnamefont {M.}~\bibnamefont
  {{Pl{\'a}vala}}}, \ and\ \bibinfo {author} {\bibfnamefont {J.}~\bibnamefont
  {{Sikora}}},\ }\href {\doibase 10.1038/s41467-021-22275-0} {\bibfield
  {journal} {\bibinfo  {journal} {Nature Communications}\ }\textbf {\bibinfo
  {volume} {12}},\ \bibinfo {eid} {2129} (\bibinfo {year} {2021})},\ \Eprint
  {http://arxiv.org/abs/2009.03279} {arXiv:2009.03279 [quant-ph]} \BibitemShut
  {NoStop}%
\bibitem [{\citenamefont {Fawzi}\ and\ \citenamefont {Renner}(2015)}]{FaRe15}%
  \BibitemOpen
  \bibfield  {author} {\bibinfo {author} {\bibfnamefont {O.}~\bibnamefont
  {Fawzi}}\ and\ \bibinfo {author} {\bibfnamefont {R.}~\bibnamefont {Renner}},\
  }\href {\doibase 10.1007/s00220-015-2466-x} {\bibfield  {journal} {\bibinfo
  {journal} {Commun. Math. Phys.}\ }\textbf {\bibinfo {volume} {340}},\
  \bibinfo {pages} {575} (\bibinfo {year} {2015})}\BibitemShut {NoStop}%
\bibitem [{\citenamefont {Wilde}(2015)}]{Wilde15}%
  \BibitemOpen
  \bibfield  {author} {\bibinfo {author} {\bibfnamefont {M.~M.}\ \bibnamefont
  {Wilde}},\ }\href {\doibase 10.1098/rspa.2015.0338} {\bibfield  {journal}
  {\bibinfo  {journal} {Proc. R. Soc. A}\ }\textbf {\bibinfo {volume} {471}},\
  \bibinfo {pages} {20150338} (\bibinfo {year} {2015})}\BibitemShut {NoStop}%
\bibitem [{\citenamefont {Junge}\ \emph {et~al.}(2016)\citenamefont {Junge},
  \citenamefont {Renner}, \citenamefont {Sutter}, \citenamefont {Wilde},\ and\
  \citenamefont {Winter}}]{JRSWW16}%
  \BibitemOpen
  \bibfield  {author} {\bibinfo {author} {\bibfnamefont {M.}~\bibnamefont
  {Junge}}, \bibinfo {author} {\bibfnamefont {R.}~\bibnamefont {Renner}},
  \bibinfo {author} {\bibfnamefont {D.}~\bibnamefont {Sutter}}, \bibinfo
  {author} {\bibfnamefont {M.~M.}\ \bibnamefont {Wilde}}, \ and\ \bibinfo
  {author} {\bibfnamefont {A.}~\bibnamefont {Winter}},\ }in\ \href {\doibase
  10.1109/ISIT.2016.7541748} {\emph {\bibinfo {booktitle} {2016 IEEE
  International Symposium on Information Theory (ISIT)}}}\ (\bibinfo {year}
  {2016})\ pp.\ \bibinfo {pages} {2494--2498}\BibitemShut {NoStop}%
\bibitem [{\citenamefont {Carlen}\ and\ \citenamefont
  {Vershynina}(2020)}]{CaVe20}%
  \BibitemOpen
  \bibfield  {author} {\bibinfo {author} {\bibfnamefont {E.~A.}\ \bibnamefont
  {Carlen}}\ and\ \bibinfo {author} {\bibfnamefont {A.}~\bibnamefont
  {Vershynina}},\ }\href {\doibase 10.1088/1751-8121/ab5ab7} {\bibfield
  {journal} {\bibinfo  {journal} {J. Phys. A}\ }\textbf {\bibinfo {volume}
  {53}},\ \bibinfo {pages} {035204} (\bibinfo {year} {2020})},\ \Eprint
  {http://arxiv.org/abs/1710.02409} {arXiv:1710.02409 [math.OA]} \BibitemShut
  {NoStop}%
\bibitem [{\citenamefont {{Surace}}\ and\ \citenamefont
  {{Scandi}}(2022)}]{SuSc22}%
  \BibitemOpen
  \bibfield  {author} {\bibinfo {author} {\bibfnamefont {J.}~\bibnamefont
  {{Surace}}}\ and\ \bibinfo {author} {\bibfnamefont {M.}~\bibnamefont
  {{Scandi}}},\ }\href@noop {} {\enquote {\bibinfo {title} {{State retrieval
  beyond Bayes' retrodiction and reverse processes}},}\ } (\bibinfo {year}
  {2022}),\ \Eprint {http://arxiv.org/abs/2201.09899} {arXiv:2201.09899
  [quant-ph]} \BibitemShut {NoStop}%
\bibitem [{Note5()}]{Note5}%
  \BibitemOpen
  \bibinfo {note} {For general algebras, the swap map $\gamma :{{\protect
  \mathcal {A}}}\otimes {{\protect \mathcal {B}}}\to {{\protect \mathcal
  {B}}}\otimes {{\protect \mathcal {A}}}$ is not to be confused with the swap
  operator $\protect \mathscr {D}[\protect \mathrm {id}_{{{\protect \mathcal
  {A}}}}]$, the latter of which is an \protect \emph {element} of ${{\protect
  \mathcal {A}}}\otimes {{\protect \mathcal {A}}}$.}\BibitemShut {Stop}%
\bibitem [{\citenamefont {{Baez}}(2006)}]{Ba06}%
  \BibitemOpen
  \bibfield  {author} {\bibinfo {author} {\bibfnamefont {J.~C.}\ \bibnamefont
  {{Baez}}},\ }in\ \href {\doibase 10.1093/acprof:oso/9780199269693.003.0008}
  {\emph {\bibinfo {booktitle} {Structural Foundations of Quantum Gravity}}},\
  \bibinfo {editor} {edited by\ \bibinfo {editor} {\bibfnamefont
  {S.}~\bibnamefont {French}}, \bibinfo {editor} {\bibfnamefont
  {D.}~\bibnamefont {Rickles}}, \ and\ \bibinfo {editor} {\bibfnamefont
  {J.}~\bibnamefont {Saatsi}}}\ (\bibinfo  {publisher} {Oxford U. Press},\
  \bibinfo {year} {2006})\ pp.\ \bibinfo {pages} {240--265},\ \Eprint
  {http://arxiv.org/abs/quant-ph/0404040} {arXiv:quant-ph/0404040 [quant-ph]}
  \BibitemShut {NoStop}%
\bibitem [{\citenamefont {Coecke}\ and\ \citenamefont
  {Kissinger}(2017)}]{CoKi17}%
  \BibitemOpen
  \bibfield  {author} {\bibinfo {author} {\bibfnamefont {B.}~\bibnamefont
  {Coecke}}\ and\ \bibinfo {author} {\bibfnamefont {A.}~\bibnamefont
  {Kissinger}},\ }\href {\doibase 10.1017/9781316219317} {\emph {\bibinfo
  {title} {Picturing Quantum Processes: A First Course in Quantum Theory and
  Diagrammatic Reasoning}}}\ (\bibinfo  {publisher} {Cambridge University
  Press},\ \bibinfo {year} {2017})\BibitemShut {NoStop}%
\bibitem [{\citenamefont {Chiribella}\ and\ \citenamefont
  {Liu}(2022)}]{ChLi22}%
  \BibitemOpen
  \bibfield  {author} {\bibinfo {author} {\bibfnamefont {G.}~\bibnamefont
  {Chiribella}}\ and\ \bibinfo {author} {\bibfnamefont {Z.}~\bibnamefont
  {Liu}},\ }\href {\doibase 10.1038/s42005-022-00967-3} {\bibfield  {journal}
  {\bibinfo  {journal} {Commun Phys}\ }\textbf {\bibinfo {volume} {5}},\
  \bibinfo {pages} {1} (\bibinfo {year} {2022})},\ \Eprint
  {http://arxiv.org/abs/2012.03859} {arXiv:2012.03859 [quant-ph]} \BibitemShut
  {NoStop}%
\bibitem [{Note6()}]{Note6}%
  \BibitemOpen
  \bibinfo {note} {This is because \protect \[ \begin {split} &\protect
  \mathscr {D}_{{{\protect \mathcal {A}}},{{\protect \mathcal {B}}}\otimes
  {{\protect \mathcal {C}}}}^{-1}\left [{\protect \rm tr} (1_{{{\protect
  \mathcal {A}}}})\left ((\protect \mathcal {F}\circ {\protect \rm tr}
  _{{{\protect \mathcal {A}}}})\star \left (\protect \frac {\protect \mathscr
  {D}_{{{\protect \mathcal {A}}},{{\protect \mathcal {B}}}}[\protect \mathcal
  {E}]}{{\protect \rm tr} (1_{{{\protect \mathcal {A}}}})}\right )\right
  )\right ]\star \rho \\ &=\rho \otimes \protect \mathcal {E}(\rho )\otimes
  \protect \mathcal {F}\left (\protect \mathcal {E}\left (\protect \frac
  {1_{{{\protect \mathcal {A}}}}}{{\protect \rm tr} (1_{{{\protect \mathcal
  {A}}}})}\right )\right ), \end {split} \protect \] while \protect \[
  (\protect \mathcal {F}\circ {\protect \rm tr} _{{{\protect \mathcal
  {A}}}})\star (\protect \mathcal {E}\star \rho )=\rho \otimes \protect
  \mathcal {E}(\rho )\otimes \protect \mathcal {F}{\setbox \z@ \hbox
  {\frozen@everymath \@emptytoks \mathsurround \z@ $\nulldelimiterspace \z@
  \left (\vcenter to\@ne \big@size {}\right .$}\box \z@ }\protect \mathcal
  {E}(\rho ){\setbox \z@ \hbox {\frozen@everymath \@emptytoks \mathsurround \z@
  $\nulldelimiterspace \z@ \left )\vcenter to\@ne \big@size {}\right .$}\box
  \z@ }. \protect \] Since $\protect \mathcal {F}\left (\protect \mathcal
  {E}\left (\protect \frac {1_{{{\protect \mathcal {A}}}}}{{\protect \rm tr}
  (1_{{{\protect \mathcal {A}}}})}\right )\right )$ need not equal $\protect
  \mathcal {F}{\setbox \z@ \hbox {\frozen@everymath \@emptytoks \mathsurround
  \z@ $\nulldelimiterspace \z@ \left (\vcenter to\@ne \big@size {}\right
  .$}\box \z@ }\protect \mathcal {E}(\rho ){\setbox \z@ \hbox
  {\frozen@everymath \@emptytoks \mathsurround \z@ $\nulldelimiterspace \z@
  \left )\vcenter to\@ne \big@size {}\right .$}\box \z@ }$ for arbitrary
  $\protect \mathcal {F},\protect \mathcal {E},$ and $\rho $, this shows that
  associativity fails in general.}\BibitemShut {Stop}%
\bibitem [{Note7()}]{Note7}%
  \BibitemOpen
  \bibinfo {note} {Technically, Ohya's original definition of a compound state
  does not give a well-defined state over time function because it assumes
  \protect \emph {additional data} when the input state has repeating
  eigenvalues. Namely, it requires a \protect \emph {choice} of an eigenvector
  decomposition. Therefore, we have slightly modified Ohya's construction so
  that no such additional data are needed, and so that we obtain a well-defined
  state over time function.}\BibitemShut {Stop}%
\bibitem [{\citenamefont {Werner}(1989)}]{We89}%
  \BibitemOpen
  \bibfield  {author} {\bibinfo {author} {\bibfnamefont {R.~F.}\ \bibnamefont
  {Werner}},\ }\href {\doibase 10.1103/PhysRevA.40.4277} {\bibfield  {journal}
  {\bibinfo  {journal} {Phys. Rev. A}\ }\textbf {\bibinfo {volume} {40}},\
  \bibinfo {pages} {4277} (\bibinfo {year} {1989})}\BibitemShut {NoStop}%
\bibitem [{\citenamefont {Choi}(1975)}]{Ch75}%
  \BibitemOpen
  \bibfield  {author} {\bibinfo {author} {\bibfnamefont {M.~D.}\ \bibnamefont
  {Choi}},\ }\href {\doibase 10.1016/0024-3795(75)90075-0} {\bibfield
  {journal} {\bibinfo  {journal} {Linear Algebra Appl.}\ }\textbf {\bibinfo
  {volume} {10}},\ \bibinfo {pages} {285} (\bibinfo {year} {1975})}\BibitemShut
  {NoStop}%
\bibitem [{\citenamefont {Petz}(1988{\natexlab{a}})}]{Pe88}%
  \BibitemOpen
  \bibfield  {author} {\bibinfo {author} {\bibfnamefont {D.}~\bibnamefont
  {Petz}},\ }\href {\doibase https://doi.org/10.1093/qmath/39.1.97} {\bibfield
  {journal} {\bibinfo  {journal} {Q. J. Math.}\ }\textbf {\bibinfo {volume}
  {39}},\ \bibinfo {pages} {97} (\bibinfo {year}
  {1988}{\natexlab{a}})}\BibitemShut {NoStop}%
\bibitem [{\citenamefont {Petz}(1986)}]{Pe86}%
  \BibitemOpen
  \bibfield  {author} {\bibinfo {author} {\bibfnamefont {D.}~\bibnamefont
  {Petz}},\ }\href {\doibase 10.1007/BF01212345} {\bibfield  {journal}
  {\bibinfo  {journal} {Commun.Math. Phys.}\ }\textbf {\bibinfo {volume}
  {105}},\ \bibinfo {pages} {123} (\bibinfo {year} {1986})}\BibitemShut
  {NoStop}%
\bibitem [{\citenamefont {{Petz}}(2003)}]{Pe03}%
  \BibitemOpen
  \bibfield  {author} {\bibinfo {author} {\bibfnamefont {D.}~\bibnamefont
  {{Petz}}},\ }\href {\doibase 10.1142/S0129055X03001576} {\bibfield  {journal}
  {\bibinfo  {journal} {Rev. Math. Phys.}\ }\textbf {\bibinfo {volume} {15}},\
  \bibinfo {pages} {79} (\bibinfo {year} {2003})},\ \Eprint
  {http://arxiv.org/abs/quant-ph/0209053} {arXiv:quant-ph/0209053 [quant-ph]}
  \BibitemShut {NoStop}%
\bibitem [{\citenamefont {Jen{\v{c}}ov{\'{a}}}(2017)}]{Je17}%
  \BibitemOpen
  \bibfield  {author} {\bibinfo {author} {\bibfnamefont {A.}~\bibnamefont
  {Jen{\v{c}}ov{\'{a}}}},\ }\href {\doibase 10.1088/1751-8121/aa5661}
  {\bibfield  {journal} {\bibinfo  {journal} {J. Phys. A: Math. Theor.}\
  }\textbf {\bibinfo {volume} {50}},\ \bibinfo {pages} {085303} (\bibinfo
  {year} {2017})},\ \Eprint {http://arxiv.org/abs/1604.02831} {arXiv:1604.02831
  [quant-ph]} \BibitemShut {NoStop}%
\bibitem [{\citenamefont {Barnum}\ and\ \citenamefont {Knill}(2002)}]{BaKn02}%
  \BibitemOpen
  \bibfield  {author} {\bibinfo {author} {\bibfnamefont {H.}~\bibnamefont
  {Barnum}}\ and\ \bibinfo {author} {\bibfnamefont {E.}~\bibnamefont {Knill}},\
  }\href@noop {} {\bibfield  {journal} {\bibinfo  {journal} {J. Math. Phys.}\
  }\textbf {\bibinfo {volume} {43}},\ \bibinfo {pages} {2097} (\bibinfo {year}
  {2002})}\BibitemShut {NoStop}%
\bibitem [{\citenamefont {{Ng}}\ and\ \citenamefont
  {{Mandayam}}(2010)}]{NgMa10}%
  \BibitemOpen
  \bibfield  {author} {\bibinfo {author} {\bibfnamefont {H.~K.}\ \bibnamefont
  {{Ng}}}\ and\ \bibinfo {author} {\bibfnamefont {P.}~\bibnamefont
  {{Mandayam}}},\ }\href {\doibase 10.1103/PhysRevA.81.062342} {\bibfield
  {journal} {\bibinfo  {journal} {Phys. Rev. A}\ }\textbf {\bibinfo {volume}
  {81}},\ \bibinfo {eid} {062342} (\bibinfo {year} {2010})},\ \Eprint
  {http://arxiv.org/abs/0909.0931} {arXiv:0909.0931 [quant-ph]} \BibitemShut
  {NoStop}%
\bibitem [{\citenamefont {{Crooks}}(2008)}]{Cr08}%
  \BibitemOpen
  \bibfield  {author} {\bibinfo {author} {\bibfnamefont {G.~E.}\ \bibnamefont
  {{Crooks}}},\ }\href {\doibase 10.1103/PhysRevA.77.034101} {\bibfield
  {journal} {\bibinfo  {journal} {Phys. Rev. A.}\ }\textbf {\bibinfo {volume}
  {77}},\ \bibinfo {eid} {034101} (\bibinfo {year} {2008})},\ \Eprint
  {http://arxiv.org/abs/0706.3749} {arXiv:0706.3749 [quant-ph]} \BibitemShut
  {NoStop}%
\bibitem [{\citenamefont {Leifer}\ and\ \citenamefont {Pusey}(2017)}]{LePu17}%
  \BibitemOpen
  \bibfield  {author} {\bibinfo {author} {\bibfnamefont {M.~S.}\ \bibnamefont
  {Leifer}}\ and\ \bibinfo {author} {\bibfnamefont {M.~F.}\ \bibnamefont
  {Pusey}},\ }\href {\doibase 10.1098/rspa.2016.0607} {\bibfield  {journal}
  {\bibinfo  {journal} {Proc. R. Soc. A}\ }\textbf {\bibinfo {volume} {473}},\
  \bibinfo {pages} {20160607} (\bibinfo {year} {2017})},\ \Eprint
  {http://arxiv.org/abs/1607.07871} {arXiv:1607.07871 [quant-ph]} \BibitemShut
  {NoStop}%
\bibitem [{\citenamefont {{Buscemi}}\ \emph {et~al.}(2022)\citenamefont
  {{Buscemi}}, \citenamefont {{Schindler}},\ and\ \citenamefont
  {{{\v{S}}afr{\'a}nek}}}]{BSS22}%
  \BibitemOpen
  \bibfield  {author} {\bibinfo {author} {\bibfnamefont {F.}~\bibnamefont
  {{Buscemi}}}, \bibinfo {author} {\bibfnamefont {J.}~\bibnamefont
  {{Schindler}}}, \ and\ \bibinfo {author} {\bibfnamefont {D.}~\bibnamefont
  {{{\v{S}}afr{\'a}nek}}},\ }\href {\doibase 10.48550/ARXIV.2209.03803}
  {\enquote {\bibinfo {title} {Observational entropy, coarse quantum states,
  and {P}etz recovery: information-theoretic properties and bounds},}\ }
  (\bibinfo {year} {2022}),\ \Eprint {http://arxiv.org/abs/2209.03803}
  {arXiv:2209.03803 [quant-ph]} \BibitemShut {NoStop}%
\bibitem [{\citenamefont {Kwon}\ and\ \citenamefont {Kim}(2019)}]{KwKi19}%
  \BibitemOpen
  \bibfield  {author} {\bibinfo {author} {\bibfnamefont {H.}~\bibnamefont
  {Kwon}}\ and\ \bibinfo {author} {\bibfnamefont {M.~S.}\ \bibnamefont {Kim}},\
  }\href {\doibase 10.1103/PhysRevX.9.031029} {\bibfield  {journal} {\bibinfo
  {journal} {Phys. Rev. X}\ }\textbf {\bibinfo {volume} {9}},\ \bibinfo {pages}
  {031029} (\bibinfo {year} {2019})}\BibitemShut {NoStop}%
\bibitem [{\citenamefont {Huang}(2022)}]{Hu22}%
  \BibitemOpen
  \bibfield  {author} {\bibinfo {author} {\bibfnamefont {Z.}~\bibnamefont
  {Huang}},\ }\href {\doibase 10.1103/PhysRevA.105.062217} {\bibfield
  {journal} {\bibinfo  {journal} {Phys. Rev. A}\ }\textbf {\bibinfo {volume}
  {105}},\ \bibinfo {pages} {062217} (\bibinfo {year} {2022})},\ \Eprint
  {http://arxiv.org/abs/2201.08691} {arXiv:2201.08691 [quant-ph]} \BibitemShut
  {NoStop}%
\bibitem [{\citenamefont {Furuya}\ \emph {et~al.}(2022)\citenamefont {Furuya},
  \citenamefont {Lashkari},\ and\ \citenamefont {Ouseph}}]{FuLaOu22}%
  \BibitemOpen
  \bibfield  {author} {\bibinfo {author} {\bibfnamefont {K.}~\bibnamefont
  {Furuya}}, \bibinfo {author} {\bibfnamefont {N.}~\bibnamefont {Lashkari}}, \
  and\ \bibinfo {author} {\bibfnamefont {S.}~\bibnamefont {Ouseph}},\ }\href
  {\doibase 10.1007/JHEP01(2022)170} {\bibfield  {journal} {\bibinfo  {journal}
  {J. High Energ. Phys.}\ }\textbf {\bibinfo {volume} {2022}} (\bibinfo {year}
  {2022}),\ 10.1007/JHEP01(2022)170},\ \Eprint
  {http://arxiv.org/abs/2012.14001} {arXiv:2012.14001 [hep-th]} \BibitemShut
  {NoStop}%
\bibitem [{\citenamefont {{Chen}}\ \emph {et~al.}(2020)\citenamefont {{Chen}},
  \citenamefont {{Penington}},\ and\ \citenamefont {{Salton}}}]{CPS20}%
  \BibitemOpen
  \bibfield  {author} {\bibinfo {author} {\bibfnamefont {C.-F.}\ \bibnamefont
  {{Chen}}}, \bibinfo {author} {\bibfnamefont {G.}~\bibnamefont {{Penington}}},
  \ and\ \bibinfo {author} {\bibfnamefont {G.}~\bibnamefont {{Salton}}},\
  }\href {\doibase 10.1007/JHEP01(2020)168} {\bibfield  {journal} {\bibinfo
  {journal} {J. High Energy Phys.}\ }\textbf {\bibinfo {volume} {2020}},\
  \bibinfo {eid} {168} (\bibinfo {year} {2020})},\ \Eprint
  {http://arxiv.org/abs/1902.02844} {arXiv:1902.02844 [hep-th]} \BibitemShut
  {NoStop}%
\bibitem [{\citenamefont {Penington}\ \emph {et~al.}(2022)\citenamefont
  {Penington}, \citenamefont {Shenker}, \citenamefont {Stanford},\ and\
  \citenamefont {Yang}}]{PSSY22}%
  \BibitemOpen
  \bibfield  {author} {\bibinfo {author} {\bibfnamefont {G.}~\bibnamefont
  {Penington}}, \bibinfo {author} {\bibfnamefont {S.~H.}\ \bibnamefont
  {Shenker}}, \bibinfo {author} {\bibfnamefont {D.}~\bibnamefont {Stanford}}, \
  and\ \bibinfo {author} {\bibfnamefont {Z.}~\bibnamefont {Yang}},\ }\href
  {\doibase 10.1007/JHEP03(2022)205} {\bibfield  {journal} {\bibinfo  {journal}
  {J. High Energ. Phys.}\ }\textbf {\bibinfo {volume} {2022}},\ \bibinfo
  {pages} {1} (\bibinfo {year} {2022})},\ \Eprint
  {http://arxiv.org/abs/1911.11977} {arXiv:1911.11977 [hep-th]} \BibitemShut
  {NoStop}%
\bibitem [{\citenamefont {Kibe}\ \emph {et~al.}(2022)\citenamefont {Kibe},
  \citenamefont {Mandayam},\ and\ \citenamefont {Mukhopadhyay}}]{KMM22}%
  \BibitemOpen
  \bibfield  {author} {\bibinfo {author} {\bibfnamefont {T.}~\bibnamefont
  {Kibe}}, \bibinfo {author} {\bibfnamefont {P.}~\bibnamefont {Mandayam}}, \
  and\ \bibinfo {author} {\bibfnamefont {A.}~\bibnamefont {Mukhopadhyay}},\
  }\href {\doibase 10.1140/epjc/s10052-022-10382-1} {\bibfield  {journal}
  {\bibinfo  {journal} {Eur. Phys. J. C}\ }\textbf {\bibinfo {volume} {82}},\
  \bibinfo {eid} {463} (\bibinfo {year} {2022}),\
  10.1140/epjc/s10052-022-10382-1},\ \Eprint {http://arxiv.org/abs/2110.14669}
  {arXiv:2110.14669 [hep-th]} \BibitemShut {NoStop}%
\bibitem [{\citenamefont {Junge}\ \emph {et~al.}(2018)\citenamefont {Junge},
  \citenamefont {Renner}, \citenamefont {Sutter}, \citenamefont {Wilde},\ and\
  \citenamefont {Winter}}]{JRSWW18}%
  \BibitemOpen
  \bibfield  {author} {\bibinfo {author} {\bibfnamefont {M.}~\bibnamefont
  {Junge}}, \bibinfo {author} {\bibfnamefont {R.}~\bibnamefont {Renner}},
  \bibinfo {author} {\bibfnamefont {D.}~\bibnamefont {Sutter}}, \bibinfo
  {author} {\bibfnamefont {M.~M.}\ \bibnamefont {Wilde}}, \ and\ \bibinfo
  {author} {\bibfnamefont {A.}~\bibnamefont {Winter}},\ }\href {\doibase
  10.1007/s00023-018-0716-0} {\bibfield  {journal} {\bibinfo  {journal} {Ann.
  Henri Poincar\'e}\ }\textbf {\bibinfo {volume} {19}},\ \bibinfo {pages}
  {2955} (\bibinfo {year} {2018})},\ \Eprint {http://arxiv.org/abs/1509.07127}
  {arXiv:1509.07127 [quant-ph]} \BibitemShut {NoStop}%
\bibitem [{\citenamefont {Sutter}\ \emph {et~al.}(2016)\citenamefont {Sutter},
  \citenamefont {Tomamichel},\ and\ \citenamefont {Harrow}}]{SuToHa16}%
  \BibitemOpen
  \bibfield  {author} {\bibinfo {author} {\bibfnamefont {D.}~\bibnamefont
  {Sutter}}, \bibinfo {author} {\bibfnamefont {M.}~\bibnamefont {Tomamichel}},
  \ and\ \bibinfo {author} {\bibfnamefont {A.~W.}\ \bibnamefont {Harrow}},\
  }\href {\doibase 10.1109/tit.2016.2545680} {\bibfield  {journal} {\bibinfo
  {journal} {{IEEE} Transactions on Information Theory}\ }\textbf {\bibinfo
  {volume} {62}},\ \bibinfo {pages} {2907} (\bibinfo {year}
  {2016})}\BibitemShut {NoStop}%
\bibitem [{\citenamefont {{Sutter}}\ \emph {et~al.}(2017)\citenamefont
  {{Sutter}}, \citenamefont {{Berta}},\ and\ \citenamefont
  {{Tomamichel}}}]{SBT17}%
  \BibitemOpen
  \bibfield  {author} {\bibinfo {author} {\bibfnamefont {D.}~\bibnamefont
  {{Sutter}}}, \bibinfo {author} {\bibfnamefont {M.}~\bibnamefont {{Berta}}}, \
  and\ \bibinfo {author} {\bibfnamefont {M.}~\bibnamefont {{Tomamichel}}},\
  }\href {\doibase 10.1007/s00220-016-2778-5} {\bibfield  {journal} {\bibinfo
  {journal} {Comm. Math. Phys.}\ }\textbf {\bibinfo {volume} {352}},\ \bibinfo
  {pages} {37} (\bibinfo {year} {2017})},\ \Eprint
  {http://arxiv.org/abs/1604.03023} {arXiv:1604.03023 [math-ph]} \BibitemShut
  {NoStop}%
\bibitem [{Note8()}]{Note8}%
  \BibitemOpen
  \bibinfo {note} {This is not to be confused with the \protect \emph
  {operational} time-reverse in the sense defined by M.\ Leifer and M.\
  Pusey~\cite {LePu17} because our `control' and `uncontrolled' variables have
  switched}\BibitemShut {NoStop}%
\bibitem [{\citenamefont {Jeffrey}(1990)}]{Je90}%
  \BibitemOpen
  \bibfield  {author} {\bibinfo {author} {\bibfnamefont {R.~C.}\ \bibnamefont
  {Jeffrey}},\ }\href@noop {} {\emph {\bibinfo {title} {The logic of
  decision}}},\ \bibinfo {edition} {2nd}\ ed.\ (\bibinfo  {publisher}
  {University of Chicago Press},\ \bibinfo {year} {1990})\BibitemShut {NoStop}%
\bibitem [{\citenamefont {{Barandes}}\ and\ \citenamefont
  {{Kagan}}(2022)}]{BaKa21}%
  \BibitemOpen
  \bibfield  {author} {\bibinfo {author} {\bibfnamefont {J.~A.}\ \bibnamefont
  {{Barandes}}}\ and\ \bibinfo {author} {\bibfnamefont {D.}~\bibnamefont
  {{Kagan}}},\ }\href {\doibase 10.1016/j.aop.2022.169192} {\bibfield
  {journal} {\bibinfo  {journal} {Annals of Physics}\ ,\ \bibinfo {pages}
  {169192}} (\bibinfo {year} {2022})},\ \Eprint
  {http://arxiv.org/abs/2109.07447} {arXiv:2109.07447 [quant-ph]} \BibitemShut
  {NoStop}%
\bibitem [{\citenamefont {Furber}\ and\ \citenamefont {Jacobs}(2015)}]{FuJa13}%
  \BibitemOpen
  \bibfield  {author} {\bibinfo {author} {\bibfnamefont {R.}~\bibnamefont
  {Furber}}\ and\ \bibinfo {author} {\bibfnamefont {B.}~\bibnamefont
  {Jacobs}},\ }\href@noop {} {\bibfield  {journal} {\bibinfo  {journal} {Log.
  Methods Comput. Sci.}\ }\textbf {\bibinfo {volume} {11}},\ \bibinfo {pages}
  {1} (\bibinfo {year} {2015})},\ \Eprint {http://arxiv.org/abs/1303.1115}
  {arXiv:1303.1115 [math.CT]} \BibitemShut {NoStop}%
\bibitem [{\citenamefont {Parzygnat}(2017)}]{Pa17}%
  \BibitemOpen
  \bibfield  {author} {\bibinfo {author} {\bibfnamefont {A.~J.}\ \bibnamefont
  {Parzygnat}},\ }\href {\doibase 10.48550/arXiv.1708.00091} {\enquote
  {\bibinfo {title} {Discrete probabilistic and algebraic dynamics: a
  stochastic commutative {G}elfand-{N}aimark theorem},}\ } (\bibinfo {year}
  {2017}),\ \Eprint {http://arxiv.org/abs/1708.00091} {arXiv:1708.00091
  [math.FA]} \BibitemShut {NoStop}%
\bibitem [{\citenamefont {{Fritz}}\ \emph {et~al.}(2020)\citenamefont
  {{Fritz}}, \citenamefont {{Gonda}}, \citenamefont {{Perrone}},\ and\
  \citenamefont {{Fjeldgren Rischel}}}]{FGPR20}%
  \BibitemOpen
  \bibfield  {author} {\bibinfo {author} {\bibfnamefont {T.}~\bibnamefont
  {{Fritz}}}, \bibinfo {author} {\bibfnamefont {T.}~\bibnamefont {{Gonda}}},
  \bibinfo {author} {\bibfnamefont {P.}~\bibnamefont {{Perrone}}}, \ and\
  \bibinfo {author} {\bibfnamefont {E.}~\bibnamefont {{Fjeldgren Rischel}}},\
  }\href {\doibase 10.48550/ARXIV.2010.07416} {\enquote {\bibinfo {title}
  {{Representable Markov Categories and Comparison of Statistical Experiments
  in Categorical Probability}},}\ } (\bibinfo {year} {2020}),\ \Eprint
  {http://arxiv.org/abs/2010.07416} {arXiv:2010.07416 [math.ST]} \BibitemShut
  {NoStop}%
\bibitem [{\citenamefont {Jones}(1991)}]{Jo90}%
  \BibitemOpen
  \bibfield  {author} {\bibinfo {author} {\bibfnamefont {K.~R.~W.}\
  \bibnamefont {Jones}},\ }\href {\doibase 10.1016/0003-4916(91)90182-8}
  {\bibfield  {journal} {\bibinfo  {journal} {Ann. Phys.}\ }\textbf {\bibinfo
  {volume} {207}},\ \bibinfo {pages} {140} (\bibinfo {year}
  {1991})}\BibitemShut {NoStop}%
\bibitem [{\citenamefont {Jones}(2012)}]{Jo12}%
  \BibitemOpen
  \bibfield  {author} {\bibinfo {author} {\bibfnamefont {K.~R.~W.}\
  \bibnamefont {Jones}},\ }\href@noop {} {\emph {\bibinfo {title} {Quantum
  inference and the optimal determination of quantum states}}}\ (\bibinfo
  {publisher} {Createspace Independent Publishing Platform},\ \bibinfo {year}
  {2012})\BibitemShut {NoStop}%
\bibitem [{\citenamefont {Hellwig}\ and\ \citenamefont {Kraus}(1969)}]{HeKr69}%
  \BibitemOpen
  \bibfield  {author} {\bibinfo {author} {\bibfnamefont {K.-E.}\ \bibnamefont
  {Hellwig}}\ and\ \bibinfo {author} {\bibfnamefont {K.}~\bibnamefont
  {Kraus}},\ }\href {\doibase 10.1007/BF01645807} {\bibfield  {journal}
  {\bibinfo  {journal} {Comm. Math. Phys.}\ }\textbf {\bibinfo {volume} {11}},\
  \bibinfo {pages} {214} (\bibinfo {year} {1969})}\BibitemShut {NoStop}%
\bibitem [{\citenamefont {Kraus}(1971)}]{Kr71}%
  \BibitemOpen
  \bibfield  {author} {\bibinfo {author} {\bibfnamefont {K.}~\bibnamefont
  {Kraus}},\ }\href {\doibase 10.1016/0003-4916(71)90108-4} {\bibfield
  {journal} {\bibinfo  {journal} {Ann. Physics}\ }\textbf {\bibinfo {volume}
  {64}},\ \bibinfo {pages} {311} (\bibinfo {year} {1971})}\BibitemShut
  {NoStop}%
\bibitem [{\citenamefont {Ozawa}(1985)}]{Oz85}%
  \BibitemOpen
  \bibfield  {author} {\bibinfo {author} {\bibfnamefont {M.}~\bibnamefont
  {Ozawa}},\ }\href {\doibase 10.2977/prims/1195179625} {\bibfield  {journal}
  {\bibinfo  {journal} {Publ. Res. Inst. Math. Sci.}\ }\textbf {\bibinfo
  {volume} {21}},\ \bibinfo {pages} {279} (\bibinfo {year} {1985})}\BibitemShut
  {NoStop}%
\bibitem [{\citenamefont {L{\"u}ders}(1950)}]{Lu51}%
  \BibitemOpen
  \bibfield  {author} {\bibinfo {author} {\bibfnamefont {G.}~\bibnamefont
  {L{\"u}ders}},\ }\href {\doibase 10.1002/andp.19504430510} {\bibfield
  {journal} {\bibinfo  {journal} {Annalen der Physik}\ }\textbf {\bibinfo
  {volume} {8}},\ \bibinfo {pages} {322} (\bibinfo {year} {1950})}\BibitemShut
  {NoStop}%
\bibitem [{\citenamefont {Vanslette}(2018)}]{Va18}%
  \BibitemOpen
  \bibfield  {author} {\bibinfo {author} {\bibfnamefont {K.}~\bibnamefont
  {Vanslette}},\ }\href {\doibase 10.1088/2399-6528/aaaa08} {\bibfield
  {journal} {\bibinfo  {journal} {Commun. Phys.}\ }\textbf {\bibinfo {volume}
  {2}},\ \bibinfo {pages} {025017} (\bibinfo {year} {2018})},\ \Eprint
  {http://arxiv.org/abs/1710.10949} {arXiv:1710.10949 [quant-ph]} \BibitemShut
  {NoStop}%
\bibitem [{\citenamefont {Freeman}\ \emph {et~al.}(2022)\citenamefont
  {Freeman}, \citenamefont {Giannakis},\ and\ \citenamefont
  {Slawinska}}]{FGS22}%
  \BibitemOpen
  \bibfield  {author} {\bibinfo {author} {\bibfnamefont {D.}~\bibnamefont
  {Freeman}}, \bibinfo {author} {\bibfnamefont {D.}~\bibnamefont {Giannakis}},
  \ and\ \bibinfo {author} {\bibfnamefont {J.}~\bibnamefont {Slawinska}},\
  }\href {\doibase 10.48550/ARXIV.2208.03390} {\enquote {\bibinfo {title}
  {Quantum mechanics for closure of dynamical systems},}\ } (\bibinfo {year}
  {2022}),\ \Eprint {http://arxiv.org/abs/2208.03390} {arXiv:2208.03390
  [math.DS]} \BibitemShut {NoStop}%
\bibitem [{\citenamefont {Umegaki}(1962)}]{Um62}%
  \BibitemOpen
  \bibfield  {author} {\bibinfo {author} {\bibfnamefont {H.}~\bibnamefont
  {Umegaki}},\ }\href {\doibase 10.2996/kmj/1138844604} {\bibfield  {journal}
  {\bibinfo  {journal} {Kodai Math. Sem. Rep.}\ }\textbf {\bibinfo {volume}
  {14}},\ \bibinfo {pages} {59} (\bibinfo {year} {1962})}\BibitemShut {NoStop}%
\bibitem [{\citenamefont {Takesaki}(1972)}]{Tak72}%
  \BibitemOpen
  \bibfield  {author} {\bibinfo {author} {\bibfnamefont {M.}~\bibnamefont
  {Takesaki}},\ }\href {\doibase 10.1016/0022-1236(72)90004-3} {\bibfield
  {journal} {\bibinfo  {journal} {J. Funct. Anal.}\ }\textbf {\bibinfo {volume}
  {9}},\ \bibinfo {pages} {306} (\bibinfo {year} {1972})}\BibitemShut {NoStop}%
\bibitem [{\citenamefont {Alicki}(1995)}]{Al95}%
  \BibitemOpen
  \bibfield  {author} {\bibinfo {author} {\bibfnamefont {R.}~\bibnamefont
  {Alicki}},\ }\href {\doibase 10.1103/PhysRevLett.75.3020} {\bibfield
  {journal} {\bibinfo  {journal} {Phys. Rev. Lett.}\ }\textbf {\bibinfo
  {volume} {75}},\ \bibinfo {pages} {3020} (\bibinfo {year}
  {1995})}\BibitemShut {NoStop}%
\bibitem [{\citenamefont {Buscemi}\ \emph {et~al.}(2014)\citenamefont
  {Buscemi}, \citenamefont {Dall'Arno}, \citenamefont {Ozawa},\ and\
  \citenamefont {Vedral}}]{BDOV14}%
  \BibitemOpen
  \bibfield  {author} {\bibinfo {author} {\bibfnamefont {F.}~\bibnamefont
  {Buscemi}}, \bibinfo {author} {\bibfnamefont {M.}~\bibnamefont {Dall'Arno}},
  \bibinfo {author} {\bibfnamefont {M.}~\bibnamefont {Ozawa}}, \ and\ \bibinfo
  {author} {\bibfnamefont {V.}~\bibnamefont {Vedral}},\ }\href {\doibase
  10.1142/S0219749915600023} {\bibfield  {journal} {\bibinfo  {journal} {Int.
  J. Quantum Inf.}\ }\textbf {\bibinfo {volume} {12}},\ \bibinfo {pages}
  {1560002} (\bibinfo {year} {2014})}\BibitemShut {NoStop}%
\bibitem [{Note9()}]{Note9}%
  \BibitemOpen
  \bibinfo {note} {One could equivalently define a new state over time function
  $\star ^{\protect \dag }:\protect \mathrm {TP}({{\protect \mathcal
  {A}}},{{\protect \mathcal {B}}})\times \protect \mathcal {S}({{\protect
  \mathcal {A}}})\to {{\protect \mathcal {A}}}\otimes {{\protect \mathcal
  {B}}}$ whose value on $(\protect \mathcal {E},\rho )$ is given by $\protect
  \mathcal {E}\star ^{\protect \dag }\rho :={\setbox \z@ \hbox
  {\frozen@everymath \@emptytoks \mathsurround \z@ $\nulldelimiterspace \z@
  \left (\vcenter to\@ne \big@size {}\right .$}\box \z@ }(\protect \dag \circ
  \protect \mathcal {E}\circ \protect \dag )\star \rho {\setbox \z@ \hbox
  {\frozen@everymath \@emptytoks \mathsurround \z@ $\nulldelimiterspace \z@
  \left )\vcenter to\@ne \big@size {}\right .$}\box \z@ }^{\protect \dag }$. In
  this case, we could call $\star ^{\protect \dag }$ the \protect \emph
  {\protect \textbf {reverse orientation}} state over time function associated
  with $\star $. Then, an equivalent formulation of Bayes' rule reads $\protect
  \mathcal {E}\star \rho =\gamma {\setbox \z@ \hbox {\frozen@everymath
  \@emptytoks \mathsurround \z@ $\nulldelimiterspace \z@ \left (\vcenter to\@ne
  \big@size {}\right .$}\box \z@ }\protect \mathcal {E}^{\star }_{\rho }\star
  ^{\protect \dag }\protect \mathcal {E}(\rho ){\setbox \z@ \hbox
  {\frozen@everymath \@emptytoks \mathsurround \z@ $\nulldelimiterspace \z@
  \left )\vcenter to\@ne \big@size {}\right .$}\box \z@ }$, where only the swap
  map $\gamma $ is used, but with the additional modification of using the
  reverse orientation state over time on the right.}\BibitemShut {Stop}%
\bibitem [{Note10()}]{Note10}%
  \BibitemOpen
  \bibinfo {note} {Using notation from Ref.~\cite {FuPa22}, this is because the
  functional associated with $\protect \mathcal {E}\star \rho $ sends $A\otimes
  B\in {{\protect \mathcal {A}}}\otimes {{\protect \mathcal {B}}}$ to \[
  \<\mathcal {E}\star \rho ,A\otimes B\>=\tr \big ((\mathcal {E}\star \rho
  )^{\dag }(A\otimes B)\big )=\tr (\rho A\mathcal {E}^*(B)\big ), \] which
  agrees with $\omega {\setbox \z@ \hbox {\frozen@everymath \@emptytoks
  \mathsurround \z@ $\nulldelimiterspace \z@ \left (\vcenter to\@ne \big@size
  {}\right .$}\box \z@ }A F(B){\setbox \z@ \hbox {\frozen@everymath \@emptytoks
  \mathsurround \z@ $\nulldelimiterspace \z@ \left )\vcenter to\@ne \big@size
  {}\right .$}\box \z@ }$ upon setting $F=\protect \mathcal {E}^{*}$ and $\rho
  =\protect \mathscr {D}[\omega ]$.}\BibitemShut {Stop}%
\bibitem [{\citenamefont {Parzygnat}(2021)}]{PaQPL21}%
  \BibitemOpen
  \bibfield  {author} {\bibinfo {author} {\bibfnamefont {A.~J.}\ \bibnamefont
  {Parzygnat}},\ }in\ \href@noop {} {\emph {\bibinfo {booktitle} {Proceedings
  18th {I}nternational {C}onference on {Q}uantum {P}hysics and {L}ogic}}},\
  \bibinfo {series and number} {Electronic Proceedings in Theoretical Computer
  Science (EPTCS)}\ (\bibinfo {year} {2021})\ \Eprint
  {http://arxiv.org/abs/2102.01529} {arXiv:2102.01529 [quant-ph]} \BibitemShut
  {NoStop}%
\bibitem [{\citenamefont {Accardi}\ and\ \citenamefont
  {Cecchini}(1982)}]{AcCe82}%
  \BibitemOpen
  \bibfield  {author} {\bibinfo {author} {\bibfnamefont {L.}~\bibnamefont
  {Accardi}}\ and\ \bibinfo {author} {\bibfnamefont {C.}~\bibnamefont
  {Cecchini}},\ }\href {\doibase 10.1016/0022-1236(82)90022-2} {\bibfield
  {journal} {\bibinfo  {journal} {J. Funct. Anal.}\ }\textbf {\bibinfo {volume}
  {45}},\ \bibinfo {pages} {245} (\bibinfo {year} {1982})}\BibitemShut
  {NoStop}%
\bibitem [{\citenamefont {Bannon}\ \emph {et~al.}(2016)\citenamefont {Bannon},
  \citenamefont {Cameron},\ and\ \citenamefont {Mukherjee}}]{BCM16}%
  \BibitemOpen
  \bibfield  {author} {\bibinfo {author} {\bibfnamefont {J.~P.}\ \bibnamefont
  {Bannon}}, \bibinfo {author} {\bibfnamefont {J.}~\bibnamefont {Cameron}}, \
  and\ \bibinfo {author} {\bibfnamefont {K.}~\bibnamefont {Mukherjee}},\ }\href
  {\doibase https://doi.org/10.1016/j.jmaa.2016.03.013} {\bibfield  {journal}
  {\bibinfo  {journal} {J. Math. Anal. Appl}\ }\textbf {\bibinfo {volume}
  {439}},\ \bibinfo {pages} {701} (\bibinfo {year} {2016})},\ \Eprint
  {http://arxiv.org/abs/1905.06729} {arXiv:1905.06729 [math.OA]} \BibitemShut
  {NoStop}%
\bibitem [{\citenamefont {Anantharaman-Delaroche}(2006)}]{An06}%
  \BibitemOpen
  \bibfield  {author} {\bibinfo {author} {\bibfnamefont {C.}~\bibnamefont
  {Anantharaman-Delaroche}},\ }\href {\doibase 10.1007/s00440-005-0456-1}
  {\bibfield  {journal} {\bibinfo  {journal} {Probab. Theory Relat. Fields}\
  }\textbf {\bibinfo {volume} {135}},\ \bibinfo {pages} {520} (\bibinfo {year}
  {2006})},\ \Eprint {http://arxiv.org/abs/math/0412253} {arXiv:math/0412253
  [math.OA]} \BibitemShut {NoStop}%
\bibitem [{\citenamefont {Marvian}(2012)}]{Ma12}%
  \BibitemOpen
  \bibfield  {author} {\bibinfo {author} {\bibfnamefont {I.}~\bibnamefont
  {Marvian}},\ }\emph {\bibinfo {title} {Symmetry, Asymmetry and Quantum
  Information}},\ \href {http://hdl.handle.net/10012/7088} {Ph.D. thesis},\
  \bibinfo  {school} {University of Waterloo}, \bibinfo {address} {UWSpace}
  (\bibinfo {year} {2012})\BibitemShut {NoStop}%
\bibitem [{\citenamefont {Lostaglio}\ \emph {et~al.}(2015)\citenamefont
  {Lostaglio}, \citenamefont {Korzekwa}, \citenamefont {Jennings},\ and\
  \citenamefont {Rudolph}}]{LKJR15}%
  \BibitemOpen
  \bibfield  {author} {\bibinfo {author} {\bibfnamefont {M.}~\bibnamefont
  {Lostaglio}}, \bibinfo {author} {\bibfnamefont {K.}~\bibnamefont {Korzekwa}},
  \bibinfo {author} {\bibfnamefont {D.}~\bibnamefont {Jennings}}, \ and\
  \bibinfo {author} {\bibfnamefont {T.}~\bibnamefont {Rudolph}},\ }\href
  {\doibase 10.1103/PhysRevX.5.021001} {\bibfield  {journal} {\bibinfo
  {journal} {Phys. Rev. X}\ }\textbf {\bibinfo {volume} {5}},\ \bibinfo {pages}
  {021001} (\bibinfo {year} {2015})},\ \Eprint {http://arxiv.org/abs/1410.4572}
  {arXiv:1410.4572 [quant-ph]} \BibitemShut {NoStop}%
\bibitem [{\citenamefont {Connes}\ and\ \citenamefont
  {Rovelli}(1994)}]{CoRo94}%
  \BibitemOpen
  \bibfield  {author} {\bibinfo {author} {\bibfnamefont {A.}~\bibnamefont
  {Connes}}\ and\ \bibinfo {author} {\bibfnamefont {C.}~\bibnamefont
  {Rovelli}},\ }\href {\doibase 10.1088/0264-9381/11/12/007} {\bibfield
  {journal} {\bibinfo  {journal} {Class. Quant. Grav.}\ }\textbf {\bibinfo
  {volume} {11}},\ \bibinfo {pages} {2899} (\bibinfo {year} {1994})},\ \Eprint
  {http://arxiv.org/abs/gr-qc/9406019} {arXiv:gr-qc/9406019 [gr-qc]}
  \BibitemShut {NoStop}%
\bibitem [{\citenamefont {Nakata}\ \emph {et~al.}(2021)\citenamefont {Nakata},
  \citenamefont {Takayanagi}, \citenamefont {Taki}, \citenamefont {Tamaoka},\
  and\ \citenamefont {Wei}}]{NTTTW21}%
  \BibitemOpen
  \bibfield  {author} {\bibinfo {author} {\bibfnamefont {Y.}~\bibnamefont
  {Nakata}}, \bibinfo {author} {\bibfnamefont {T.}~\bibnamefont {Takayanagi}},
  \bibinfo {author} {\bibfnamefont {Y.}~\bibnamefont {Taki}}, \bibinfo {author}
  {\bibfnamefont {K.}~\bibnamefont {Tamaoka}}, \ and\ \bibinfo {author}
  {\bibfnamefont {Z.}~\bibnamefont {Wei}},\ }\href {\doibase
  10.1103/PhysRevD.103.026005} {\bibfield  {journal} {\bibinfo  {journal}
  {Phys. Rev. D}\ }\textbf {\bibinfo {volume} {103}},\ \bibinfo {pages}
  {026005} (\bibinfo {year} {2021})},\ \Eprint
  {http://arxiv.org/abs/2005.13801} {arXiv:2005.13801 [hep-th]} \BibitemShut
  {NoStop}%
\bibitem [{\citenamefont {Aharonov}\ and\ \citenamefont
  {Vaidman}(2008)}]{AhVa08}%
  \BibitemOpen
  \bibfield  {author} {\bibinfo {author} {\bibfnamefont {Y.}~\bibnamefont
  {Aharonov}}\ and\ \bibinfo {author} {\bibfnamefont {L.}~\bibnamefont
  {Vaidman}},\ }\enquote {\bibinfo {title} {The two-state vector formalism: An
  updated review},}\ in\ \href {\doibase 10.1007/978-3-540-73473-4_13} {\emph
  {\bibinfo {booktitle} {Time in Quantum Mechanics}}},\ \bibinfo {editor}
  {edited by\ \bibinfo {editor} {\bibfnamefont {J.}~\bibnamefont {Muga}},
  \bibinfo {editor} {\bibfnamefont {R.~S.}\ \bibnamefont {Mayato}}, \ and\
  \bibinfo {editor} {\bibfnamefont {{\'I}.}~\bibnamefont {Egusquiza}}}\
  (\bibinfo  {publisher} {Springer Berlin Heidelberg},\ \bibinfo {address}
  {Berlin, Heidelberg},\ \bibinfo {year} {2008})\ pp.\ \bibinfo {pages}
  {399--447}\BibitemShut {NoStop}%
\bibitem [{\citenamefont {Aharonov}\ \emph {et~al.}(1988)\citenamefont
  {Aharonov}, \citenamefont {Albert},\ and\ \citenamefont {Vaidman}}]{AAV88}%
  \BibitemOpen
  \bibfield  {author} {\bibinfo {author} {\bibfnamefont {Y.}~\bibnamefont
  {Aharonov}}, \bibinfo {author} {\bibfnamefont {D.~Z.}\ \bibnamefont
  {Albert}}, \ and\ \bibinfo {author} {\bibfnamefont {L.}~\bibnamefont
  {Vaidman}},\ }\href {\doibase 10.1103/PhysRevLett.60.1351} {\bibfield
  {journal} {\bibinfo  {journal} {Phys. Rev. Lett.}\ }\textbf {\bibinfo
  {volume} {60}},\ \bibinfo {pages} {1351} (\bibinfo {year}
  {1988})}\BibitemShut {NoStop}%
\bibitem [{Note11()}]{Note11}%
  \BibitemOpen
  \bibinfo {note} {When going from the Schr{\"o}dinger picture to the
  Heisenberg picture, the two-state $\rho _{x}$ gets sent to its weak value
  expectation functional sending an observable $A$ to ${\protect \rm tr} (\rho
  _{x}^{\protect \dag }A)$. This duality, which is obtained by using the
  Hilbert--Schmidt dual, is a conjugate-linear isomorphism with our convention.
  This is why $\rho _{x}^{\protect \dag }$, as opposed to $\rho _{x}$, appears
  in the expression for weak values. Technically, ${\protect \rm tr} (\rho
  _{x}A)$ reproduces the expression~\cite [Equation~(6)]{AAV88} exactly. This
  is due to a matter of convention, since, if we use the right bloom (to be
  defined in the next section), the two-state would instead be given by $\rho
  _{x}=\protect \frac {M\rho _{x}}{{\protect \rm tr} (M\rho _{x})}=\protect
  \frac {|\phi _{x}\delimiter "526930B \delimiter "426830A \psi |}{\delimiter
  "426830A \psi |\phi _{x}\delimiter "526930B }$. The associated weak values
  for this two-state would then equal ${\protect \rm tr} (\rho _{x}^{\protect
  \dag }A)$, in complete agreement with \cite
  [Equation~(6)]{AAV88}.}\BibitemShut {Stop}%
\bibitem [{\citenamefont {Fullwood}\ and\ \citenamefont
  {Parzygnat}(2021)}]{FuPa21}%
  \BibitemOpen
  \bibfield  {author} {\bibinfo {author} {\bibfnamefont {J.}~\bibnamefont
  {Fullwood}}\ and\ \bibinfo {author} {\bibfnamefont {A.~J.}\ \bibnamefont
  {Parzygnat}},\ }\href {\doibase 10.3390/e23081021} {\bibfield  {journal}
  {\bibinfo  {journal} {Entropy}\ }\textbf {\bibinfo {volume} {23}} (\bibinfo
  {year} {2021}),\ 10.3390/e23081021},\ \Eprint
  {http://arxiv.org/abs/2107.01975} {arXiv:2107.01975 [cs.IT]} \BibitemShut
  {NoStop}%
\bibitem [{Note12()}]{Note12}%
  \BibitemOpen
  \bibinfo {note} {Indeed, when ${{\protect \mathcal {A}}}=\protect \mathbb
  {M}_{m}$, the channel state $\protect \mathscr {D}[\protect \mathrm
  {id}_{{{\protect \mathcal {A}}}}]$ equals the swap operator sending
  $|i\>\otimes |j\>$ to $|j\>\otimes |i\>$ in Ref.~\cite {BDOV13}.}\BibitemShut
  {Stop}%
\bibitem [{Note13()}]{Note13}%
  \BibitemOpen
  \bibinfo {note} {In terms of reverse orientation state over time functions,
  this reads $\star _{L}^{\protect \dag }=\star _{R}$, and a proof is given by
  \begingroup \protect \allowdisplaybreaks \protect \[ \begin {split} \protect
  \mathcal {E}\star _{\protect \mathrm {L}}^{\protect \dag }\rho &={\setbox \z@
  \hbox {\frozen@everymath \@emptytoks \mathsurround \z@ $\nulldelimiterspace
  \z@ \left (\vcenter to\@ne \big@size {}\right .$}\box \z@ }(\protect \dag
  \circ \protect \mathcal {E}\circ \protect \dag )\star _{\protect \mathrm
  {L}}\rho {\setbox \z@ \hbox {\frozen@everymath \@emptytoks \mathsurround \z@
  $\nulldelimiterspace \z@ \left )\vcenter to\@ne \big@size {}\right .$}\box
  \z@ }^{\protect \dag } ={\setbox \z@ \hbox {\frozen@everymath \@emptytoks
  \mathsurround \z@ $\nulldelimiterspace \z@ \left (\vcenter to\@ne \big@size
  {}\right .$}\box \z@ }(\rho \otimes 1_{{{\protect \mathcal {B}}}})\protect
  \mathscr {D}[\protect \dag \circ \protect \mathcal {E}\circ \protect \dag
  ]{\setbox \z@ \hbox {\frozen@everymath \@emptytoks \mathsurround \z@
  $\nulldelimiterspace \z@ \left )\vcenter to\@ne \big@size {}\right .$}\box
  \z@ }^{\protect \dag }\\ &=\protect \mathscr {D}[\protect \dag \circ \protect
  \mathcal {E}\circ \protect \dag ]^{\protect \dag }(\rho \otimes 1_{{{\protect
  \mathcal {B}}}}) =\protect \mathscr {D}[\protect \mathcal {E}](\rho \otimes
  1_{{{\protect \mathcal {B}}}}) =\protect \mathcal {E}\star _{\protect \mathrm
  {R}}\rho , \end {split} \protect \] \endgroup where we used the second
  identity in Equation~\protect \textup {\hbox {\mathsurround \z@ \protect
  \normalfont (\ignorespaces \ref {eq:EgammaDE}\unskip \@@italiccorr )}} of
  Lemma~\ref {lem:EgamDE} in the second-last equality.}\BibitemShut {Stop}%
\bibitem [{\citenamefont {Hayashi}(2016)}]{Ha17}%
  \BibitemOpen
  \bibfield  {author} {\bibinfo {author} {\bibfnamefont {M.}~\bibnamefont
  {Hayashi}},\ }\href {\doibase 10.1007/978-3-662-49725-8} {\emph {\bibinfo
  {title} {Quantum information theory}}},\ \bibinfo {edition} {2nd}\ ed.\
  (\bibinfo  {publisher} {Springer Berlin, Heidelberg},\ \bibinfo {year}
  {2016})\BibitemShut {NoStop}%
\bibitem [{Note14()}]{Note14}%
  \BibitemOpen
  \bibinfo {note} {Axiom~(T\ref {item:T3}) is weaker than what appears in
  Refs.~\cite {Ha17,Ts22}, where the latter demanded $\Theta _{\rho \otimes
  \rho '}=\Theta _{\rho }\otimes \Theta _{\rho '}$. See the next footnote for
  the significance of this.}\BibitemShut {Stop}%
\bibitem [{Note15()}]{Note15}%
  \BibitemOpen
  \bibinfo {note} {As mentioned briefly in the previous footnote, our
  axiom~(T\ref {item:T3}) is weaker than in Refs.~\cite {Ha17,Ts22}. If we had
  used $\Theta _{\rho \otimes \rho '}=\Theta _{\rho }\otimes \Theta _{\rho '}$
  as in Refs.~\cite {Ha17,Ts22}, then the symmetric bloom and $(r,s)$ family,
  with $s\in (0,1)$, would \protect \emph {not} satisfy this axiom (however, it
  does hold if $s\in \protect \{0,1\protect \}$). It suffices to illustrate
  this in the case of the symmetric bloom. A counter-example can be obtained
  for ${{\protect \mathcal {A}}}=\protect \mathbb {M}_{2}={{\protect \mathcal
  {A}}}'$ by taking $\rho =\left [\begin {smallmatrix}p&0\\0&1-p\end
  {smallmatrix}\right ]$, $\rho '=\left [\begin {smallmatrix}p'&0\\0&1-p'\end
  {smallmatrix}\right ]$, $A=\left [\begin {smallmatrix}0&1\\1&0\end
  {smallmatrix}\right ]=A'$, and $p,p'\in [0,\protect \frac {1}{2})\cup
  (\protect \frac {1}{2},1]$. Indeed, $\Theta _{\rho \otimes \rho '}^{\protect
  \mathrm {J}}(A\otimes A')\not =\Theta _{\rho }^{\protect \mathrm
  {J}}(A)\otimes \Theta _{\rho '}^{\protect \mathrm {J}}(A')$.}\BibitemShut
  {Stop}%
\bibitem [{\citenamefont {Petz}(1988{\natexlab{b}})}]{Pe88b}%
  \BibitemOpen
  \bibfield  {author} {\bibinfo {author} {\bibfnamefont {D.}~\bibnamefont
  {Petz}},\ }in\ \href {\doibase 10.1007/BFb0078067} {\emph {\bibinfo
  {booktitle} {Quantum Probability and Applications III}}},\ \bibinfo {editor}
  {edited by\ \bibinfo {editor} {\bibfnamefont {L.}~\bibnamefont {Accardi}}\
  and\ \bibinfo {editor} {\bibfnamefont {W.}~\bibnamefont {von Waldenfels}}}\
  (\bibinfo  {publisher} {Springer},\ \bibinfo {address} {Berlin, Heidelberg},\
  \bibinfo {year} {1988})\ pp.\ \bibinfo {pages} {251--260}\BibitemShut
  {NoStop}%
\bibitem [{\citenamefont {Tsang}(2022{\natexlab{b}})}]{Ts22b}%
  \BibitemOpen
  \bibfield  {author} {\bibinfo {author} {\bibfnamefont {M.}~\bibnamefont
  {Tsang}},\ }\href {\doibase 10.48550/ARXIV.2212.13162} {\enquote {\bibinfo
  {title} {Operational meaning of a generalized conditional expectation in
  quantum metrology},}\ } (\bibinfo {year} {2022}{\natexlab{b}}),\ \Eprint
  {http://arxiv.org/abs/2212.13162} {arXiv:2212.13162 [quant-ph]} \BibitemShut
  {NoStop}%
\bibitem [{\citenamefont {{Cotler}}\ \emph {et~al.}(2018)\citenamefont
  {{Cotler}}, \citenamefont {{Jian}}, \citenamefont {{Qi}},\ and\ \citenamefont
  {{Wilczek}}}]{CJQW18}%
  \BibitemOpen
  \bibfield  {author} {\bibinfo {author} {\bibfnamefont {J.}~\bibnamefont
  {{Cotler}}}, \bibinfo {author} {\bibfnamefont {C.-M.}\ \bibnamefont
  {{Jian}}}, \bibinfo {author} {\bibfnamefont {X.-L.}\ \bibnamefont {{Qi}}}, \
  and\ \bibinfo {author} {\bibfnamefont {F.}~\bibnamefont {{Wilczek}}},\ }\href
  {\doibase 10.1007/JHEP09(2018)093} {\bibfield  {journal} {\bibinfo  {journal}
  {J. High Energ. Phys.}\ }\textbf {\bibinfo {volume} {2018}},\ \bibinfo
  {pages} {93} (\bibinfo {year} {2018})},\ \Eprint
  {http://arxiv.org/abs/1711.03119} {arXiv:1711.03119 [quant-ph]} \BibitemShut
  {NoStop}%
\bibitem [{\citenamefont {Cotler}\ \emph
  {et~al.}(2019{\natexlab{b}})\citenamefont {Cotler}, \citenamefont {Han},
  \citenamefont {Qi},\ and\ \citenamefont {Yang}}]{CHQY19}%
  \BibitemOpen
  \bibfield  {author} {\bibinfo {author} {\bibfnamefont {J.}~\bibnamefont
  {Cotler}}, \bibinfo {author} {\bibfnamefont {X.}~\bibnamefont {Han}},
  \bibinfo {author} {\bibfnamefont {X.-L.}\ \bibnamefont {Qi}}, \ and\ \bibinfo
  {author} {\bibfnamefont {Z.}~\bibnamefont {Yang}},\ }\href {\doibase
  10.1007/JHEP07(2019)042} {\bibfield  {journal} {\bibinfo  {journal} {J. High
  Energ. Phys.}\ }\textbf {\bibinfo {volume} {2019}} (\bibinfo {year}
  {2019}{\natexlab{b}}),\ 10.1007/JHEP07(2019)042},\ \Eprint
  {http://arxiv.org/abs/1811.05485} {arXiv:1811.05485 [hep-th]} \BibitemShut
  {NoStop}%
\bibitem [{\citenamefont {Aharonov}\ \emph {et~al.}(2009)\citenamefont
  {Aharonov}, \citenamefont {Popescu}, \citenamefont {Tollaksen},\ and\
  \citenamefont {Vaidman}}]{APTV09}%
  \BibitemOpen
  \bibfield  {author} {\bibinfo {author} {\bibfnamefont {Y.}~\bibnamefont
  {Aharonov}}, \bibinfo {author} {\bibfnamefont {S.}~\bibnamefont {Popescu}},
  \bibinfo {author} {\bibfnamefont {J.}~\bibnamefont {Tollaksen}}, \ and\
  \bibinfo {author} {\bibfnamefont {L.}~\bibnamefont {Vaidman}},\ }\href
  {\doibase 10.1103/PhysRevA.79.052110} {\bibfield  {journal} {\bibinfo
  {journal} {Phys. Rev. A}\ }\textbf {\bibinfo {volume} {79}},\ \bibinfo
  {pages} {052110} (\bibinfo {year} {2009})},\ \Eprint
  {http://arxiv.org/abs/0712.0320} {arXiv:0712.0320 [quant-ph]} \BibitemShut
  {NoStop}%
\bibitem [{\citenamefont {Huang}\ and\ \citenamefont {Guo}(2022)}]{HuGu22}%
  \BibitemOpen
  \bibfield  {author} {\bibinfo {author} {\bibfnamefont {Z.}~\bibnamefont
  {Huang}}\ and\ \bibinfo {author} {\bibfnamefont {X.-K.}\ \bibnamefont
  {Guo}},\ }\href {\doibase 10.48550/ARXIV.2211.13396} {\enquote {\bibinfo
  {title} {{Leggett--Garg inequalities for multitime processe}},}\ } (\bibinfo
  {year} {2022}),\ \Eprint {http://arxiv.org/abs/2211.13396} {arXiv:2211.13396
  [quant-ph]} \BibitemShut {NoStop}%
\bibitem [{\citenamefont {Witten}(2018)}]{Wi18}%
  \BibitemOpen
  \bibfield  {author} {\bibinfo {author} {\bibfnamefont {E.}~\bibnamefont
  {Witten}},\ }\href {\doibase 10.1103/RevModPhys.90.045003} {\bibfield
  {journal} {\bibinfo  {journal} {Rev. Mod. Phys.}\ }\textbf {\bibinfo {volume}
  {90}},\ \bibinfo {pages} {045003} (\bibinfo {year} {2018})}\BibitemShut
  {NoStop}%
\bibitem [{\citenamefont {Mollabashi}\ \emph
  {et~al.}(2021{\natexlab{a}})\citenamefont {Mollabashi}, \citenamefont
  {Shiba}, \citenamefont {Takayanagi}, \citenamefont {Tamaoka},\ and\
  \citenamefont {Wei}}]{MSTTW21}%
  \BibitemOpen
  \bibfield  {author} {\bibinfo {author} {\bibfnamefont {A.}~\bibnamefont
  {Mollabashi}}, \bibinfo {author} {\bibfnamefont {N.}~\bibnamefont {Shiba}},
  \bibinfo {author} {\bibfnamefont {T.}~\bibnamefont {Takayanagi}}, \bibinfo
  {author} {\bibfnamefont {K.}~\bibnamefont {Tamaoka}}, \ and\ \bibinfo
  {author} {\bibfnamefont {Z.}~\bibnamefont {Wei}},\ }\href {\doibase
  10.1103/PhysRevLett.126.081601} {\bibfield  {journal} {\bibinfo  {journal}
  {Phys. Rev. Lett.}\ }\textbf {\bibinfo {volume} {126}},\ \bibinfo {pages}
  {081601} (\bibinfo {year} {2021}{\natexlab{a}})},\ \Eprint
  {http://arxiv.org/abs/2011.09648} {arXiv:2011.09648 [hep-th]} \BibitemShut
  {NoStop}%
\bibitem [{\citenamefont {Mollabashi}\ \emph
  {et~al.}(2021{\natexlab{b}})\citenamefont {Mollabashi}, \citenamefont
  {Shiba}, \citenamefont {Takayanagi}, \citenamefont {Tamaoka},\ and\
  \citenamefont {Wei}}]{MSTTW21b}%
  \BibitemOpen
  \bibfield  {author} {\bibinfo {author} {\bibfnamefont {A.}~\bibnamefont
  {Mollabashi}}, \bibinfo {author} {\bibfnamefont {N.}~\bibnamefont {Shiba}},
  \bibinfo {author} {\bibfnamefont {T.}~\bibnamefont {Takayanagi}}, \bibinfo
  {author} {\bibfnamefont {K.}~\bibnamefont {Tamaoka}}, \ and\ \bibinfo
  {author} {\bibfnamefont {Z.}~\bibnamefont {Wei}},\ }\href {\doibase
  10.1103/PhysRevResearch.3.033254} {\bibfield  {journal} {\bibinfo  {journal}
  {Phys. Rev. Res.}\ }\textbf {\bibinfo {volume} {3}},\ \bibinfo {pages}
  {033254} (\bibinfo {year} {2021}{\natexlab{b}})},\ \Eprint
  {http://arxiv.org/abs/2106.03118} {arXiv:2106.03118 [hep-th]} \BibitemShut
  {NoStop}%
\bibitem [{\citenamefont {Doi}\ \emph {et~al.}(2023)\citenamefont {Doi},
  \citenamefont {Harper}, \citenamefont {Mollabashi}, \citenamefont
  {Takayanagi},\ and\ \citenamefont {Taki}}]{DHMTT22}%
  \BibitemOpen
  \bibfield  {author} {\bibinfo {author} {\bibfnamefont {K.}~\bibnamefont
  {Doi}}, \bibinfo {author} {\bibfnamefont {J.}~\bibnamefont {Harper}},
  \bibinfo {author} {\bibfnamefont {A.}~\bibnamefont {Mollabashi}}, \bibinfo
  {author} {\bibfnamefont {T.}~\bibnamefont {Takayanagi}}, \ and\ \bibinfo
  {author} {\bibfnamefont {Y.}~\bibnamefont {Taki}},\ }\href {\doibase
  10.1103/PhysRevLett.130.031601} {\bibfield  {journal} {\bibinfo  {journal}
  {Phys. Rev. Lett.}\ }\textbf {\bibinfo {volume} {130}},\ \bibinfo {pages}
  {031601} (\bibinfo {year} {2023})},\ \Eprint
  {http://arxiv.org/abs/2210.09457} {arXiv:2210.09457 [hep-th]} \BibitemShut
  {NoStop}%
\bibitem [{\citenamefont {Leggett}\ and\ \citenamefont {Garg}(1985)}]{LeGa85}%
  \BibitemOpen
  \bibfield  {author} {\bibinfo {author} {\bibfnamefont {A.~J.}\ \bibnamefont
  {Leggett}}\ and\ \bibinfo {author} {\bibfnamefont {A.}~\bibnamefont {Garg}},\
  }\href {\doibase 10.1103/PhysRevLett.54.857} {\bibfield  {journal} {\bibinfo
  {journal} {Phys. Rev. Lett.}\ }\textbf {\bibinfo {volume} {54}},\ \bibinfo
  {pages} {857} (\bibinfo {year} {1985})}\BibitemShut {NoStop}%
\bibitem [{\citenamefont {Schumacher}\ and\ \citenamefont
  {Nielsen}(1996)}]{NiSh96}%
  \BibitemOpen
  \bibfield  {author} {\bibinfo {author} {\bibfnamefont {B.}~\bibnamefont
  {Schumacher}}\ and\ \bibinfo {author} {\bibfnamefont {M.~A.}\ \bibnamefont
  {Nielsen}},\ }\href {\doibase 10.1103/PhysRevA.54.2629} {\bibfield  {journal}
  {\bibinfo  {journal} {Phys. Rev. A}\ }\textbf {\bibinfo {volume} {54}},\
  \bibinfo {pages} {2629} (\bibinfo {year} {1996})}\BibitemShut {NoStop}%
\bibitem [{\citenamefont {Hawking}(1975)}]{Ha75}%
  \BibitemOpen
  \bibfield  {author} {\bibinfo {author} {\bibfnamefont {S.~W.}\ \bibnamefont
  {Hawking}},\ }\href {http://projecteuclid.org/euclid.cmp/1103899181}
  {\bibfield  {journal} {\bibinfo  {journal} {Comm. Math. Phys.}\ }\textbf
  {\bibinfo {volume} {43}},\ \bibinfo {pages} {199} (\bibinfo {year}
  {1975})}\BibitemShut {NoStop}%
\bibitem [{\citenamefont {Hawking}(1976)}]{Ha76}%
  \BibitemOpen
  \bibfield  {author} {\bibinfo {author} {\bibfnamefont {S.~W.}\ \bibnamefont
  {Hawking}},\ }\href {\doibase 10.1103/PhysRevD.14.2460} {\bibfield  {journal}
  {\bibinfo  {journal} {Phys. Rev. D}\ }\textbf {\bibinfo {volume} {14}},\
  \bibinfo {pages} {2460} (\bibinfo {year} {1976})}\BibitemShut {NoStop}%
\bibitem [{\citenamefont {Hawking}(1982)}]{Ha82}%
  \BibitemOpen
  \bibfield  {author} {\bibinfo {author} {\bibfnamefont {S.~W.}\ \bibnamefont
  {Hawking}},\ }\href {https://projecteuclid.org:443/euclid.cmp/1103922050}
  {\bibfield  {journal} {\bibinfo  {journal} {Comm. Math. Phys.}\ }\textbf
  {\bibinfo {volume} {87}},\ \bibinfo {pages} {395} (\bibinfo {year}
  {1982})}\BibitemShut {NoStop}%
\bibitem [{Note16()}]{Note16}%
  \BibitemOpen
  \bibinfo {note} {The Hilbert--Schmidt adjoint $T^*$ of a conjugate-linear map
  $T:\protect \mathcal {H}\to \protect \mathcal {K}$ between Hilbert spaces
  $\protect \mathcal {H}$ and $\protect \mathcal {K}$ is defined by
  $\<v,T^*w\>_{\mathcal {H}}=\<w,Tv\>_{\mathcal {K}}$ for all $v\in \protect
  \mathcal {H}$ and $w\in \protect \mathcal {K}$. In this case, the
  conjugate-linear map in question is $\protect \dag $, which is self-adjoint
  with respect to this definition.}\BibitemShut {Stop}%
\bibitem [{\citenamefont {de~Pillis}(1967)}]{dePi67}%
  \BibitemOpen
  \bibfield  {author} {\bibinfo {author} {\bibfnamefont {J.}~\bibnamefont
  {de~Pillis}},\ }\href {\doibase 10.2140/pjm.1967.23.129} {\bibfield
  {journal} {\bibinfo  {journal} {Pac. J. Math.}\ }\textbf {\bibinfo {volume}
  {23}},\ \bibinfo {pages} {129} (\bibinfo {year} {1967})}\BibitemShut
  {NoStop}%
\bibitem [{\citenamefont {Hill}(1973)}]{Hi73}%
  \BibitemOpen
  \bibfield  {author} {\bibinfo {author} {\bibfnamefont {R.~D.}\ \bibnamefont
  {Hill}},\ }\href {\doibase https://doi.org/10.1016/0024-3795(73)90026-8}
  {\bibfield  {journal} {\bibinfo  {journal} {Linear Algebra Appl.}\ }\textbf
  {\bibinfo {volume} {6}},\ \bibinfo {pages} {257} (\bibinfo {year}
  {1973})}\BibitemShut {NoStop}%
\bibitem [{Note17()}]{Note17}%
  \BibitemOpen
  \bibinfo {note} {As with some earlier proofs, the sequence of calculations
  from this proof are much more easily visualized using the string diagrams of
  quantum Markov categories~\cite {PaBayes}.}\BibitemShut {Stop}%
\end{thebibliography}%

\end{document}